%% file: main.tex
\documentclass[10pt,journal,compsoc]{IEEEtran}
\input{packages}
\newcolumntype{P}[1]{>{\centering\arraybackslash}p{#1}}
\newcolumntype{M}[1]{>{\centering\arraybackslash}m{#1}}
\newtheorem{theorem}{Theorem}
\theoremstyle{definition}
\newtheorem{definition}{Definition}[section]
\newtheorem{corollary}{Corollary}[theorem]

\ifCLASSOPTIONcompsoc
  \usepackage[nocompress]{cite}
\else
  \usepackage{cite}
\fi

\ifCLASSINFOpdf
\else
\fi

\begin{document}
%
\title{Robust Regularized Locality Preserving Indexing for Fiedler Vector Estimation}

\author{Aylin~Ta{\c{s}}tan,~\IEEEmembership{Student Member,~IEEE,}
        Michael~Muma,~\IEEEmembership{Member,~IEEE,}
        and~Abdelhak~M.~Zoubir,~\IEEEmembership{Fellow,~IEEE}
\thanks{The authors are with the Signal Processing Group, Technische Universität
Darmstadt, Darmstadt, Germany (e-mail: atastan@spg.tu-darmstadt.de; muma@spg.tu-darmstadt.de; zoubir@spg.tu-darmstadt.de).\protect\\
}}

\markboth{}%
{Shell \MakeLowercase{\textit{et al.}}: Robust Regularized Locality Preserving Indexing for Fiedler Vector Estimation}

\IEEEtitleabstractindextext{%
\begin{abstract}
The Fiedler vector of a connected graph is the eigenvector associated with the algebraic connectivity of the graph Laplacian and it provides substantial information to learn the latent structure of a graph. In real-world applications, however, the data may be subject to heavy-tailed noise and outliers which results in deteriorations in the structure of the Fiedler vector estimate. We design a Robust Regularized 
Locality Preserving Indexing (RRLPI) method for Fiedler vector estimation that aims to approximate the nonlinear manifold structure of 
the Laplace Beltrami operator while minimizing the negative impact of outliers. First, an analysis of the effects of two fundamental outlier types on the eigen-decomposition for block affinity matrices which are essential in cluster analysis is conducted. Then, an error model is formulated and a robust Fiedler vector estimation algorithm is developed. An unsupervised penalty parameter selection algorithm is proposed that leverages the geometric structure of the projection space 
to perform robust regularized Fiedler estimation. The performance of RRLPI is benchmarked against existing competitors in terms of detection probability, partitioning quality, image segmentation capability, robustness and computation time using a large variety of synthetic and real data experiments.

\end{abstract}

}

\maketitle

\IEEEdisplaynontitleabstractindextext
\IEEEpeerreviewmaketitle

\ifCLASSOPTIONcompsoc
\IEEEraisesectionheading{\section{Introduction}\label{sec:introduction}}
\else
\section{Introduction}
\label{sec:introduction}
\fi

\IEEEPARstart{T}{he} Fiedler vector of a connected graph is the eigenvector associated with the second smallest eigenvalue, the so called Fiedler value, of the graph Laplacian matrix. The Fiedler vector and the Fiedler value provide important information for estimating  \citeform[1]-\citeform[3] and controlling \citeform[4]-\citeform[6] the algebraic connectivity of a graph, finding densely connected groups of vertices that are hidden in the graph structure \citeform[7]-\citeform[10], and representing the implicit relationships between variables in a low-dimensional space \citeform[11]-\citeform[12]. More generally, eigenvector decomposition is used in a variety of applications to achieve tasks, such as dimension reduction \citeform[13]-\citeform[18], recognition \citeform[19]-\citeform[21], clustering/classification \citeform[22]-\citeform[27] and localization \cite{Localapp1}.

Due to its central role in graph analysis, the estimation of the Fiedler vector has been a fundamental research area for decades \citeform[12]-\citeform[20]. One of the most popular approaches to estimate the Fiedler vector is known as Laplacian eigenmaps (LE) \cite{Laplacianeigenmaps}. LE performs a nonlinear dimensionality reduction while preserving the local neighborhood information in a certain sense, and it explicitly reveals the manifold structure \cite{Laplacianeigenmapsdatarep}. An alternative approach is that of locality preserving indexing (LPI) \cite{LPI}, which transforms the nonlinear dimensionality reduction in the Laplace Beltrami operator into a linear system of equations. LPI requires a complete singular value decomposition (SVD), resulting in a considerable computational complexity which is why computationally more attractive alternative approaches have been proposed in \cite{RLPI}, \cite{FSCEFM}.  However, the performance of \cite{RLPI} strongly depends on the penalty parameter selection. Further, in real-world scenarios, outliers and heavy-tailed noise may obscure the graph structure that represents the clean data. Consequentially, the Fiedler vector estimates are corrupted, and the feature mapping no longer provides useful information about the dataset if it is dominated by outliers. Therefore, robust methods are needed. One popular strategy to estimate eigenvectors robustly mitigates the effects of outliers in the representation space by restructuring the affinity matrix, e.g. \citeform[31]-\citeform[33]. Alternatively, the outliers effects can be suppressed in the projection space, such as in \cite{robustproject}. However, the robust projection operation in \cite{robustproject}, uses the $\ell_1$ norm that creates a different eigenbasis and requires prior information about the data, i.e. representative samples. In \cite{RLPFM}, a robust locality preserving feature mapping (RLPFM) is proposed that builds upon M-estimation \cite{Robuststatistics} of the Fiedler vector. In contrast to \cite{robustproject}, the $\ell_2$ norm is used in the projection operation, just as in the classical eigen-decomposition used in spectral clustering \cite{Spectralclustering}.

In the sequel, we propose a new Robust Regularized Locality Preserving Indexing (RRLPI) method for Fiedler vector estimation. The key idea is to transform the Fiedler vector estimation problem into a linear system of equations that reveals the hidden group structure in a given graph. To understand how to best integrate robustness into Fiedler vector estimation, we begin by comprehensively analyzing the outlier effects on the eigen-decomposition. 
Based on our analysis, it is shown that the overall edge weight attached to a vertex is a valuable information to identify an outlier. Therefore, an error model is formulated based on the typical overall edge weight of a graph. In contrast to RLPFM \cite{RLPFM}, which performs M-estimation of the eigenvectors by iteratively reweighting the residuals of a Laplacian eigenmaps-based prediction, RRLPI robustly estimates the Fiedler vector based on the typical overall edge weight of the graph. As in RLPFM, the penalty parameter is optimized based on $\Delta$-separated sets \cite{arora2009expander}.

The paper is organized as follows. Section~\ref{sec:Preliminaries} introduces the basic concepts and briefly discusses Fiedler vector estimation using LPI. The ideas underlying the proposed algorithm and the problem formulation are the subject of Section~\ref{sec:Motivation}.
Section 4 is dedicated to the proposed robust regularized locality preserving indexing method. This section includes the theoretical analysis, penalty parameter selection, computational complexity analysis and possible applications. Section~\ref{sec:Experimental} demonstrates the performance of the proposed approach in comparison to popular competitors both in cluster enumeration and in image segmentation using real-world data sets. Finally, conclusions are drawn in Section~\ref{sec:Conclusion}.

\section{Preliminaries}\label{sec:Preliminaries}
\subsection{Summary of Notations}\label{sec:notations}
Lower and upper-case bold letters denote vectors and matrices, respectively; $\mathbb{R}$ denotes the set of real numbers; $\mathbb{Z}^{+}$ denotes the set of positive integers; $|x|$ denotes the absolute value of $x$; $\|\mathbf{x}\|$ denotes the norm of vector $\mathbf{x}$, e.g. $\|\mathbf{x}\|_2$ is the $\ell_2$ norm; $\mathrm{med}(\mathbf{x})$ denotes the median of vector $\mathbf{x}$; $\mathrm{sign}(x)$ denotes the sign function defined as ${\mathrm{sign}(x)={x}/{|x|}}$; $\mathrm{diag}(x_1,\dots,x_n)$ denotes a diagonal matrix of size $n\times n$ with $x_1,\dots,x_n$ on its diagonal; $\mathbf{x}^{\top}$ denotes transpose of vector $\mathbf{x}$; $\mathbf{I}$ denotes the identity matrix; $\mathbf{1}$ denotes the vector of ones; $\hat{\mathbf{x}}$ denotes the estimate of vector $\mathbf{x}$; $\Tilde{\mathbf{W}}$ refers to a corrupted affinity matrix; $d_i$ denotes the overall edge weight of the $i$th feature vector for a weighted affinity matrix and the corresponding degree for an adjacency matrix; $\lambda_{i}$ denotes the $i$th eigenvalue; $y_{j,i}$ denotes the embedding result of the $j$th feature vector in the eigenvector $\mathbf{y}_i$ associated with the $i$th eigenvalue $\lambda_{i}$.

\subsection{LPI for Fiedler Vector Estimation}\label{sec:LPI}
Suppose that a data set ${\mathbf{X}=[\mathbf{x}_1,\dots,\mathbf{x}_n]\in\mathbb{R}^{m\times n}}$ with $m$ denoting the feature dimension and $n$ being the number of feature vectors, can be represented as a graph ${G=\{V,E,\mathbf{W}\}}$, where $V$ denotes the vertices, $E$ represents the edges, and $\mathbf{W}\in\mathbb{R}^{n\times n}$ is the symmetric affinity matrix that is computed using a similarity measure, e.g. cosine similarity. Let $\mathbf{L}\in\mathbb{R}^{n\times n}$ denote the discrete Laplacian operator (also known as the Laplacian matrix) that is nonnegative definite with associated eigenvalues ${0\leq\lambda_0\leq \lambda_1\leq\dots\leq\lambda_{n-1}}$ sorted in ascending order. Then, it follows that the Fiedler vector $\mathbf{y}_F\in\mathbb{R}^{n}$ is the eigenvector associated with the second smallest eigenvalue $\lambda_1$ of the eigen-problem
\begin{equation}\label{eq:eigenproblem_LaplacianD1}
\mathbf{L}\mathbf{y}_i=\lambda_i\mathbf{y}_i, 
\end{equation}
or in a generalized eigenvalue problem form
\begin{equation}\label{eq:eigenproblem_Laplacian}
       \mathbf{L}\mathbf{y}_i=\lambda_i\mathbf{D}\mathbf{y}_i,
\end{equation}
Here, the Laplacian matrix $\mathbf{L}$ is defined analogously to the Laplace Beltrami operator on the manifold by ${\mathbf{L}=\mathbf{D}-\mathbf{W}}$, where ${\mathbf{D}\in\mathbb{R}^{n\times n}}$ is a diagonal weight matrix with overall edge weights ${d_{i,i}=\sum_{j}w_{i,j}}$ on the diagonal,  ${\lambda_i\in\{\lambda_0,\lambda_1\dots\lambda_{n-1}\}}$ denotes the $i$th eigenvalue and ${\mathbf{y}_i\in\mathbb{R}^{n}}$ is the eigenvector associated with $\lambda_i$.

The LPI method determines linear approximations to the eigenfunctions of the Laplace Beltrami operator \cite{LPP}, \cite{LPI} by representing the Fiedler vector as the response of a linear regression with input variables $\mathbf{X}$, i.e., $\mathbf{y}_F=\mathbf{X}^{\top}\boldsymbol{\beta}_F$. Hence, the LPI finds a transformation vector $\boldsymbol{\beta}_F\in\mathbb{R}^{m}$ that is the eigenvector associated with the second smallest eigenvalue of the generalized eigen-problem
\begin{equation}
        \mathbf{X}\mathbf{L}\mathbf{X}^{\top}\boldsymbol{\beta}_i=\lambda_i\mathbf{X}\mathbf{D}\mathbf{X}^{\top}\boldsymbol{\beta}_i,\hspace{2mm}i=0,\dots,n-1
\end{equation}
or, equivalently, it is associated with the second largest eigenvalue of the generalized eigen-problem
\begin{equation}\label{eq:eigenproblem_proposed}     
        \mathbf{X}\mathbf{W}\mathbf{X}^{\top}\boldsymbol{\beta}_j=\lambda_j\mathbf{X}\mathbf{D}\mathbf{X}^{\top}\boldsymbol{\beta}_j,
\end{equation}
which has the same eigenvalue $\lambda_i$ as in Eq.~\eqref{eq:eigenproblem_Laplacian} for $\lambda_j$ with ${j=n-(i+1)}$. In the following theorem, the link between the well-known Fiedler vector estimation as in \cite{Laplacianeigenmaps} and LPI \cite{LPI} is shown.

\begin{theorem}
Let $\mathbf{y}_F$ be the Fiedler vector associated with the second largest eigenvalue $\lambda_F$ such that $F=n-2$ of the eigen-problem
\begin{equation}\label{eq:generaleigenproblem} 
    \mathbf{W}\mathbf{y}_F=\lambda_F\mathbf{D}\mathbf{y}_F.
\end{equation} If $\mathbf{X}^{\top}\boldsymbol{\beta}_F=\mathbf{y}_F$, then $\boldsymbol{\beta}_F$ is the eigenvector of the eigen-problem
in Eq.~\eqref{eq:eigenproblem_proposed} with the same eigenvalue $\lambda_i$ such that $i=1$.
\end{theorem}

\begin{proof}
Replacing $\mathbf{X}^{\top}\boldsymbol{\beta}_F$ by the Fiedler vector $\mathbf{y}_F$ on the left side of Eq.~\eqref{eq:eigenproblem_proposed} yields
\begin{align*}
\begin{split}
\mathbf{X}\mathbf{W}\mathbf{X}^{\top}\boldsymbol{\beta}_F =&\mathbf{X}\mathbf{W}\mathbf{y}_F =\mathbf{X} \lambda_F\mathbf{D}\mathbf{y}_F = \lambda_F\mathbf{X}\mathbf{D}\mathbf{y}_F\\ =&\lambda_F\mathbf{X}\mathbf{D}\mathbf{X}^{\top}\boldsymbol{\beta}_F \hspace{2mm}\mathrm{s.t.} \hspace{2mm} F=n-2
\end{split}
\end{align*}
and shows that for $F=j=n-2$, $\boldsymbol{\beta}_F$ is the eigenvector of the eigen-problem of Eq.~\eqref{eq:eigenproblem_proposed} which concludes the proof.
\vspace{-1mm}
\end{proof}
Therefore, building upon \cite{RLPI}, the projective functions of LPI can be determined in two consecutive steps for Fiedler vector estimation. First, the Fiedler vector $\mathbf{y}_F$ associated with the second smallest eigenvalue of Eq.~\eqref{eq:eigenproblem_Laplacian} must be computed. Then, for the Fiedler vector $\mathbf{y}_F$, the LPI method estimates a transformation vector $\boldsymbol{\beta}_F\in\mathbb{R}^{m}$ that satisfies $\mathbf{y}_F=\mathbf{X}^{\top}\boldsymbol{\beta}_F$ by solving the following least squares problem
\begin{equation}
    \hat{\boldsymbol{\beta}}_F=\underset{\boldsymbol{\beta}}{\mathrm{argmin}}\sum_{i=1}^n(\boldsymbol{\beta}_F^{\top}\mathbf{x}_i-y_{F,i})^2,
\end{equation}
where $y_{F,i}$ is the $i$th mapping point in $\mathbf{y}_F$ and $\hat{\boldsymbol{\beta}}_F$ is the estimated transformation vector.
\vspace{-3.5mm}
\section{Motivation and Problem Statement}\label{sec:Motivation}
The previous section discussed the applicability of LPI for Fiedler vector estimation. In particular, LPI may discover the hidden nonlinear structure by finding linear approximations to the nonlinear Laplacian eigenmaps (for details, see \cite{LPI} and \cite{LPP}). However, when using the least-squares objective function, outliers and heavy-tailed noise may have a large impact  on the estimation of the transformation vector $\boldsymbol{\beta}_F$. This leads to errors in the estimated Fiedler vector and, consequently, an information loss about the representation of the underlying graph structure using such a corrupted Fiedler vector estimate. This section analyzes the effects of outliers and noise on the eigen-decomposition of the Laplacian matrix. The analysis provides the theoretical basis and an understanding of the ideas underlying the proposed robust Fiedler vector estimation approach.

\subsection{Outlier Effect on Eigen-decomposition}\label{sec:outliereffec}
The effect of outliers on eigen-decomposition is analyzed in terms of two fundamental types of outliers. All examinations that are conducted in this section are made for block affinity matrices unless it is stated otherwise. Motivated by \cite{sparsesubspace}, we begin by defining the first fundamental type of outliers as follows.

\begin{definition}\label{def:outliertype1}
\textbf{(Type I Outliers)}
\textit{The feature vectors that do not have considerable correlations with any of the samples are called Type I outliers.
}
\end{definition}
Based on this first definition, the correlation coefficient vectors that are associated to the outliers have negligibly small values or, ideally, are zero vectors. In the following theorem, the effect of such outliers on eigen-decomposition is analyzed.
 \begin{figure}[!tbp]
\begin{minipage}[b]{1.0\linewidth}
  \centering
  \includegraphics[trim={0 0mm 0 0mm},clip,width=9cm]{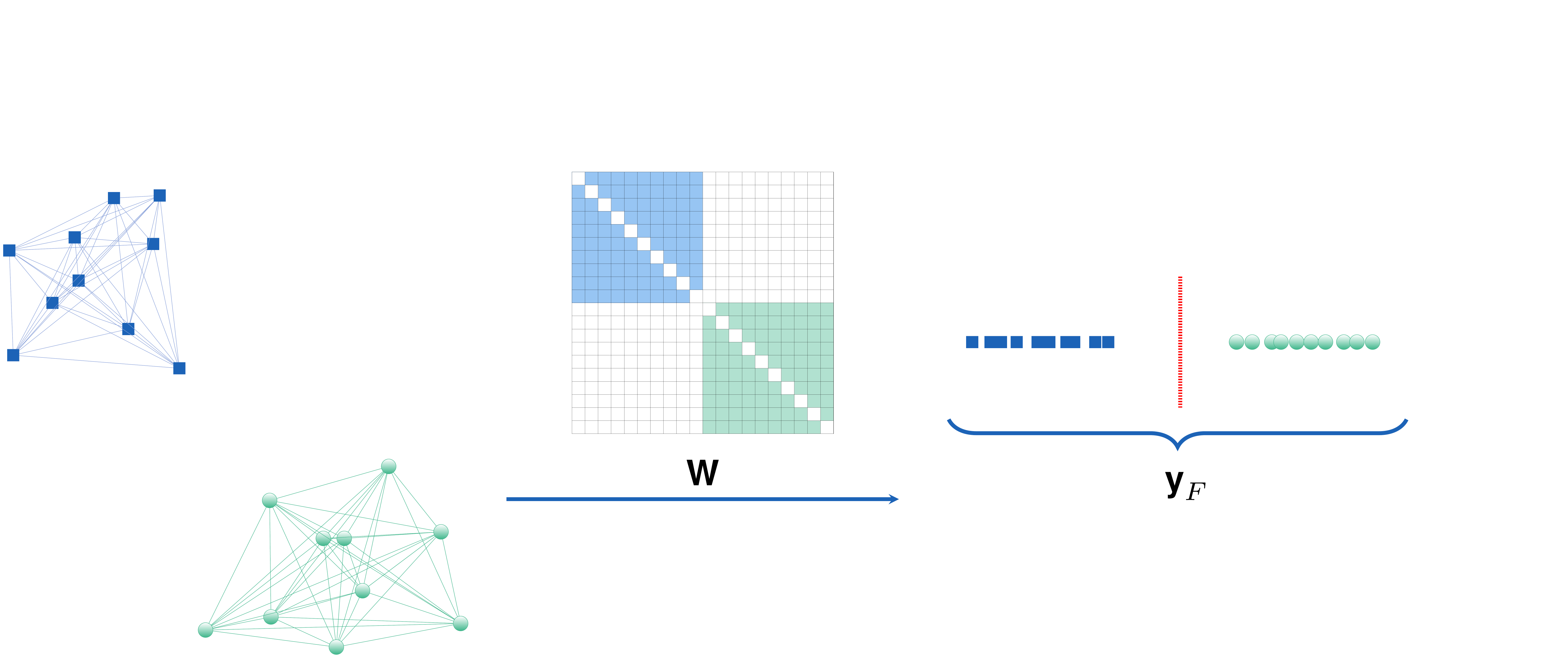}
\end{minipage}
\caption{Fiedler vector estimation for an ideal $k=2$ blocks affinity matrix.}
\label{fig:scenario1}
\end{figure}

\begin{theorem}\label{app:theorem2}
For a block diagonal symmetric affinity matrix ${\mathbf{W}\in\mathbb{R}^{n\times n}}$, the associated Laplacian ${\mathbf{L}\in\mathbb{R}^{n\times n}}$, and eigenvalues ${0\leq\lambda_0\leq \lambda_1\leq\dots\leq\lambda_{n-1}}$, the outlying vertices which have negligibly small correlation coefficients result in eigenvalues that decrease to zero.
\end{theorem}
\begin{proof}
See Appendix A.1 of the supplementary material.
\end{proof}
According to Theorem~\ref{app:theorem2}, if the data matrix $\mathbf{X}$ includes outliers that follow Definition \ref{def:outliertype1}, it will have more than $k$ eigenvalues that are close to zero for a $k$ block diagonal affinity matrix, which is a valuable information to detect the number of outliers. In the following corollary, Theorem~\ref{app:theorem2} is extended for general affinity matrices.

\begin{corollary}
\textit{For a symmetric affinity matrix ${\mathbf{W}\in\mathbb{R}^{n\times n}}$, and the associated Laplacian ${\mathbf{L}\in\mathbb{R}^{n\times n}}$, for eigenvalues ${0\leq\lambda_0\leq \lambda_1\leq\dots\leq\lambda_{n-1}}$, the outlying vertices which have negligibly small correlation coefficients result in eigenvalues that decrease to zero.}
\end{corollary}
\begin{proof}
See Appendix A.2 of the supplementary material.
\end{proof}
The eigenvectors of a Laplacian have a large variety of applications, such as \citeform[13]-\citeform[28], and, in particular, they play a crucial role in cluster analysis. In what follows, the effect of outliers on group structures is analyzed for both eigen-decompositions, i.e Eq.~\eqref{eq:eigenproblem_LaplacianD1} and Eq.~\eqref{eq:eigenproblem_Laplacian}.

\begin{corollary}\label{cor:type1outliereffecteigenvector}
\textit{For a definite nonnegative $k$ block zero-diagonal symmetric affinity matrix ${\mathbf{W}\in\mathbb{R}^{n\times n}}$ and the associated Laplacian ${\mathbf{L}\in\mathbb{R}^{n\times n}}$, let the eigenvectors have a norm, such that $\|\mathbf{y}^{(i)}\|^2_2=1$ holds, where $\mathbf{y}^{(i)}$ denotes the eigenvector associated with the $i$th eigenvalue. The outlying vertices which have negligibly small correlation coefficients lead to a decrease in the $\ell_2$ norm between the mapping points that are associated to different clusters based on the first $k$ eigenvectors. When the absolute value of the embedding results of the outliers increases to one, i.e. when $|y_o^{(i)}|\rightarrow 1$, the $\ell_2$ norm decreases to zero.}
\end{corollary}
\begin{proof}
See Appendix A.3 of the supplementary material.
\end{proof}
As a result of Corollary~\ref{cor:type1outliereffecteigenvector}, as $|y_o^{(i)}|\rightarrow 1$,  all non-outlying points are mapped to the same cluster, when using distance-based approaches, such as, spectral clustering. This explains why spectral clustering breaks down in the presence of Type I outliers. 

Next, we will study the effect of a second fundamental type of outliers that are defined as follows.
 \begin{figure}[!tbp]
\begin{minipage}[b]{1.0\linewidth}
  \centering
  \includegraphics[trim={0 0mm 0 0mm},clip,width=9cm]{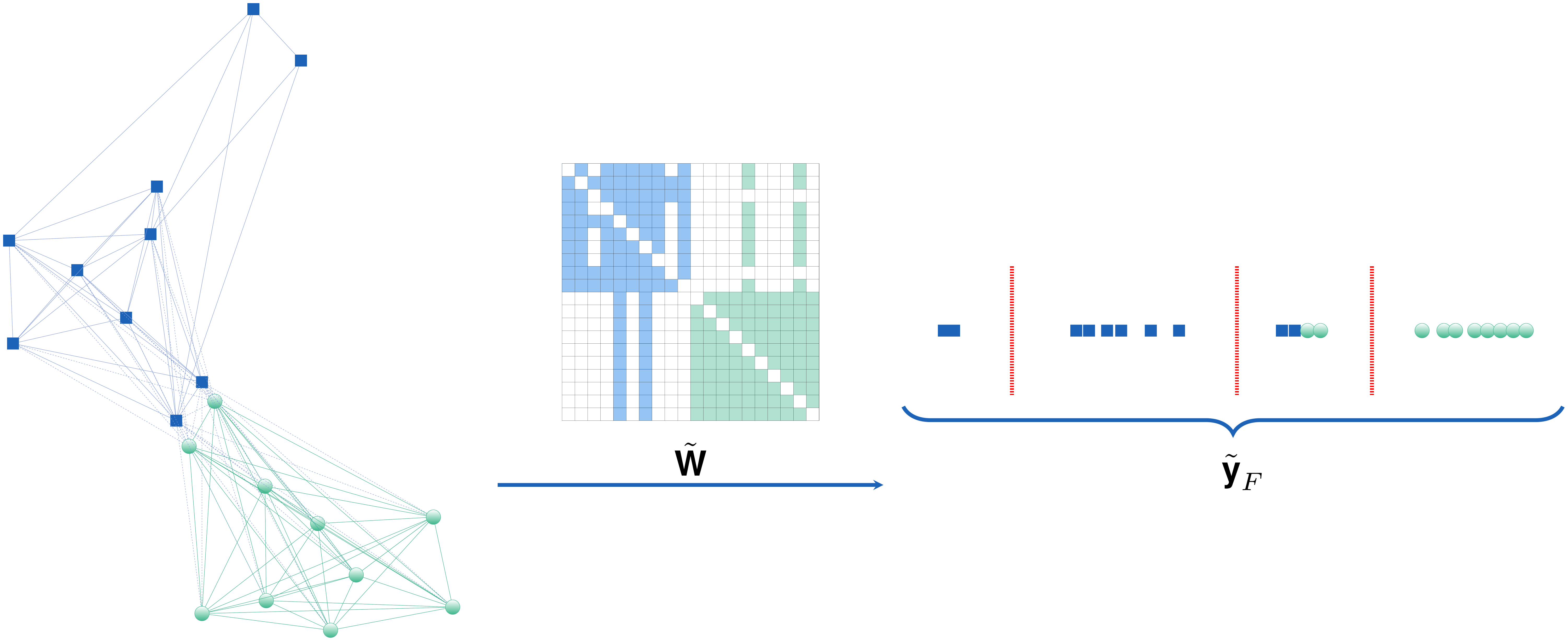}
\end{minipage}
\caption{Fiedler vector estimation for a corrupted $k=2$ blocks affinity matrix.}
\label{fig:scenario2}
\end{figure}

\begin{definition}
\textbf{(Type II Outliers)}
\textit{Outliers of Type II are the feature vectors that have considerable correlations with more than one group of feature vectors.}
\end{definition}

As for the outliers that were defined in Definition~\ref{def:outliertype1}, the effect of such outliers on eigen-decomposition is analyzed for block affinity matrices using the eigen-decompositions as in Eq.~\eqref{eq:eigenproblem_LaplacianD1} and Eq.~\eqref{eq:eigenproblem_Laplacian}. This leads to the following theorem.

\begin{theorem}\label{app:theorem3}
For a $k$ block zero-diagonal symmetric nonnegative affinity matrix $\mathbf{W}\in\mathbb{R}^{n\times n}$, let  $w_i\in \{ w_1,w_2,\dots,w_k\}$ denote a constant around which the correlation coefficients of the $i$th block are assumed to be concentrated with negligibly small variations. Further, let $\Tilde{w}_u$ denote a constant around which the correlations between blocks are concentrated. Let  $\Tilde{\mathbf{W}}$ define an affinity matrix, which is equal to $\mathbf{W}$, except that we impose  $\Tilde{w}_u>0$, such that the $i$th and $j$th block become correlated. Then, it follows that the second smallest eigenvalue associated with the Laplacian matrix of the correlated $i$th and $j$th blocks $\Tilde{\mathbf{L}}_{ij}$ will be greater than zero, while they were equal to zero for $\mathbf{L}_{ij}$.
\end{theorem}
\begin{proof}
See Appendix A.4 of the supplementary material.
\end{proof}

Based on Theorem~\ref{app:theorem3}, if any blocks are correlated with each other, less than $k$ eigenvalues will be equal to zero. In addition to their effect on the eigenvalue structure, it is important to show the effect of outliers on the eigenvectors which are the basis of clustering approaches. As an example, spectral clustering designs a matrix ${\mathbf{Y}\in\mathbb{R}^{n\times K}}$ whose column vectors are the eigenvectors of Eq.~\eqref{eq:eigenproblem_Laplacian} and then clusters the row vectors of that matrix based on $K$-means partitioning for a number $K$ of clusters. Therefore, in the sequel, the effects of outliers are examined from the perspective of spectral clustering where the group structure plays a crucial role.

\begin{theorem}\label{app:theorem4}
Let the eigenvectors associated with the $K$ smallest eigenvalues of a correlated block affinity matrix $\Tilde{\mathbf{W}}\in\mathbb{R}^{n\times n}$ be column vectors of a matrix ${\Tilde{\mathbf{Y}}=[\Tilde{\mathbf{y}}_0,\dots,\Tilde{\mathbf{y}}_{K-1}]\in\mathbb{R}^{n\times K}}$ where $K$ denotes the number of clusters. The correlations between different blocks results in mapping all feature vectors onto the same location on the column vector $\Tilde{\mathbf{y}}_0$ that is associated with the smallest eigenvalue $\Tilde{\lambda}_0$.
\end{theorem}
\begin{proof}
\vspace{-0.5mm}
See Appendix A.5 of the supplementary material.
\vspace{-0.5mm}
\end{proof}

The analysis of the effect of outliers on group structures is extended in the following corollary.
\vspace{-0.5mm}
\begin{corollary}
\textit{Let the eigenvectors associated with the $K$ smallest eigenvalues of $\mathbf{W}\in\mathbb{R}^{n\times n}$ and $\Tilde{\mathbf{W}}\in\mathbb{R}^{n\times n}$, respectively, be the column vectors of the matrices ${\mathbf{Y}\in\mathbb{R}^{n\times K}}$ and ${\Tilde{\mathbf{Y}}\in\mathbb{R}^{n\times K}}$ where $K$ denotes the number of clusters. Assuming that the column vectors of ${\mathbf{Y}}$ and $\Tilde{\mathbf{Y}}$ are valued in a range ${\{y_{\mathrm{min}},y_{\mathrm{max}}\}}$, the squared Euclidean distance between row vectors of ${\mathbf{Y}}$ associated with different blocks is greater than that of $\Tilde{\mathbf{Y}}$.}
\end{corollary}

\begin{proof}
\vspace{-0.5mm}
See Appendix A.6 of the supplementary material.
\vspace{-0.5mm}
\end{proof}

To illustrate the different outlier effects, examples of Fiedler vector estimates are shown for ideal and corrupted affinity matrices in Fig.~\ref{fig:scenario1} 
and Fig.~\ref{fig:scenario2}, respectively. In the ideal case, the vertices of different clusters do not have edges between each other while vertices of the same cluster are connected with strong edges. If such an ideally clustered graph is embedded on the real line using the Fiedler vector, the vertices of the same cluster are concentrated while they are far away from the vertices of a different cluster. Therefore, the embedding results of different clusters are easily separable which is crucial for graph partitioning problems. On the other hand, the corrupted graph in Fig.~\ref{fig:scenario2} includes two typical outlier effects. Based on Theorems~\ref{app:theorem2}-\ref{app:theorem4}, Type I outliers are embedded far from the clusters while Type II outliers that correlate with more than one cluster are embedded between different clusters. In such scenarios, the outliers result in a performance degradation of the Fiedler vector estimate that would lead to losing the group structure information of the graph. However, it is important to see that both types of outliers have a common characteristic: \emph{their overall edge weights deviate from the typical nodes}. This characteristic is leveraged upon in the following section where we propose a robust Fiedler vector estimator.

\vspace{-3mm}
\subsection{Problem Statement}

Given a data set of feature vectors ${\mathbf{X}\hspace{-0.2mm}=\hspace{-0.2mm}[\mathbf{x}_1,\dots,\mathbf{x}_n]\in\mathbb{R}^{m\times n}\hspace{-0.5mm},}$ the aim of this work is to estimate the Fiedler vector ${\mathbf{y}_F\in\mathbb{R}^n}$ such that it embeds each feature vector on a real line, providing robustness at a reasonable computation cost. In the following section, the main ideas of the proposed robust Fiedler vector estimator are explained including an unsupervised penalty parameter selection procedure, an analysis of the computational cost, and possible applications of practical interest.

\vspace{-2.5mm}
\section{Robust Fiedler Vector Estimation}\label{sec:RRLPI}
Let data matrix $\mathbf{X}$ be subject to heavy tailed noise and outliers that obscure the underlying group structure in the graph  $G=\{V,E,\mathbf{W}\}$ that represents $\mathbf{X}$. In the previous section, it was shown that the overall edge weight attached to a vertex is a valuable characteristic of an outlier because it significantly differs from the typical overall edge weight. Thus, the overall edge weight of attached to vertex $i$ is modeled as
\begin{equation}\label{eq:edgeweighterrormodel}
     d_i=d_{\mathrm{typ}}+\epsilon_i.
\end{equation}
Here, $d_i=\sum_j^n w_{i,j}$ and $\epsilon_i$, respectively, denote the overall edge weight and the error term for $i$th vertex, $d_{\mathrm{typ}}$ is the typical overall edge weight of the graph $G$. In practice a robust estimator, such as the median is used, i.e. $\hat{d}_{\mathrm{typ}}=\mathrm{med}(\mathbf{d)}$ for a vector of overall edge weights $\mathbf{d}=[d_1,\dots,d_n]$.
For Fiedler vector estimation, an error vector $\boldsymbol{\epsilon}\in\mathbb{R}^n$ is constructed using the error terms associated with each overall edge weight in $\mathbf{d}$. Then, the transformation vector $\boldsymbol{\beta}_F$ associated with $\mathbf{y}_F$ is computed using penalized ridge regression M-estimation \cite{RRM} by solving the following zero gradient equation
\begin{equation}\label{eq:RRM_estimator}
    -\sum_{i=1}^{n}\psi\Big(\frac{\epsilon_i}{\hat{\sigma}}\Big)\Big(\frac{\mathbf{x}_i^{\top}}{\hat{\sigma}}\Big)+\gamma\boldsymbol{\beta}_F=\mathbf{0},
\end{equation}
where $\gamma$ denotes the penalty parameter, $\hat{\sigma}$ is a robust scale estimate of $\boldsymbol{\epsilon}$ and $\psi$ is a bounded and continuous odd function called the score-function. A popular M-estimator is defined by Huber's function
\begin{equation}
    \psi\Big(\frac{\epsilon_i}{\hat{\sigma}}\Big) = \begin{cases} \frac{\epsilon_i}{\hat{\sigma}}, \hspace{1.58cm}\mathrm{for}\hspace{3mm} \big|\frac{\epsilon_i}{\hat{\sigma}}\big|\leq c \\
        c\cdot\mathrm{sign}\big(\frac{\epsilon_i}{\hat{\sigma}}\big), \hspace{5mm}\mathrm{for}\hspace{3mm} \big|\frac{\epsilon_i}{\hat{\sigma}}\big|>c
    \end{cases},
\end{equation}
where $c$ is the tuning parameter that trades off robustness against outliers and asymptotic relative efficiency under a Gaussian distribution model for $\epsilon$ (see \cite{Robuststatistics} for a discussion). A frequently used robust scale estimate $\hat{\sigma}$ is the normalized median absolute deviation \cite{Robuststatistics} that is defined by
\begin{equation}\label{eq:madn}
     \hat{\sigma}=\mathrm{madn(\mathbf{\boldsymbol{\epsilon}}}) = 1.4826 \cdot\mathrm{med}\vert \boldsymbol{\epsilon}-\mathrm{med}(\boldsymbol{\epsilon})\vert,
\end{equation}
where $\mathrm{med}(\boldsymbol{\epsilon})$ denotes the median of the error vector. The motivation for adopting M-estimation for Fiedler vector estimation is that a bounded score function, such as, Huber's ensures that nodes with atypical edge weights are down-weighted in Eq.~\eqref{eq:RRM_estimator}.

\vspace{-3mm}
\subsection{Theoretical Analysis} Eq.~\eqref{eq:RRM_estimator} in matrix form leads to
\begin{equation}\label{eq:RRM_matrixform}
       \hat{\boldsymbol{\beta}}_F=(\mathbf{X}\boldsymbol{\Omega}\mathbf{X}^{\top}+\gamma\sigma^2\mathbf{I})^{-1}\mathbf{X}\boldsymbol{\Omega}\mathbf{y}_F,
\end{equation}
where $\boldsymbol{\Omega}\in\mathbb{R}^{n\times n}$ is diagonal matrix of weights defined by ${\boldsymbol{\Omega}=\mathrm{diag}(\omega_1,\dots,\omega_n)}$ with $\omega_i=\omega\Big(\frac{\epsilon_i}{\hat{\sigma}}\Big)$ and
\begin{equation}\label{eq:weightfunction}
   \omega\Big(\frac{\epsilon_i}{\hat{\sigma}}\Big)=
   \begin{cases}
   \psi\Big(\frac{\epsilon_i}{\hat{\sigma}}\Big)/\Big(\frac{\epsilon_i}{\hat{\sigma}}\Big) \hspace{5 mm}\mathrm{for}\hspace{3mm} \frac{\epsilon_i}{\hat{\sigma}}\neq0\\
   1, \hspace{2.05cm}\mathrm{for}\hspace{3mm} \frac{\epsilon_i}{\hat{\sigma}}=0.
   \end{cases}
\end{equation}

As discussed in Section~\ref{sec:Preliminaries}, the LPI method uses a linearization of the embedding operation for the Fiedler vector estimation. To understand the relationship between LPI and RRLPI, we first clarify the relation between RLPI and RRLPI.

\begin{theorem}\label{app:theorem5}
RRLPI is a robustly weighted RLPI, and for $\boldsymbol{\Omega}=\mathbf{I}$, it gives identical solutions to RLPI based Fiedler vector estimation.
\end{theorem}
\begin{proof}
\vspace{-0.75mm}
See Appendix B.1 of the supplementary material.
\vspace{-0.75mm}
\end{proof}
From Theorem~\ref{app:theorem5}, it follows that for $\gamma>0$ and/or $\boldsymbol{\Omega}\neq\mathbf{I}$, the estimated tranformation vector $\hat{\boldsymbol{\beta}}_F$ is not the eigenvector of the eigen-problem in Eq.~\eqref{eq:eigenproblem_proposed} which means that it is not associated with the Fiedler value. However, the following theorem shows in which cases $\boldsymbol{\beta}_F$ exactly gives the eigenvector of eigen-problem in Eq.~\eqref{eq:eigenproblem_proposed}.
\vspace{-0.5mm}
\begin{theorem}\label{app:theorem6}
Suppose $\mathbf{y}_F$ is the Fiedler vector associated with the second largest eigenvalue of the eigen-problem in Eq.~\eqref{eq:eigenproblem_proposed}. Further, let $\boldsymbol{\Omega}\in\mathbb{R}^{n\times n}$ and $\boldsymbol{\Psi}\in\mathbb{R}^{m\times m}$ be two weighting matrices such that $\mathbf{U}^{\top}\boldsymbol{\Psi}\mathbf{U}=\mathbf{I}$ and $\mathbf{V}^{\top}\boldsymbol{\Omega}\mathbf{V}=\mathbf{I}$. If $\mathbf{y}_F$ is in the space spanned by row vectors of the weighted data matrix $\mathbf{X}^{\ast}$, for $\mathbf{X}^{\ast}=\mathbf{X}\boldsymbol{\Omega}$, the corresponding transformation vector $\hat{\boldsymbol{\beta}}_F$ estimated with RRLPI is
the eigenvector of the eigen-problem in Eq.~\eqref{eq:eigenproblem_proposed} as $\gamma$ decreases to zero.
\end{theorem}
\begin{proof}
\vspace{-0.75mm}
See Appendix B.2 of the supplementary material.
\vspace{-0.75mm}
\end{proof}

Based on Theorem~\ref{app:theorem6}, the estimated transformation vector $\hat{\boldsymbol{\beta}}_F$ is the solution of Eq.~\eqref{eq:eigenproblem_proposed} for $\gamma\rightarrow0$, and $\mathbf{U}^{\top}\boldsymbol{\Psi}\mathbf{U}=\mathbf{I}$, $\mathbf{V}^{\top}\boldsymbol{\Omega}\mathbf{V}=\mathbf{I}$. To understand the relationship between RRLPI and LPI, the results of this theorem are extended for all transformation vectors $\hat{\boldsymbol{\beta}}_i\in[\hat{\boldsymbol{\beta}}_0,\dots,\hat{\boldsymbol{\beta}}_{n-1}]$ for the case that the feature space $m$ is greater than the number of feature vectors $n$ and the feature vectors are linearly independent, i.e. $\mathrm{rank}(\mathbf{X})=n$.
\vspace{-0.5mm}
\begin{corollary}\label{app:corollary6.1}
\textit{If the feature vectors are linearly independent, i.e. $\mathrm{rank}(\mathbf{X})=n$, all transformation vectors are solutions of Eq.~\eqref{eq:eigenproblem_proposed} for $\gamma\rightarrow0$, and $\mathbf{U}^{\top}\boldsymbol{\Psi}\mathbf{U}=\mathbf{I}$, $\mathbf{V}^{\top}\boldsymbol{\Omega}\mathbf{V}=\mathbf{I}$ which means that RRLPI is identical to LPI.}
\end{corollary}

\begin{proof}
\vspace{-0.75mm}
See Appendix B.3 of the supplementary material.
\vspace{-0.75mm}
\end{proof}
\vspace{-3mm}
\subsection{$\Delta$-Separated Sets for Penalty Parameter Selection}
\vspace{-0.5mm}
The geometric structure of well-spread $\ell_2^2$-representations shows that the two sets $\mathrm{\mathbf{s}}$ and $\mathrm{\mathbf{t}}$ are well separated if every pair of points $s_i \in \mathrm{\mathbf{s}}$ and $t_j \in \mathrm{\mathbf{t}} $ are mapped at least ${\Delta=\phi(1/\mathrm{log}^{-2/3}n)}$ apart in $\ell_2^2$ distance \cite{arora2009expander}. Inspired by well-spread $\ell_2^2$-representations, we propose a penalty parameter selection algorithm by projecting graph vertices onto a real line using RRLPI-based Fiedler vector estimation such that every pair of two sets $s_i \in \mathrm{\mathbf{s}}$ and $t_j \in \mathrm{\mathbf{t}}$ is at least $\Delta=\phi(1/\mathrm{log}^{-2/3}n)$ apart in $\ell_2^2$ distance for the estimated penalty paramater.

 \begin{figure}[!tbp]
\begin{minipage}[b]{1.0\linewidth}
  \centering
  \includegraphics[trim={0 3.9mm 0 3mm},clip,width=8.5cm]{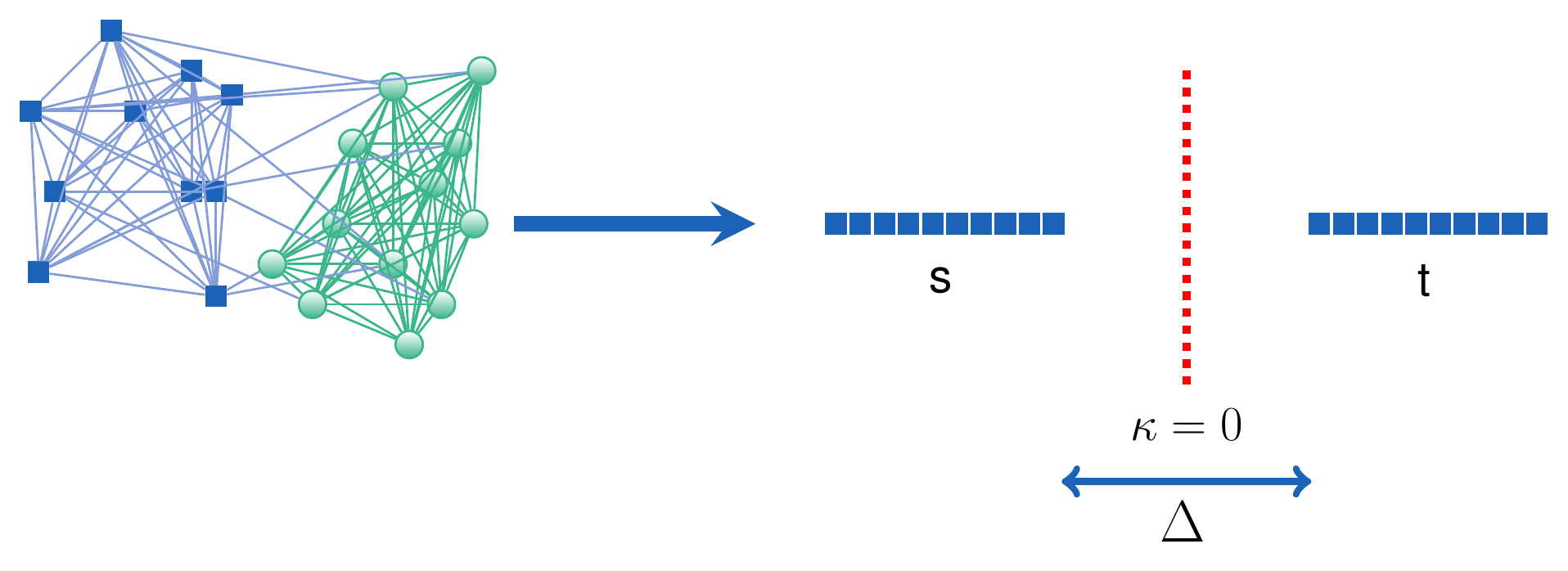}
\end{minipage}
\caption{Example of $\Delta$-separated sets $\mathbf{s}$ and $\mathbf{t}$}
\label{fig:penaltyparameter}
\vspace{-1mm}
\end{figure}

Let $\gamma_i\in\boldsymbol{\gamma}$ be the $i$th candidate penalty parameter in Eq.~\eqref{eq:RRM_matrixform} from a given vector of candidate penalty parameters ${\boldsymbol{\gamma}= [\gamma_{\mathrm{min}},\dots,\gamma_{\mathrm{max}}]\in\mathbb{R}^{N}}$. Further, suppose that for each candidate penalty parameter $\gamma_i$, there exists an associated Fiedler vector estimate $\hat{\mathbf{y}}_F^{(\gamma_i)}$ that projects the graph vertices onto the real line. The geometric structure of well-spread $\ell_2^2$-representations allows for designing the sets $\mathrm{\mathbf{s}}$ and $\mathrm{\mathbf{t}}$ by projecting the points on a random line such that, for a suitable constant $\kappa$, the points that are located on the left and right hand sides of $\kappa$ are the initial candidates for the sets $\mathrm{\mathbf{s}}$ and $\mathrm{\mathbf{t}}$, respectively \cite{arora2009expander}.

\begin{algorithm}[tbp!]
\setstretch{1}
\DontPrintSemicolon
\KwIn{
A data matrix $\mathbf{X}\in \mathbb{R}^{m\times n}$ and an associated affinity matrix $\mathbf{W}\in \mathbb{R}^{n\times n}$, $N_{\mathrm{min}}$ }
\For{$\gamma_i=\gamma_{\mathrm{min}},\dots,\gamma_{\mathrm{max}}$}{
\textbf{Initialization:}\\

Evaluate the Fiedler vector $\mathbf{y}_F\in\mathbb{R}^{n}$ as in \cite{Laplacianeigenmaps}\\
Compute $\boldsymbol{\beta}_F\in\mathbb{R}^{m}$ for $\mathbf{y}_F=\mathbf{X}^{\top}\boldsymbol{\beta}_F$\\
\textbf{RRLPI}\\
Update the error vector $\boldsymbol{\epsilon}\in\mathbb{R}^{n}$ using Eq.~\eqref{eq:edgeweighterrormodel}\\
Compute $\hat{\sigma}$ via Eq.~\eqref{eq:madn}\\
Calculate the weights $\omega_i=\omega(\frac{\epsilon_i}{\hat{\sigma}})$, $\boldsymbol{\Omega}=\mathrm{diag}(\boldsymbol{\omega})$\\
Solve Eq.~\eqref{eq:RRM_matrixform} and estimate $\hat{\boldsymbol{\beta}}_F^{(\gamma_i)}$\\
Estimate the Fiedler vector for $\hat{\mathbf{y}}_F^{(\gamma_i)}=\mathbf{X}^{\top}\hat{\boldsymbol{\beta}}_F^{(\gamma_i)}$\\ 
\textbf{$\Delta$-separated sets}\\
Generate sets $\mathbf{s}^{(\gamma_i)}$ and $\mathbf{t}^{(\gamma_i)}$ via Eq.~\eqref{eq:createsets}\\
Calculate $\|\bar{y}_{\mathrm{min},F}^{(\gamma_i)}-\bar{y}_{\mathrm{max},F}^{(\gamma_i)} \|_2^2$ s.t. $\bar{y}_{\mathrm{min},F}^{(\gamma_i)}\in \mathrm{\mathbf{s}}$ and\\ $\bar{y}_{\mathrm{max},F}^{(\gamma_i)}\in \mathrm{\mathbf{t}}$ and collect in a vector $\mathbf{z}\in\mathbb{R}^N$\\
\While{$N_{\mathbf{s}}\geq N_{\mathrm{min}}\hspace{2.5mm}\mathrm{and}\hspace{2.5mm} N_{\mathbf{t}}\geq N_{\mathrm{min}}$}{
Create $\mathbf{r}\in\mathbb{R}^{N_\mathbf{r}^{(\gamma_i)}}$ using Eq.~\eqref{eq:residualvec}\\
Update $N_\mathbf{r}^{(\gamma_i)}$\\
\uIf{$\mathbf{s}^{(\gamma_i)}$ and $\mathbf{t}^{(\gamma_i)}$ are $\Delta$-\textnormal{separated} }{break\\}\vspace{-1mm}
\textbf{end}
}\vspace{-1mm}Collect $N_\mathbf{r}^{(\gamma_i)}$ into a vector $\mathbf{h}\in\mathbb{R}^N$}
\eIf{at least one pair of $\Delta$-separated sets exist}{Estimate $\hat{\gamma}$ using Eq.~\eqref{eq:estimategamma}\\}
{ Estimate $\hat{\gamma}$ using Eq.~\eqref{eq:estimategammaalternative}}

Estimate transformation vector $\hat{\boldsymbol{\beta}}_F^{\hat{\gamma}}$ in Eq.~\eqref{eq:RRM_matrixform} for $\hat{\gamma}$ \\
Estimate the Fiedler vector for $\hat{\mathbf{y}}_F^{(\hat{\gamma})}=\mathbf{X}^{\top}\hat{\boldsymbol{\beta}}_F^{(\hat{\gamma})}$\\
\caption{Robust Fiedler vector estimation\label{IR}}
\KwOut{A robust estimate of the Fiedler vector $\hat{\mathbf{y}}_F^{(\hat{\gamma})}$}
\end{algorithm}

It has been shown (see, e.g. \cite{spielman2007spectral}) that it is possible to split a candidate Fiedler vector $\hat{\mathbf{y}}_F^{(\gamma_i)}$ into the two subsets $\mathrm{\mathbf{s}}^{(\gamma_i)}$ and $\mathrm{\mathbf{t}}^{(\gamma_i)}$ for $\kappa=0$. Another possible option for $\kappa$ is the median of embeddings such that $\kappa=\mathrm{med}(\hat{\mathbf{y}}_F^{(\gamma_i)})$. From the definition of the $\Delta$-separated sets, the projection subsets $\mathrm{\mathbf{s}}^{(\gamma_i)}$ and $\mathrm{\mathbf{t}}^{(\gamma_i)}$ associated with $\gamma_i$ taking values between zero and one. Therefore, after selecting the members of the two sets $\mathrm{\mathbf{s}}^{(\gamma_i)}\in\mathbb{R}^{N_\mathbf{s}}$ and $\mathrm{\mathbf{t}}^{(\gamma_i)}\in\mathbb{R}^{N_\mathbf{t}}$ associated with $\gamma_i$, the final design of the sets $\mathrm{\mathbf{s}}^{(\gamma_i)}$ and $\mathrm{\mathbf{t}}^{(\gamma_i)}$ is performed using the rescaled estimated Fiedler vector $\bar{\mathbf{y}}^{(\gamma_i)}$ as
\begin{equation}\label{eq:createsets}
\begin{split}
\begin{aligned}
    \mathrm{\mathbf{s}}^{(\gamma_i)}&=\big\{\bar{y}_{j,F}^{(\gamma_i)}: \hat{y}_{j,F}^{(\gamma_i)}>\kappa\big\}\\
    \mathrm{\mathbf{t}}^{(\gamma_i)}&=\big\{\bar{y}_{j,F}^{(\gamma_i)}: \hat{y}_{j,F}^{(\gamma_i)}\leq\kappa\big\}.
\end{aligned}
 \end{split}
\end{equation}
Here, $\hat{y}_{j,F}^{(\gamma_i)}$ denotes the $j$th element of the estimated Fiedler vector $\hat{\mathbf{y}}_F^{(\gamma_i)}$ and  $\bar{y}_{j,F}^{(\gamma_i)}$ is the $j$th element of the rescaled estimated Fiedler vector $\bar{\mathbf{y}}_F^{(\gamma_i)}$. If the rescaled Fiedler vector $\bar{\mathbf{y}}_F^{(\gamma_i)}$ is not sufficiently sparse, it contains pairs of points $\bar{y}_{i,F}^{(\gamma_i)}\in \mathrm{\mathbf{s}}$ and $\bar{y}_{j,F}^{(\gamma_i)}\in \mathrm{\mathbf{t}}$ whose squared Euclidean distance is less than $\Delta$. Thus, for a set of pairs of projections $\bar{y}_{i,F}^{(\gamma_i)}\in \mathrm{\mathbf{s}}$ and $\bar{y}_{j,F}^{(\gamma_i)}\in \mathrm{\mathbf{t}}$, a vector of discarded projections $\mathbf{r}^{(\gamma_i)}\in\mathbb{R}^{N_{\mathbf{r}}^{(\gamma_i)}}$ is designed as
\begin{equation}\label{eq:residualvec}
    \mathrm{\mathbf{r}}^{(\gamma_i)}=\Big\{\bar{y}_{i,F}^{(\gamma_i)},\bar{y}_{j,F}^{(\gamma_i)}:   \|\bar{y}_{i,F}^{(\gamma_i)}-\bar{y}_{j,F}^{(\gamma_i)} \|_2^2\leq\Delta\Big\},
\end{equation}
as long as the two sets $\mathrm{\mathbf{s}}^{(\gamma_i)}$ and $\mathrm{\mathbf{t}}^{(\gamma_i)}$ have a reasonable number $N_{\mathrm{min}}$ of projections. The proposed strategy to estimate penalty parameter $\gamma$ is, therefore, to minimize the number of discarded points i.e., 
\begin{equation}\label{eq:estimategamma}
\hat{\gamma}=\underset{\gamma_i=\gamma_{\mathrm{min}},\dots,\gamma_{\mathrm{max}}}{\arg\min} \{N_{\mathbf{r}}^{(\gamma_i)}\},
\end{equation}
where $N_{\mathbf{r}}^{(\gamma_i)}$ denotes the number of discarded projections for candidate penalty parameter ${\gamma_i}$, and $\hat{\gamma}$ is the estimated penalty parameter. In practice, there might not exist $\Delta$-separated sets 
$\mathrm{\mathbf{s}}^{(\gamma_i)}$ and $\mathrm{\mathbf{t}}^{(\gamma_i)}$ for any candidate penalty parameter such that $\gamma_i\in\{\gamma_{\mathrm{min}},\dots,\gamma_{\mathrm{max}}\}$. For example, the sets might not be $\Delta$-separated, although $N_{\mathbf{r}}^{(\gamma_i)}$ has reached its maximum value. Additionally, the initial sets may not satisfy $N_{\mathbf{s}}<N_{\mathrm{min}}$ or $N_{\mathbf{t}}<N_{\mathrm{min}}$. In such cases, the penalty parameter can alternatively be estimated by maximizing the squared Euclidean distance between the closely valued projections from the two sets $\mathrm{\mathbf{s}}^{(\gamma_i)}$ and $\mathrm{\mathbf{t}}^{(\gamma_i)}$,

\begin{equation}\label{eq:estimategammaalternative}
\hat{\gamma}=\underset{\gamma_i=\gamma_{\mathrm{min}},\dots,\gamma_{\mathrm{max}}}{\arg\max} \{\|\bar{y}_{\mathrm{min},F}^{(\gamma_i)}-\bar{y}_{\mathrm{max},F}^{(\gamma_i)} \|_2^2\},
\end{equation}
where $\bar{y}_{\mathrm{min},F}^{(\gamma_i)}\in \mathrm{\mathbf{s}}$ and $\bar{y}_{\mathrm{max},F}^{(\gamma_i)}\in \mathrm{\mathbf{t}}$ are the minimum and the maximum valued projections from the sets $\mathrm{\mathbf{s}}^{(\gamma_i)}$ and $\mathrm{\mathbf{t}}^{(\gamma_i)}$, respectively.

The main steps of the proposed Fiedler vector estimation are summarized in Algorithm~1.

\vspace{-3mm}
\subsection{Computational Complexity}
\vspace{-0.5mm}
As computational complexity is essential for the scalability of graph embedding techniques, the computational complexity of the proposed approach is analyzed in terms of its main operations. The computational complexity of operations is detailed using the term flam \cite{GStewartMatrixI}, which is a compound operation that includes one addition and one multiplication. For the cases when the complexity is not specified as flam, the Landau's big $O$ symbol is used. The computational complexity of the proposed approach is given as follows:

\textbf{Graph Construction:} The pairwise cosine similarity which takes $\frac{1}{2}n^2m+2nm$ as in \cite{RLPI} can be used for contructing graph $G$.

\textbf{Initialization:} For large eigen-problems, e.g. MATLAB uses a Krylov Schur decomposition \cite{KrylovLargeEigenproblems}. The algorithm includes two main phases that are known as expansion and contraction. When $n$ is larger than $p$, where $p$ denotes the number of Lanczos basis vectors (preferably chosen as $p\geq 2q$ for $q$ eigenvectors), the computational complexity of the algorithm can mainly be attributed to expansion and contraction phases. The expansion phase requires between $n(p^2-q^2)$ flam and ${2n(p^2-q^2)}$ flam. The contraction phase requires ${npq}$ flam \cite{GStewartMatrixII}.

\textbf{Robust Regularized Locality Preserving Indexing (RRLPI):} The proposed projection algorithm requires an estimate of scale that uses repetitive medians. The complexity of repetitive medians is  $O(n)$ \cite{remedian}. Further, for a densely connected matrix, the complexity is mainly attributed to the Cholesky decomposition which is of complexity $O(n^3)$ or, more specifically,  $\frac{1}{6}n^3$ flam \cite{GStewartMatrixI}. This complexity can be reduced to $O(n)$ using \cite{cholrankdeficient} if the matrix is rank deficient. If the matrix is sparse, the computation cost of decomposition can be reduced to $t(2ns+3n+5m)$ flam using a least squares algorithm such as \cite{leastsquare} where $s$ denotes the average number of nonzero features.
    
 \textbf{$\Delta$-Separated Sets:}
    To split the projection into two sets as $\mathbf{s}$ and $\mathbf{t}$, the vector $\mathbf{y}$ must be sorted which is of complexity $O(n\mathrm{log}n)$ and there are computationally efficient alternatives such as \cite{efficientsorting} for which the complexity is reduced to $O(n\sqrt{\mathrm{log}n)}$. To compute $\Delta$-separated sets a maximum of $n$ projections can be subtracted which means that this operation maximally takes $n$ flam.

Summing up the terms with respect to flam yields minimally
\begin{align*}
    t(2ns+3n+5m)+\frac{1}{2}n^2m+n(2m+p^2-k^2+pk+1)
\end{align*}
flam. Hence, the complexity is of order $O(n^2)$. Based on the information that both $O(n)$ and $O(n\mathrm{log}n)$ are considerably smaller than $O(n^2)$, the minimum computational cost can be summarized as $O(n^2)$ for each candidate penalty parameter. Overall, the algorithm is, at least, of complexity $O(Nn^2)$ for a number $N$ of candidate penalty parameters.
\vspace{-4mm}
\subsection{Example Applications}
\vspace{-0.5mm}
Eigenvector decomposition has a large variety of applications, such as, dimension reduction \citeform[13]-\citeform[18], recognition \citeform[19]-\citeform[21] and localization \cite{Localapp1}. Considering images as high-dimensional data sets, it is not surprising that eigen-decomposition is a fundamental research area also in image segmentation, e.g.  \citeform[13]-\citeform[16]. A frequently encountered problem is that the image is subject to noise, which may result in mapping noisy pixels far from the neighboring group of pixels in the embedding space and, consequently, losing the underlying structure. This problem may also occur in cluster enumeration approaches that attempt to find densely connected groups of mappings in the embedding space, which necessitates the application of a robust embedding technique. In the following section, the example of robust graph-based cluster enumeration is discussed.
\vspace{-3mm}
\subsubsection{Cluster Enumeration}
\vspace{-0.5mm}
 Assume that for each candidate number of clusters ${K_{\mathrm{cand}}\in\{K_{\mathrm{min}},\dots,K_{\mathrm{max}}\}}$ there is a clustering algorithm, e.g. \citeform[22]-\citeform[23], that partitions $\hat{\mathbf{y}}^{(\hat{\gamma})}$ into $K_{\mathrm{cand}}$ number of clusters and provides an estimated label vector $\hat{\mathbf{c}}_{K_{\mathrm{cand}}}$. After estimating label vectors for each candidate number of clusters $K_{\mathrm{cand}}$, the cluster number $K$ can be estimated by comparing quality of partitions using modularity as \cite{Newmanmod}
\begin{equation}\label{eq:maxmodscore}
\hat{K}=\underset{\hat{K}_{\mathrm{min}},\dots,\hat{K}_{\mathrm{max}}}{\arg\max} \{Q_{\hat{K}_{\mathrm{cand}}}\},
\end{equation}\\
where\vspace{-1mm}
\begin{equation}\label{eq:modscore}
Q_{\hat{K}_{\mathrm{cand}} }=\frac{1}{2g}\sum_{i,j}^{n}\Big[{w}_{i,j}-\frac{d_id_j}{2g}\Big]\delta(\hat{c}_i,\hat{c}_j)
\end{equation}
denotes the modularity score for a candidate number of clusters $\hat{K}_{\mathrm{cand}}$, $w_{i,j}$ is the edge weight between the $i$th and the $j$th feature vector of $\mathbf{X}$, $d_i$ is the overall edge weight attached to vertex $i$, $\hat{c}_i$ is the estimated community label associated to vertex $i$, $g=\frac{1}{2}\sum_{i,j}w_{i,j}$, and the function $\delta(\hat{c}_i,\hat{c}_j)$ equals 1 if $\hat{c}_i=\hat{c}_j$ and is zero, otherwise.

\vspace{-3mm}
\section{Experimental Evaluation}\label{sec:Experimental}
This section contains the numerical experimental evaluation of the proposed RRLPI method on a broad range of simulated and real-world data sets with applications to robust cluster enumeration and image segmentation. The effects of Type I and Type II outliers on the Fiedler vector estimation are studied for the LE \cite{Laplacianeigenmaps}, LPI \cite{LPI}, RLPI \cite{RLPI}, RLPFM \cite{RLPFM} and RRLPI embedding-based approaches by designing synthetic data Monte Carlo experiments. Then, in addition to the above mentioned embedding approaches, the proposed RRLPI is benchmarked against three state-of-the-art graph-based cluster enumeration approaches, i.e., Le Martelot \cite{Martelot}, Combo \cite{Combo} and Sparcode \cite{Sparcode} and two state-of-the art spectral partitioning approaches, i.e., FastEFM \cite{FSCEFM} and LSC \cite{LSC} in terms of image segmentation.

The cosine similarity is used as the similarity measure in all experiments, unless otherwise specified and RRLPI is computed with the following default parameters: ${\gamma_{\mathrm{min}}=10^{-8}}$, ${\gamma_{\mathrm{max}}=1000}$, ${K_{\mathrm{min}}=1}$, ${K_{\mathrm{max}}=10}$ and ${N_{\mathrm{min}}=\frac{n}{K_{\mathrm{max}}}}$. A MATLAB code for RRLPI is available at:

\vspace{1.5mm}
\hspace{-5mm}https://github/A-Tastan/RRLPI

\begin{figure}[!tbp]
  \centering
\includegraphics[trim={0mm 0mm 0mm 7mm},clip,width=5.3cm]{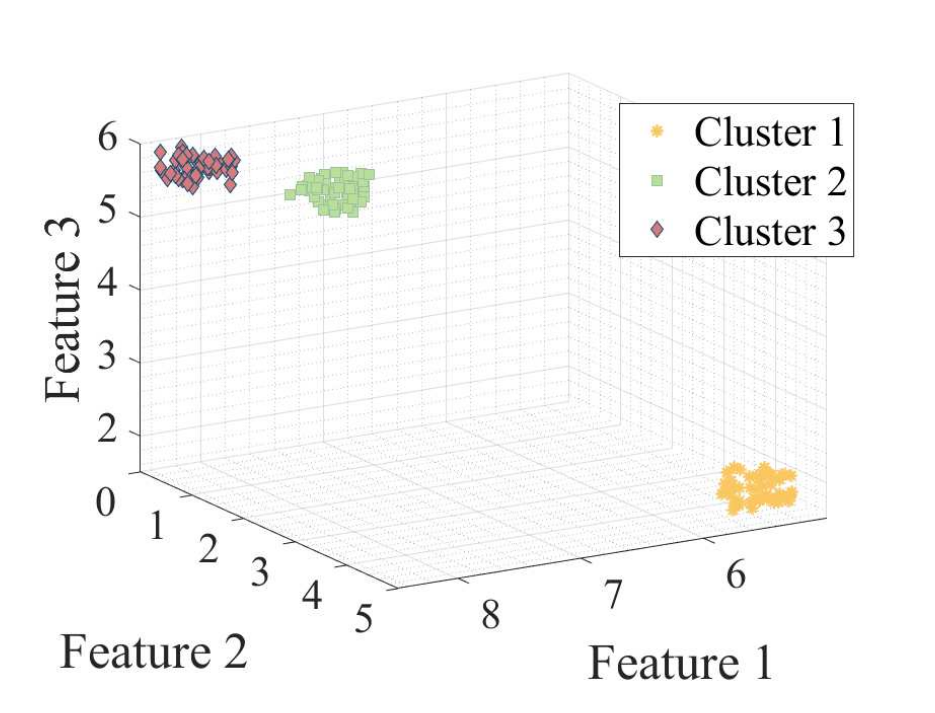}
  \caption{Examplary plot of the first three features of the uncorrupted synthetic data set.}
  \label{fig:initialdatasetthreefeatures}
\vspace{-2mm}
\end{figure}

\begin{figure}[!tbp]
  \centering
\subfloat[LE]{\includegraphics[trim={1mm 0mm 3mm 1mm},clip,width=4.5cm]{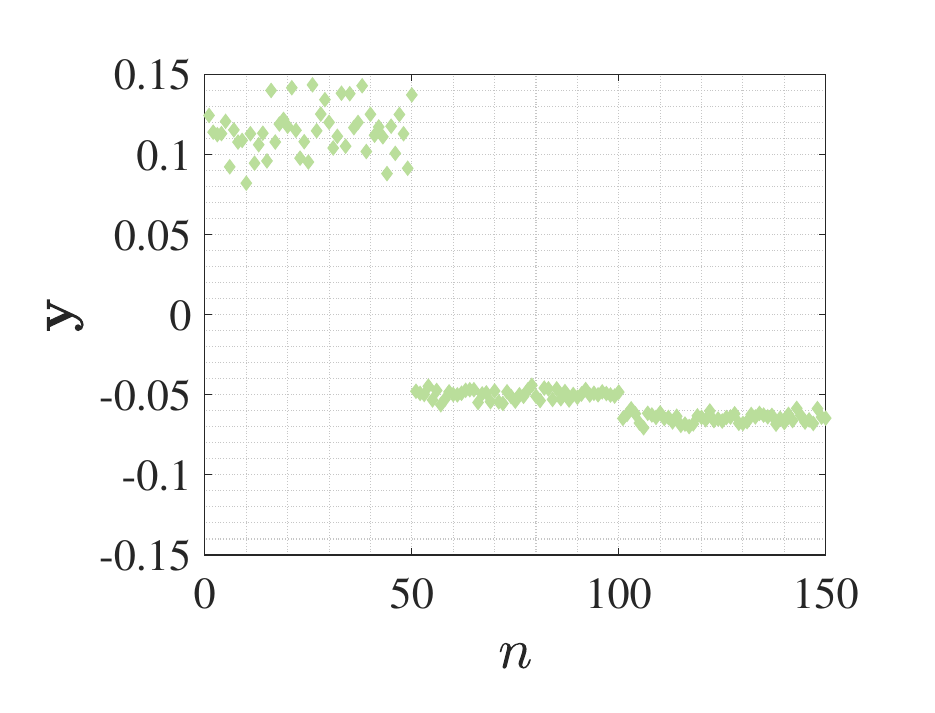}\label{fig:cleanFiedler}}
\subfloat[RRLPI]{\includegraphics[trim={1mm 0mm 3mm 1mm},clip,width=4.5cm]{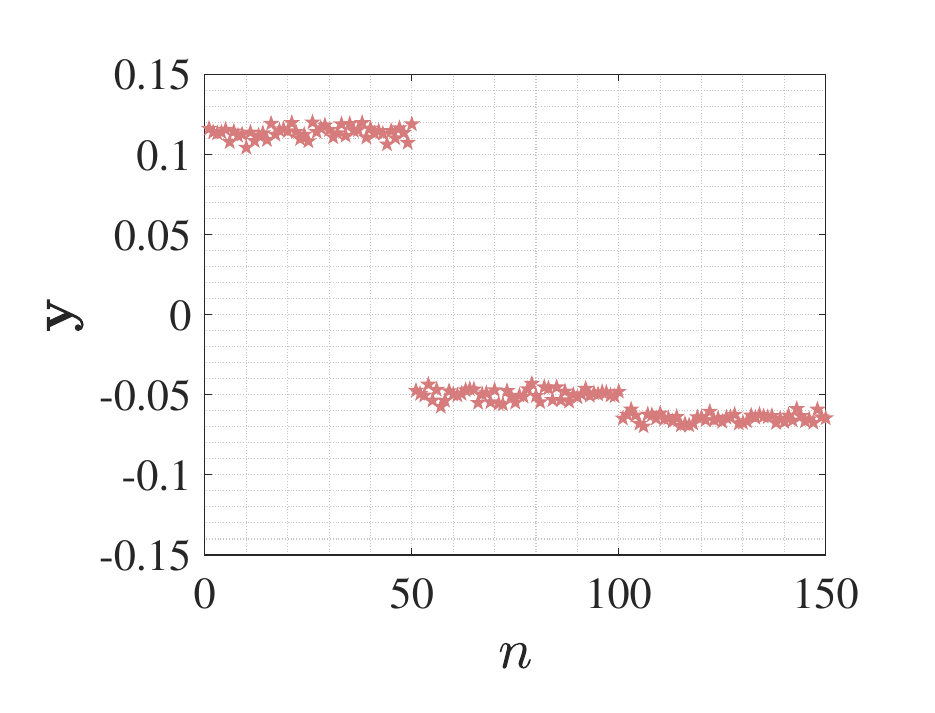}\label{fig:cleanRRLPI}}
  \label{fig:cleandatasetembedding}
  \caption{Estimated eigenvectors for the uncorrupted data set.}
\end{figure}

\vspace{-1.5mm}
\subsection{Performance Measures}
The average partition accuracy $\Bar{p}_{\mathrm{acc}}$ is measured by evaluating
\vspace{-1mm}
\begin{equation}\label{eq:probdetection}
\Bar{p}_{\mathrm{acc}}=\frac{1}{nN_E}\sum_{i=1}^{N_E}\sum_{j=1}^{n}\mathbbm{1}_{\{\hat{c}_j=c_j\}},
\vspace{-1mm}
\end{equation}
where 
\vspace{-1mm}
\begin{equation}\label{eq:indicatorfunc}
\mathbbm{1}_{\{\hat{c}_j=c_j\}}=\begin{cases}
     1 ,        & \text{if } \hat{c}_j=c_j\\
     0 ,           & \text{otherwise}
\end{cases},
\vspace{-1mm}
\end{equation}
$n$ is the number of observations, $N_E$ is the total number of experiments, and $\hat{c}_j$ and $c_j$ are the estimated and ground truth labels for the $j$th observation, respectively.

The empirical probability of detection $p_{\mathrm{det}}$ is used to assess cluster enumeration performance as follows
\vspace{-1.5mm}
\begin{equation}\label{eq:probdetection}
p_{\mathrm{det}}=\frac{1}{N_E}\sum_{i=1}^{N_E}\mathbbm{1}_{\{\hat{K}=K\}},
\vspace{-1.25mm}
\end{equation}
where $\hat{K}$ denotes the estimated number of clusters and $\mathbbm{1}_{\{\hat{K}=K\}}$ is the indicator function.

\begin{figure}[!tbp]
  \centering
\includegraphics[trim={0mm 0mm 0mm 7mm},clip,width=5.3cm]{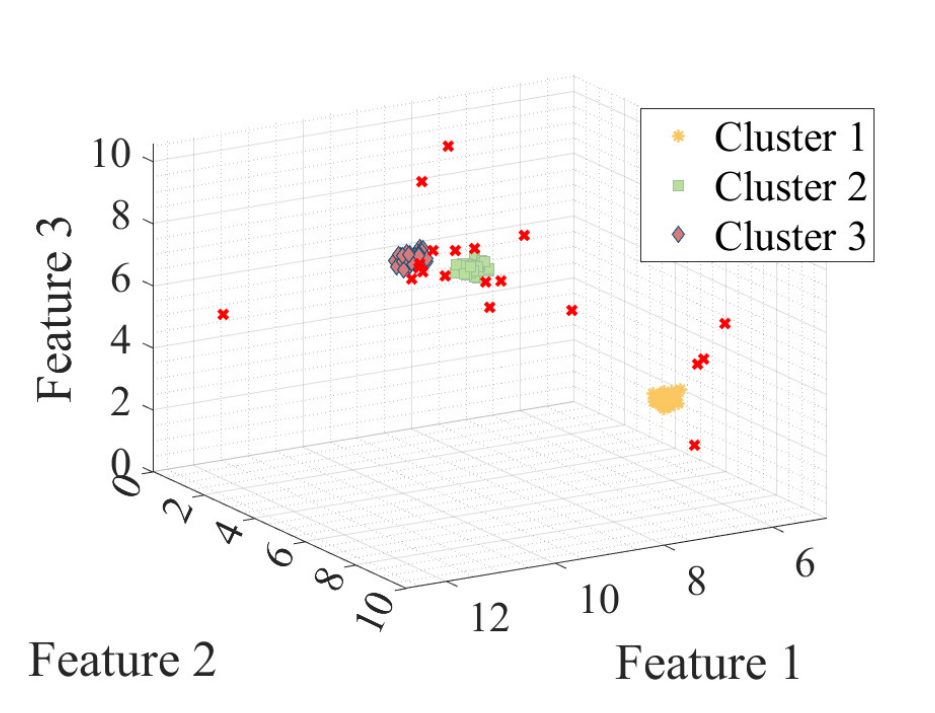}\label{fig:corrupteddataset}
  \caption{Examplary plot of the first three features of the synthetic data set after corruption with Type I and Type II outliers (red crosses). See Section~\ref{sec:Motivation}, for a discussion.}
  \label{fig:corrupteddatasetthreefeatures}
    \vspace{-6mm}
\end{figure}
\begin{figure}[!tbp]
  \centering
  \subfloat[LE]{\includegraphics[trim={1mm 0mm 3mm 1mm},clip,width=4.5cm]{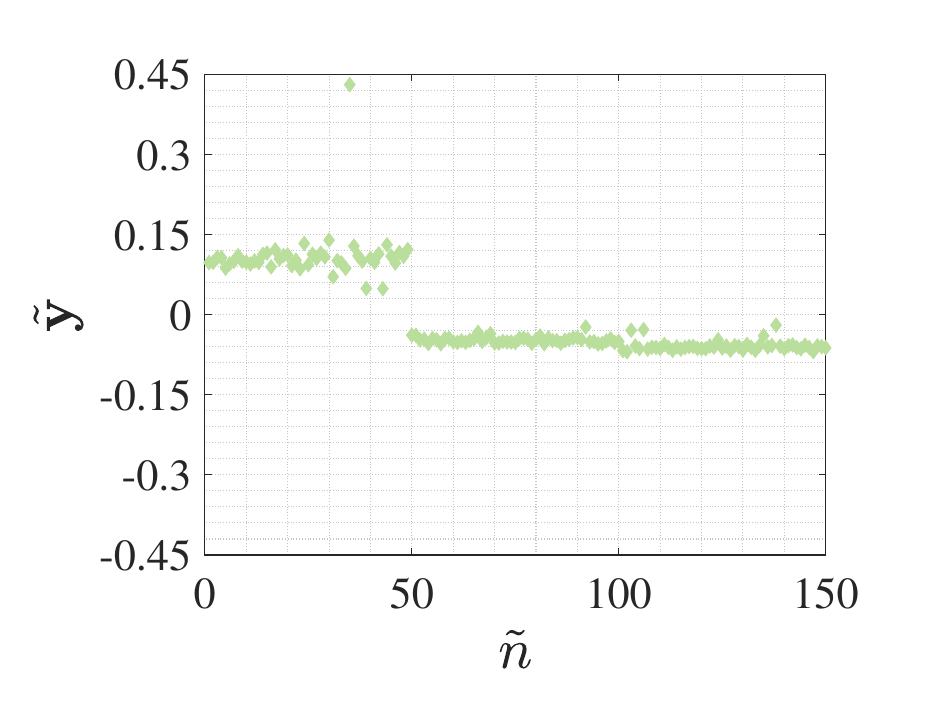}\label{fig:corruptedFiedler}}
\subfloat[RRLPI]{\includegraphics[trim={1mm 0mm 3mm 1mm},clip,width=4.5cm]{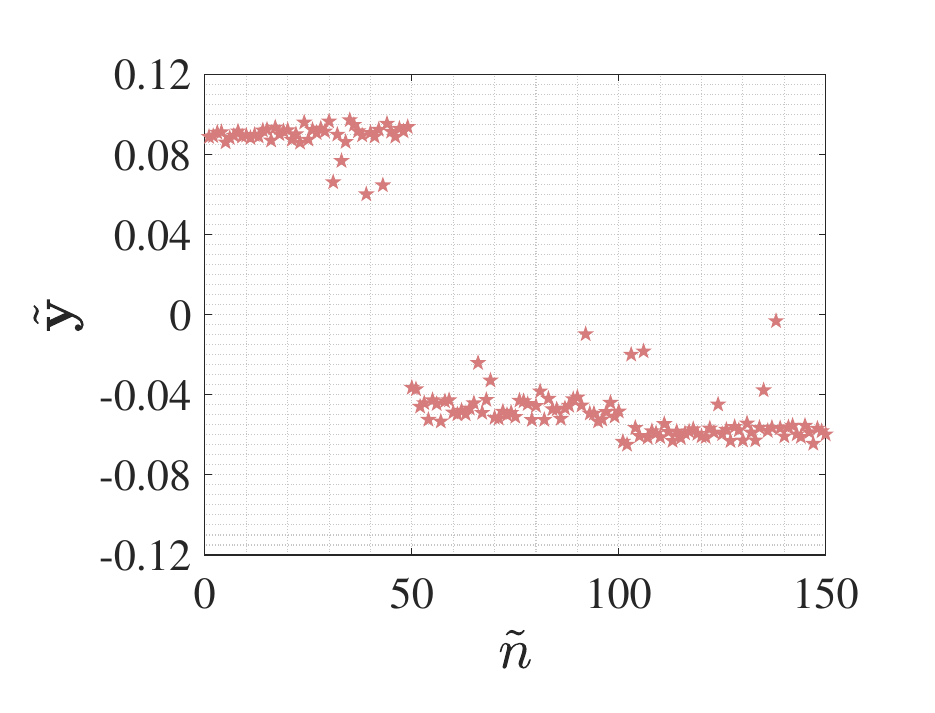}\label{fig:corruptedRRLPI}}\\
  \label{fig:corruptedatasetembedding}
  \caption{Estimated eigenvectors for the corrupted data set.}
\end{figure}

The contour matching score $F_{\mathrm{score}}$ for boundaries and the Jaccard index $J$ are used for the numerical performance analysis in the case of image segmentation \cite{performancemeasureimage}. The $F_{\mathrm{score}}$ quantifies whether a boundary has a match on the ground truth boundary as follows
\begin{equation}
    F_{\mathrm{score}}=2\frac{P\cdot R}{R+P},
\end{equation}
where $P$ and $R$ denote precision and recall values, respectively. The Jaccard index evaluates similarity between estimated and ground truth segmentations according to
\begin{equation}
    J(\hat{\bm{I}}_{\mathrm{seg}},\bm{I}_{\mathrm{seg}})=\frac{\mathrm{TP}}{\mathrm{TP} + \mathrm{FP} + \mathrm{FN}},
\end{equation}
where $\hat{\bm{I}}_{\mathrm{seg}}$ and $\bm{I}_{\mathrm{seg}}$ denote estimated and ground truth segmentations for image $\bm{I}$ and $\mathrm{TP}$, $\mathrm{FP}$, and $\mathrm{FN}$ are true positives, false positives and false negatives, respectively.


\vspace{-3mm}
\subsection{Outlier Effects and Robustness}
To visualize outlier effects on the eigen-decomposition, a synthetic data set is generated for $K=3$ easily separable clusters, see Fig.~\ref{fig:initialdatasetthreefeatures}. The $m$-dimensional feature vectors of each cluster $c_k$, with $k=1,\dots,K$, and $m=6$ are generated as ${\mathbf{x}_{i,k}=\boldsymbol{\mu}_{c_k}+\vartheta_{c_k}\boldsymbol{\upsilon}}$, where $\mathbf{x}_{i,k}$ is the $i$th feature vector associated with the $k$th cluster, $\boldsymbol{\mu}_{c_k}$ is the $k$th cluster centroid, $\vartheta_{c_k}$ is the $k$th scaling constant, and $\boldsymbol{\upsilon}$ is a vector of independently and identically distributed random variables from a uniform distribution on the interval $[-0.5,0.5]$. All details and parameter values to generate the data are provided in Appendix~D.1.1 of the supplementary material.

\begin{table*}[!tbp]
\centering
\small
\begin{tabular}{p{4.3cm}M{1.1cm}M{0.9cm}M{1.2cm}M{0.9cm}M{0.5cm}M{0.6cm}M{0.7cm}M{0.7cm}M{0.5cm}M{1.2cm}}
\hline\hline
\\[-3mm]
 & \multicolumn{8}{c}{$\hat{K}$ for Different Cluster Enumeration Methods} &  &\\\\[-3mm]
\cline{2-9}\\[-3mm]
Data Set & Martelot & Combo & Sparcode & RLPFM & LE & LPI & RLPI & RRLPI & $K$ & Similarity\\
\midrule
Human Gait \cite{humangait},& 4 & 6 & 5 & 4 & 4 & 4 & 4 & 5 &5 & enet\\
Breast Cancer Wisconsin \cite{BreastCancer},&1 & 2 & 2 & 2 & 4 & 2 & 2 & 2 & 2 & cos\\
Iris \cite{Iris},&2& 3 & 2 & 3 & 5 & 3 & 3 & 3 &3 & enet\\
Person Identification \cite{PersonIdentification},& 6 & 7 & 4& 5&10 & 4 & 4&4&4& enet\\
Sonar \cite{Sonar},& 2 & 2 & 2& 2 & 6& 2 & 2 & 2& 2& cos\\
Ionosphere \cite{Ionosphere},& 3 & 3 &4 & 2 & 7 &3 &2 & 2&2&cos\\
D. Retinopathy \cite{D.Retinopathy},& 2 & 2 & 2&2&2 & 2 & 2&2&2&cos\\
Gesture Phase S. \cite{GesturePhase},& 2 & 3 & 3&5&10 & 2 & 6&5&5&cos\\
\hline\hline
\end{tabular}
\caption{\label{tab:table-name}Performance of different cluster enumeration approaches on well-known clustering data sets. The results summarized for similarity measures cosine (cos) and elastic net (enet) using a penalty parameter of $\rho=0.5$.}
\label{tab:tablerealclusterenumeration}
\vspace{-6mm}
\end{table*}
\begin{figure}[!tbp]
  \centering
 \includegraphics[trim={0mm 1mm 0mm 0mm},clip,width=6.2cm]{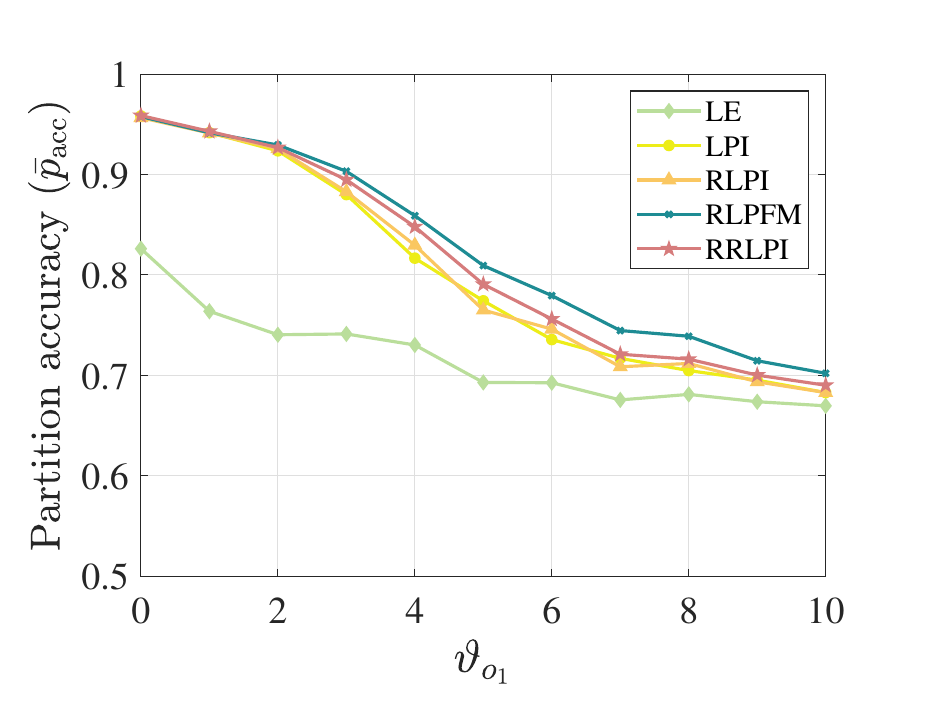}
  \caption{$\bar{p}_{\mathrm{acc}}$ for increasing $\vartheta_{o_1}$ associated with Type I outliers (${n=300,N_{\mathrm{out}}=10, \vartheta_{o_2}=1.5 ,\vartheta_{c_K}=0.5 \hspace{1mm}\mathrm{s.t.}\hspace{1mm} k=1,\dots,K}$).}
  \label{fig:PaccSigma}
  \vspace{-4mm}
\end{figure}
\begin{figure}[!tbp]
  \centering
 \includegraphics[trim={0mm 1mm 0mm 0mm},clip,width=6.2cm]{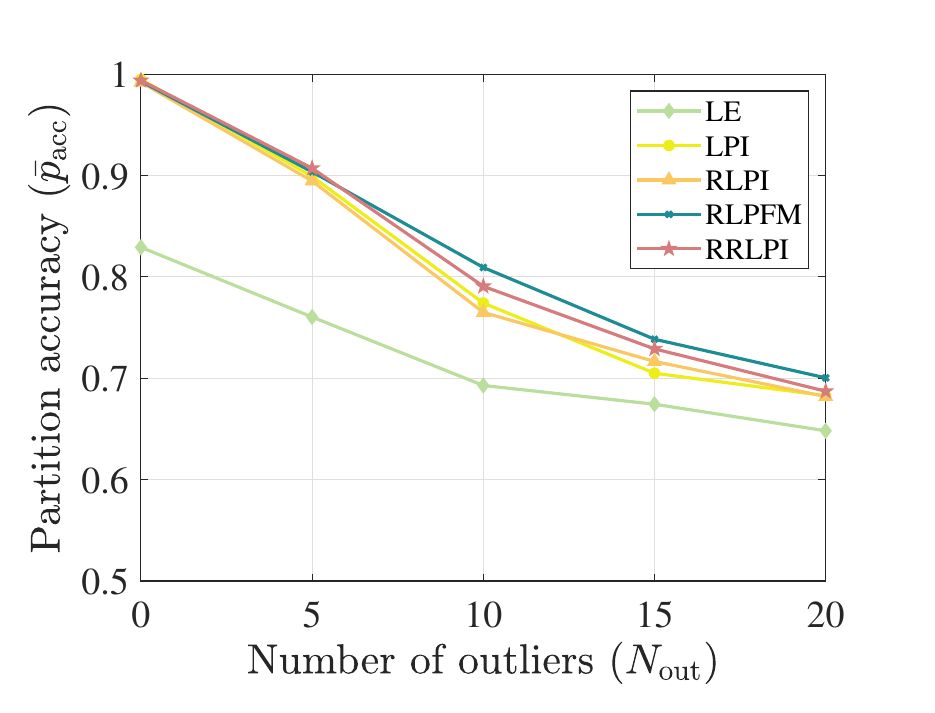}
  \caption{$\bar{p}_{\mathrm{acc}}$ for each of outlier type with increasing $N_{\mathrm{out}}$ (${n=300,\vartheta_{o_1}=5, \vartheta_{o_2}=1.5 ,\vartheta_{c_K}=0.5 \hspace{1mm}\mathrm{s.t.}\hspace{1mm} k=1,\dots,K}$).}
\label{fig:PaccNumout}
\end{figure}

Representative examples of the estimated eigenvectors are shown for LE and for the proposed RRLPI in Fig.~\ref{fig:cleanFiedler} and Fig.~\ref{fig:cleanRRLPI}, respectively. In the absence of outliers, both algorithms provide mappings where the mapping points that are associated with the same cluster are concentrated, and the different clusters are separated. To study robustness, the data set is contaminated with two types outliers, i.e., outliers that do not correlate with any cluster (Type I) and outliers correlate with more than one cluster (Type II); see Sec.~\ref{sec:Motivation} for a definition and a discussion. An example showing the first three features of the contaminated data set is shown in Fig.~\ref{fig:corrupteddatasetthreefeatures}, where both Type I and Type II outliers are highlighted as red crosses.
\begin{figure}[!tbp]
  \centering
\includegraphics[trim={0mm 1mm 0mm 0},clip,width=7.25cm]{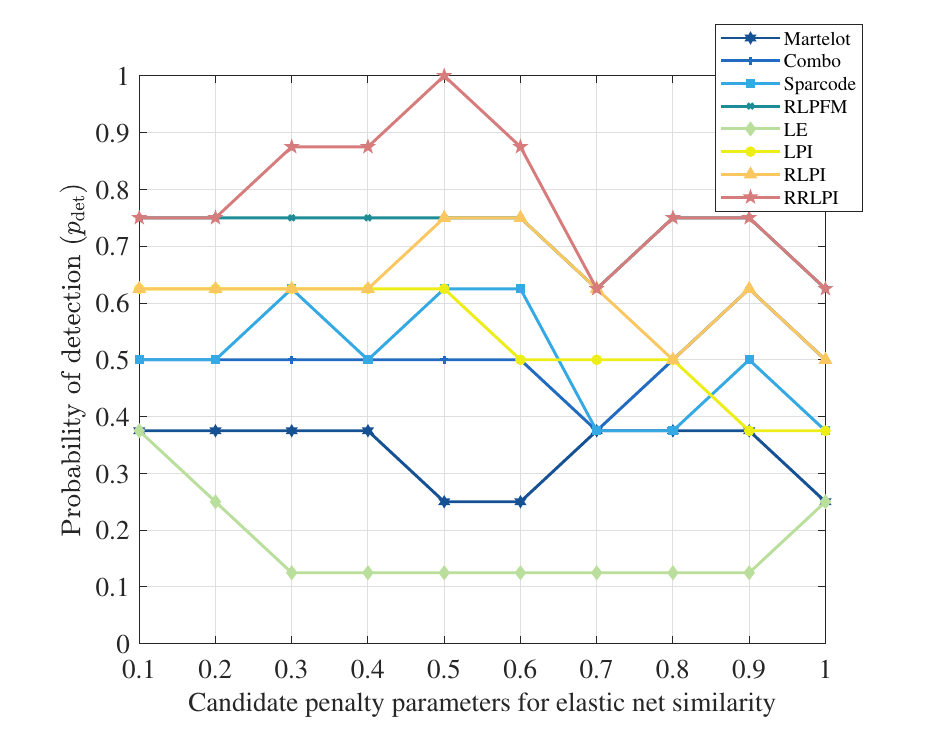}
  \caption{$\bar{p}_{\mathrm{det}}$ with respect to different penalty parameters for $K$-medoids partitioning with Tukey's distance function \cite{Robuststatistics} for the initialization ($c_{\mathrm{Tukey}}=3$).}
  \label{fig:Penaltydetectionrelation}
\end{figure}

The Type I and Type II outliers are, respectively, generated as $\Tilde{\mathbf{x}}_i^{(1)}=\mathbf{x}_{i,k}+\vartheta_{o_1}\boldsymbol{\upsilon}$ and ${\Tilde{\mathbf{x}}_i^{(2)}=\boldsymbol{\mu}_{o_2}+\vartheta_{o_2}\boldsymbol{\upsilon}}$ where $\Tilde{\mathbf{x}}_i^{(j)},j=1,2$ denotes the type of the outlier, $\vartheta_{o_j}, j=1,2$ is a scaling constant associated with the outlier type and $\boldsymbol{\mu}_{o_2}$ is a vector associated with the location of Type II outliers. A detailed explanation including all parameter values, is provided in the supplementary material Appendix~D.1.1. Examples of the eigenvector estimates based on the corrupted data set are shown for LE and RRLPI in Fig.~\ref{fig:corruptedFiedler} and Fig.~\ref{fig:corruptedRRLPI}, respectively. As can be seen, for the LE method, Type I outliers in the data produce outliers in the mapping results that obscure the underlying structure of $K=3$ clusters. In contrast, the proposed RRLPI provides a mapping that is less influenced by the outliers.

\begin{figure*}[!tbp]
    \centering
    \includegraphics[trim={2cm 11cm 2cm 10.6cm},clip,width=16cm]{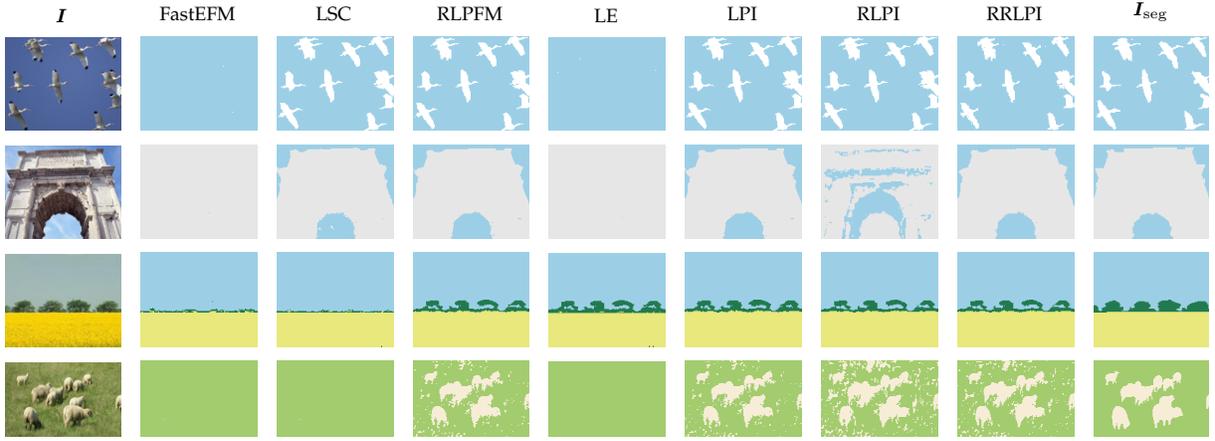}
  \caption{Image segmentation results for the original images.\label{fig:originalimagesegmentationresults}}
\vspace{-4mm}
\end{figure*}
\begin{figure*}
    \centering
\hspace{-6mm}
\subfloat[$\bm{I}$]{\includegraphics[trim={1mm 5mm 8mm 5mm},clip,width=4.13cm]{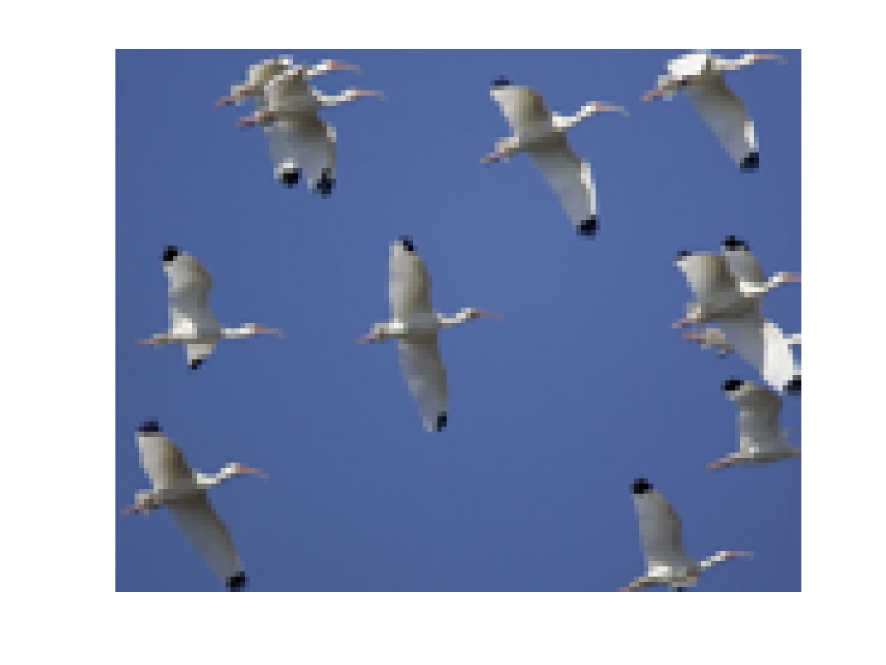}}
\subfloat[$\hat{\bm{I}}_{\mathrm{seg}}$ for LE]{\includegraphics[trim={1mm 5mm 8mm 5mm},clip,width=4.13cm]{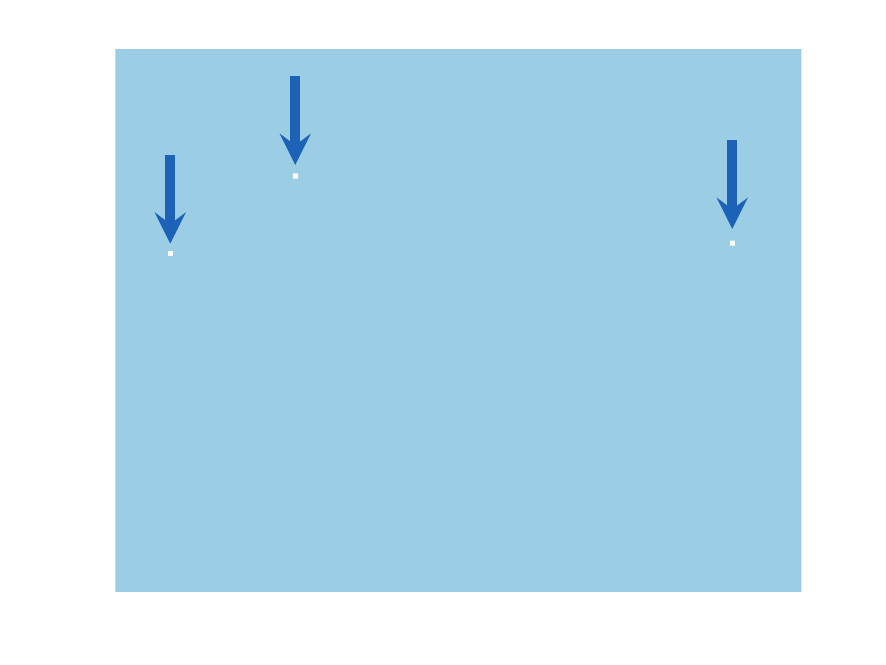}\label{fig:LESegresult}}
\subfloat[$\hat{\bm{I}}_{\mathrm{seg}}$ for RRLPI]{\includegraphics[trim={1mm 5mm 8mm 5mm},clip,width=4.13cm]{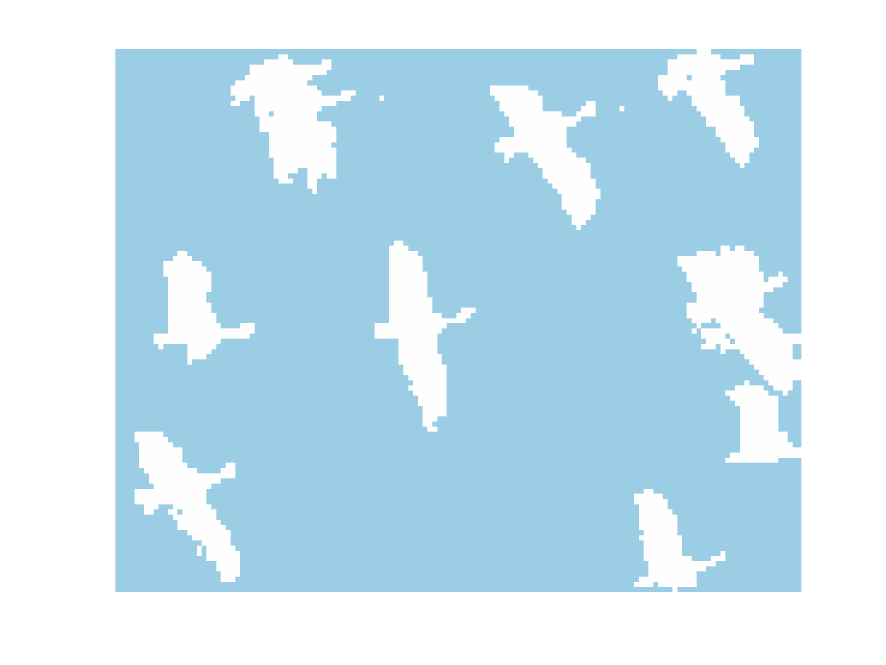}\label{fig:RRLPISegresult}}
\subfloat[Annotated image ${\bm{I}}_{\mathrm{seg}}$]{\includegraphics[trim={1mm 5mm 8mm 5mm},clip,width=4.13cm]{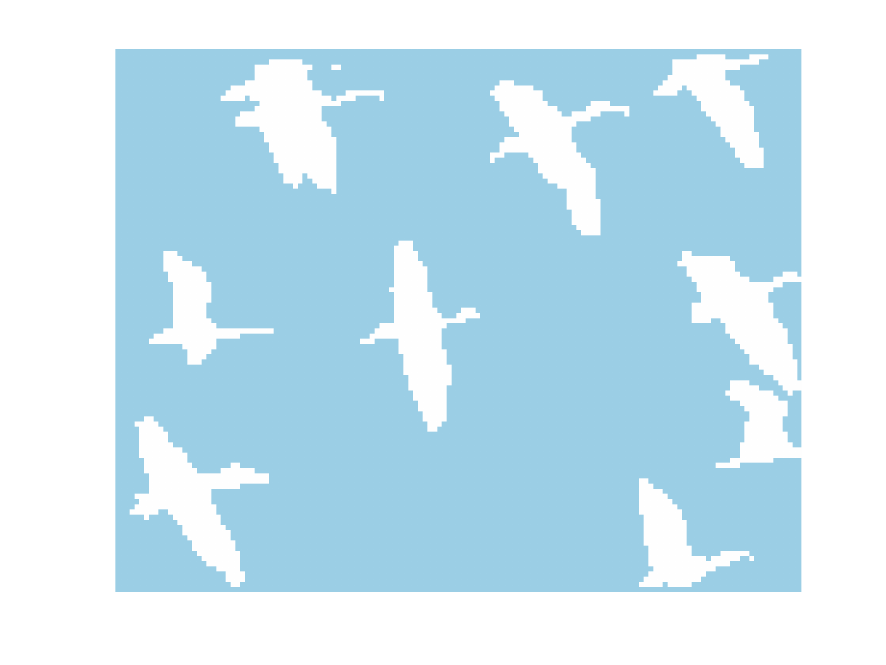}}
\caption{Example segmentations for LE and RRLPI methods. The embeddings that are mapped far away from the group of pixels are pointed out using arrows.}
\label{fig:outliereffectimage}
\vspace{-5mm}
\end{figure*}
\begin{figure}[tbp!]
  \centering
\includegraphics[trim={0mm 1mm 0mm 4mm},clip,width=7.3cm]{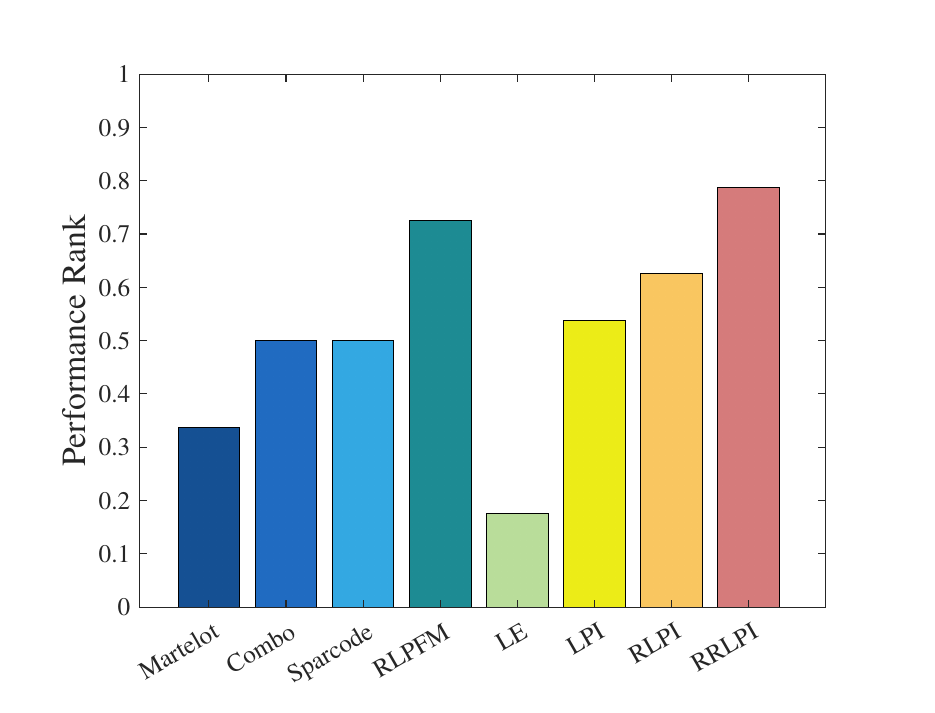}
  \caption{Performance rank for $K$-medoids partitioning with Tukey's distance function for the initialization ($c_{\mathrm{Tukey}}=3$).}
  \label{fig:Performancerank}
\end{figure}

Fig.~\ref{fig:PaccSigma} and Fig.~\ref{fig:PaccNumout} report the average partition accuracy as a function of the constant $\vartheta_{o_1}$ associated with Type I outliers and the number of outliers $N_{\mathrm{out}}$ for each outlier type, respectively. The value of $\vartheta_{o_2}$ is kept constant to generate points that lie between clusters two and three. The robust methods show best performance while the performance of LE quickly decreases in the presence of outliers.
\vspace{-2mm}
\subsection{Cluster Enumeration}
In this section, the cluster enumeration performance of different approaches is benchmarked in terms of their empirical probability of detection using the following data sets:
 
 \textit{Human Gait Data Set:} The experimental data set \cite{humangait} was collected in an office environment at Technische Universität Darmstadt using a 24 GHz radar system \cite{radarinfo}. The data set consists of 800 measurements from five different gait types measured from two different directions \cite{RadarConf}.

\textit{Breast Cancer Wisconsin Data Set:} The data set includes 569 observations from two classes that define benign or malignant tumors \cite{BreastCancer}.
    
 \textit{Iris Data Set:} The iris data set consists of 150 measurements of three different iris flower species \cite{Iris}.
 
 \textit{Person Identification Data Set:} The experimental data set \cite{PersonIdentification} was collected using the same settings as in the human gait data set. The data set includes radar observations of four different subjects walking towards and away from the radar system.
    
 \textit{Connectionist Bench Data Set (Sonar):} The data set includes sonar returns collected from a metal cylinder and a cylindric rock positioned on a sandy ocean floor \cite{Sonar}. The number of observations is equal to 208 for two object clusters.
 
\textit{Ionosphere Data Set:} The data set includes 351 radar returns from the ionosphere for two clusters \cite{Ionosphere}.
    
\textit{Diabetic Retinopathy Debrecen Data Set (D. Retinopathy):} The data set consists of 1151 observations of two clusters using image-based features of diabetic retinopathy \cite{D.Retinopathy}.

\textit{Gesture Phase Segmentation Data Set (Gesture Phase S.):} The processed features that contain scalar velocity and acceleration of hand and wrist movements have been used for videos A1, A2, A3, C1 and C3. The data set includes five phases which are rest, preparation, stroke, hold and retraction \cite{GesturePhase}.

 If none of the cluster enumeration approaches estimates the cluster number correctly with the default cosine similarity, the elastic net similarity measure as in \cite{Robuststatistics}, is used with ten candidate penalty parameters $\rho$ on an equidistant grid ranging from $0.1$ to one. Results are reported for $\rho=0.5$, which gave the best average overall detection performance for all methods. Tukey's distance function \cite{Robuststatistics} where the threshold defined as $c_{\mathrm{Tukey}}=3$ is used as an initialization for $K$-medoids partitioning in the proposed algorithm. For a detailed discussion about different partitioning results, see Appendix~D.2.1 of the supplementary material. The estimated cluster numbers are reported for the different cluster enumeration approaches in Tab.~\ref{tab:tablerealclusterenumeration}. As can be seen from the table, the human gait and gesture phase data sets include a considerable number of outliers that result in misdetection of the cluster number for almost all competitors. The proposed method is the only one that consistently estimates the correct cluster numbers for all data sets.


\begin{figure*}[!tbp]
    \centering
    \includegraphics[trim={2cm 11cm 2cm 10.6cm},clip,width=16cm]{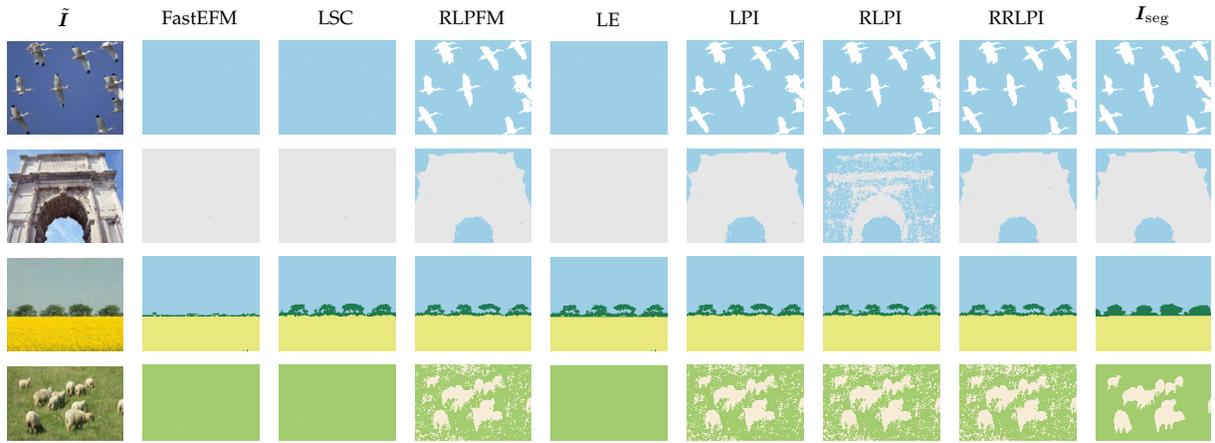}
  \caption{Image segmentation results for the corrupted images. ($\sigma^{(\xi)}=10^{-3}$)\label{fig:corruptedimagesegmentationresults}}
  \vspace{-4mm}
\end{figure*}
\begin{figure*}[!tbp]
  \centering
\subfloat[$\bar{F}_{\mathrm{score}}$ for increasing $\sigma^{(\xi)}$]{\includegraphics[trim={1mm 0mm 2mm 0mm},clip,width=6.2cm]{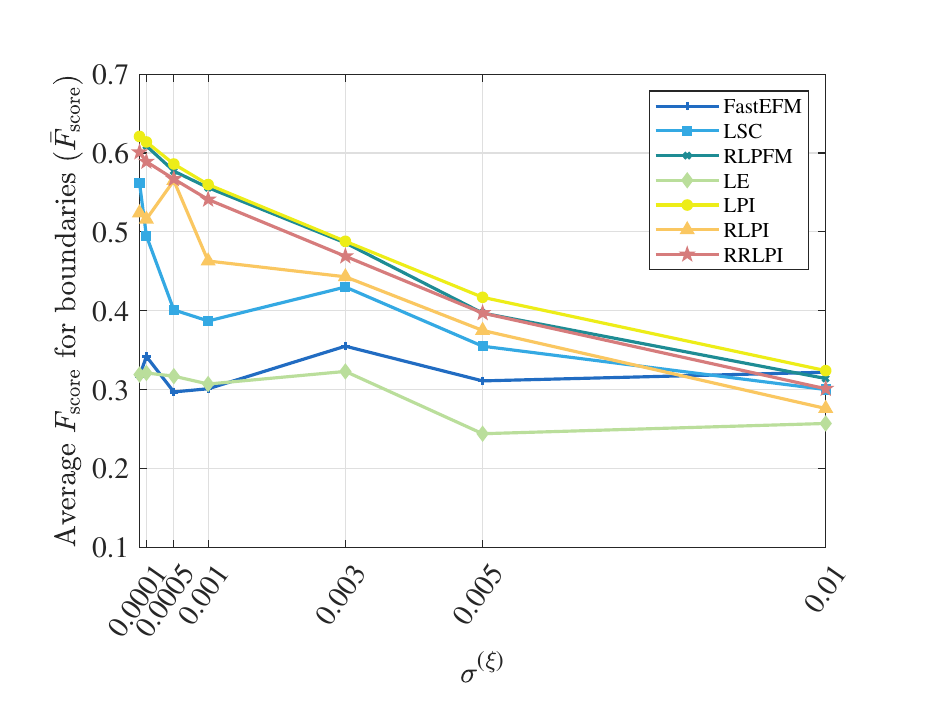}\label{fig:Fscoreresults}}
\subfloat[$\bar{J}$ for increasing $\sigma^{(\xi)}$]{\includegraphics[trim={1mm 0mm 2mm 0mm},clip,width=6.2cm]{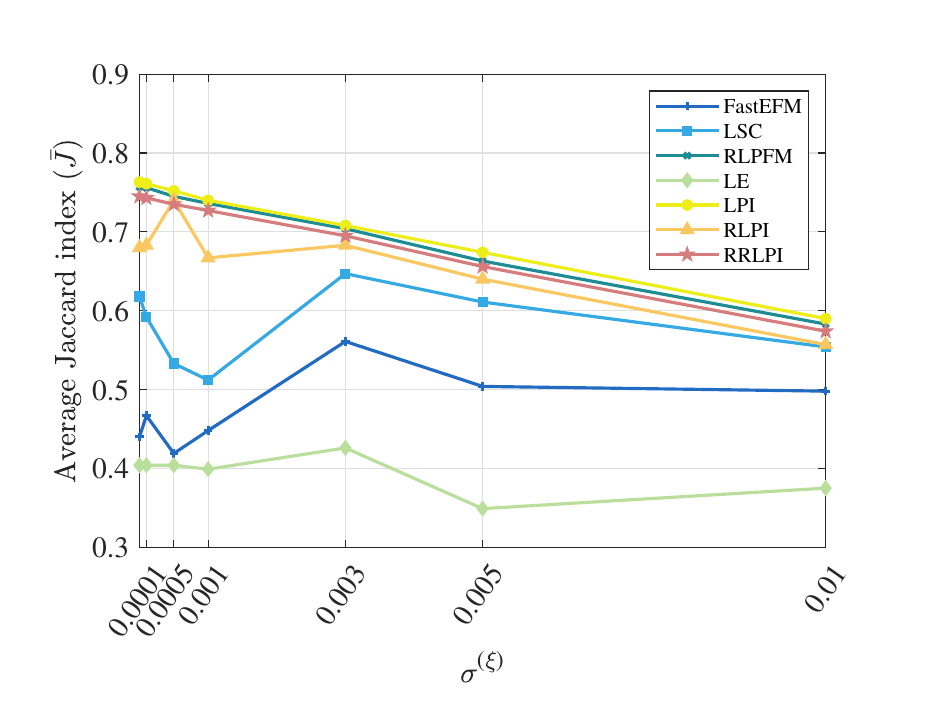}\label{fig:Jscoreresults}}
\subfloat[$\bar{p}_{\mathrm{acc}}$ for increasing $\sigma^{(\xi)}$]{\includegraphics[trim={1mm 0mm 2mm 0mm},clip,width=6.2cm]{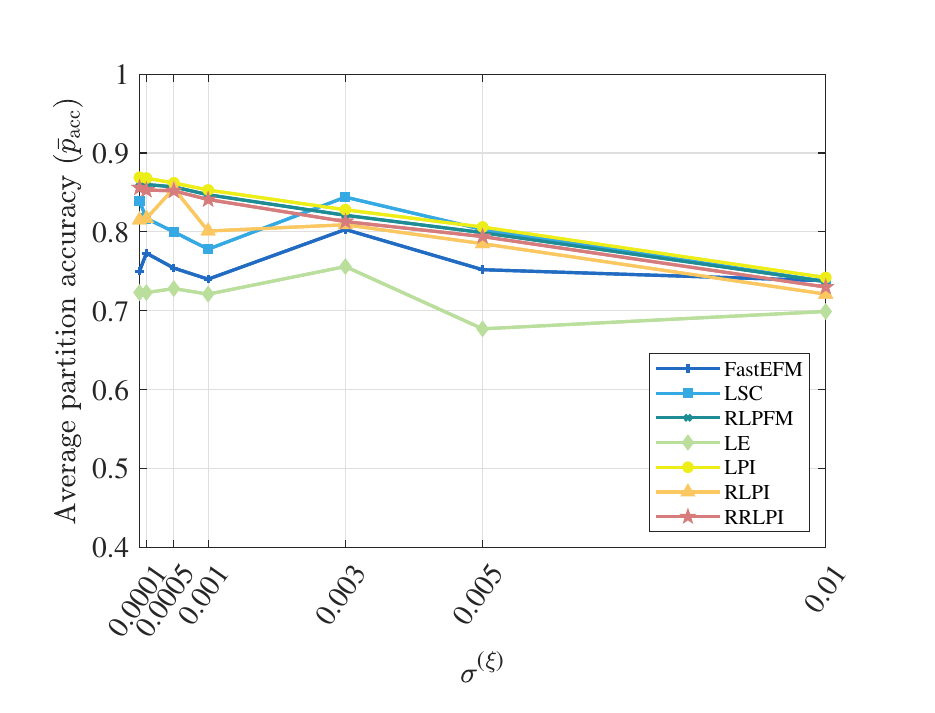}\label{fig:Paccresults}}
  \caption{Numerical results for the image segmentation.  \label{fig:imagesegmentationnumericalresults}}
\end{figure*}

 The empirical probability of detection with respect to different penalty parameters is detailed in Fig. ~\ref{fig:Penaltydetectionrelation}. Then, the performance is summarized in Fig.~\ref{fig:Performancerank} by averaging the results over all penalty parameters. 
The results for cluster enumeration demonstrate that the proposed RRLPI shows the best probability of detection performance for all candidate penalty parameters with an average score of 79~$\%$, whereas the best competitors (RLPFM and RLPI) have scores of 73 $\%$ and 63 $\%$, respectively.

\vspace{-3mm}
\subsection{Image Segmentation}
ADE20K \cite{ADE20K}, is a large-scale dataset that includes high
quality pixel-level annotations of 25210 images (20210, 2000,
and 3000 for the training, validation, and test sets, respectively.). In our experiments, 10 images from the ADE20K data set containing different objects, where each object has a different color, have been selected for color-based image segmentation. The selected and corresponding annotated images are denoted as $\bm{I}$ and $\bm{I}_{\mathrm{seg}}$, respectively. The images are down-sampled, where the dimension of data set $\mathbf{X}$ is $m=3$ and $n\cong15000$ using RGB color codes associated to down-sampled image pixels as features. To analyze robustness, the images are corrupted by adding multiplicative noise using the equation $\Tilde{\bm{I}} = \bm{I}+\xi\times\bm{I}$, where $\Tilde{\bm{I}}$ denotes the corrupted image and $\xi$ is uniformly distributed random noise with zero mean and variance $\sigma^{(\xi)}$.

The down-sampled images are segmented for a pre-defined number of segments $K$ using the default setting which performs $K$-means partitioning for the data sets that have more than $n=3000$ samples. In Fig.~\ref{fig:originalimagesegmentationresults}, examples of the original image $\bm{I}$ and associated segmented images using the estimated Fiedler vectors for seven different embedding approaches are shown along with the ground truth segmented image $\bm{I}_\mathrm{seg}$. The uncorrupted images $\bm{I}$ may also contain outlying pixels and/or noisy pixels. The effect of outliers is that a small number of pixels are mapped far away from the group of pixels and, thus, the remaining group of pixels assigned to a single large segment based on the distance-based partitioning methods. 

A typical example of a segmentation result illustrating the outlier effects is provided
in Fig.~\ref{fig:LESegresult}. As can be seen, the described outlier effect is observed even for the mappings of the uncorrupted (original) image when using LE. To exemplify the robust Fiedler vector estimation, the segmentation result of RRLPI is shown in Fig.~\ref{fig:RRLPISegresult}. The segmentation result demonstrates that the proposed robust Fiedler vector estimation suppress outlier effects on the eigen-decomposition and provides segmentation results that are more consistent with the annotated image $\bm{I}_{\mathrm{seg}}$. Further, in Fig.~\ref{fig:corruptedimagesegmentationresults}, examples of segmented images are presented for the corrupted images where ${\sigma^{(\xi)}=10^{-3}}$. The results show that the outlier effect on eigen-decomposition causes a breakdown of the FastEFM, LSC, and LE approaches. For further examples and detailed numerical results, see supplementary material Appendix~D.3.2.

The experiments are evaluated quantitatively using  $\bar{F}_\mathrm{score}$, $\bar{J}$ and $\bar{p}_{\mathrm{acc}}$, and the results are summarized in Fig.~\ref{fig:imagesegmentationnumericalresults}. All performance measures are evaluated by comparing each estimated segmented image $\hat{\bm{I}}_{\mathrm{seg}}$ with the annotated image $\bm{I}_{\mathrm{seg}}$. The LE and FastEFM show poor performance, even for the original images. Although LSC shows a reasonably good performance for the original images, its performance reduces drastically in the outlier-corrupted case in terms of $\bar{F}_\mathrm{score}$ and $\bar{J}$. The LPI, RLPFM and RRLPI are the top three methods in all performance measures and RLPI follows them with a reasonably good performance, which indicates that the proposed penalty parameter selection algorithm is a promising approach, even when using non-robust methods.

\section{Conclusion}\label{sec:Conclusion}

The effect of outliers on eigen-decomposition has been analyzed and the importance of overall edge weight attached to a vertex has been shown to be a useful measure of outlyingness. Based on the derived theoretical results, we proposed RRLPI, a method to robustly estimate the Fiedler vector that down-weights mappings, for which the overall edge weight deviates from the typical overall edge weight of a given graph. The objective function to estimate the Fiedler vector is penalized using the proposed unsupervised penalty parameter selection algorithm that builds upon $\Delta$-separated sets. The performance of RRLPI is benchmarked for different applications on a variety of real-world data sets. The numerical results for cluster analysis and image segmentation showed that the RRLPI is a promising approach for Fiedler vector estimation in situations where robustly determining the group structure in a data set is essential.
\vspace{-3mm}
\section*{Acknowledgement}
The authors would like to thank Ayfer Taştan
for checking the eigen-decompositions. The work of A. Taştan is supported by the Republic of Turkey Ministry of National Education. The work of M. Muma has been funded by the LOEWE initiative (Hesse, 
Germany) within the emergenCITY centre and is supported by the ‘Athene Young
Investigator Programme’ of Technische Universität Darmstadt, Hesse, 
Germany.

\begin{IEEEbiography}[{\includegraphics[trim={0 0mm 5mm 0mm},width=1in,height=1.25in,clip,keepaspectratio]{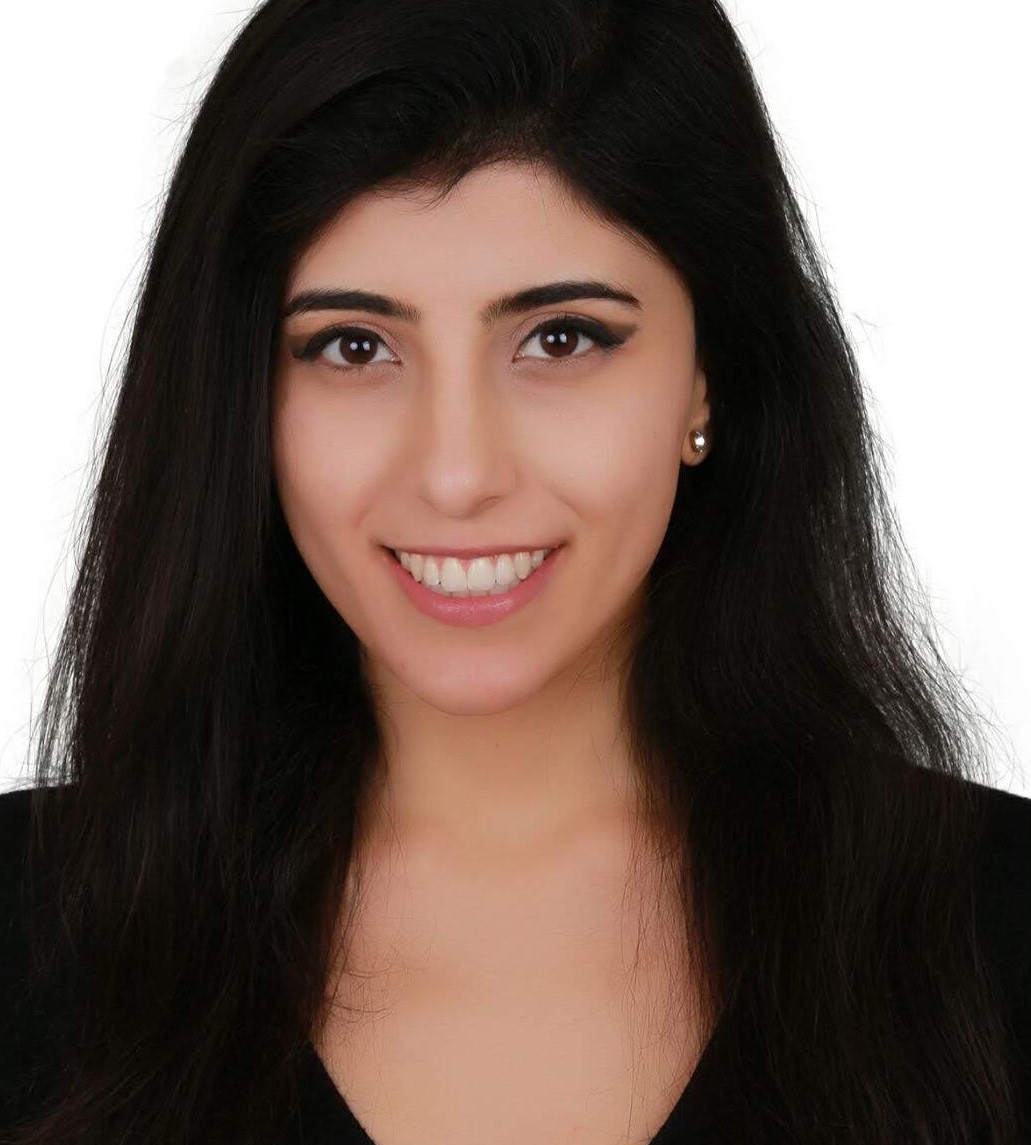}}]{Aylin Ta{\c{s}}tan}
(S’20) received the
B.Sc. and the M.Sc. degrees in electrical and electronics engineering from Gazi University, Ankara, Turkey, in 2016 and 2018, respectively. 
In 2019, she joined the Signal Processing Group, Technische Universitat Darmstadt, Germany, as a Research Associate, where she is currently working toward the Ph.D. degree in electrical engineering and information technology. In September 2020, together with her coauthors, she received the Best Student Paper Award at the 2020 IEEE Radar Conference for the paper entitled “An Unsupervised Approach for Graph-based Robust Clustering of Human Gait Signatures”. Her research interests include robust statistics for graph theory, statistical signal processing, unsupervised learning and robust clustering techniques.
\end{IEEEbiography}

\begin{IEEEbiography}[{\includegraphics[width=1in,height=1.25in,clip,keepaspectratio]{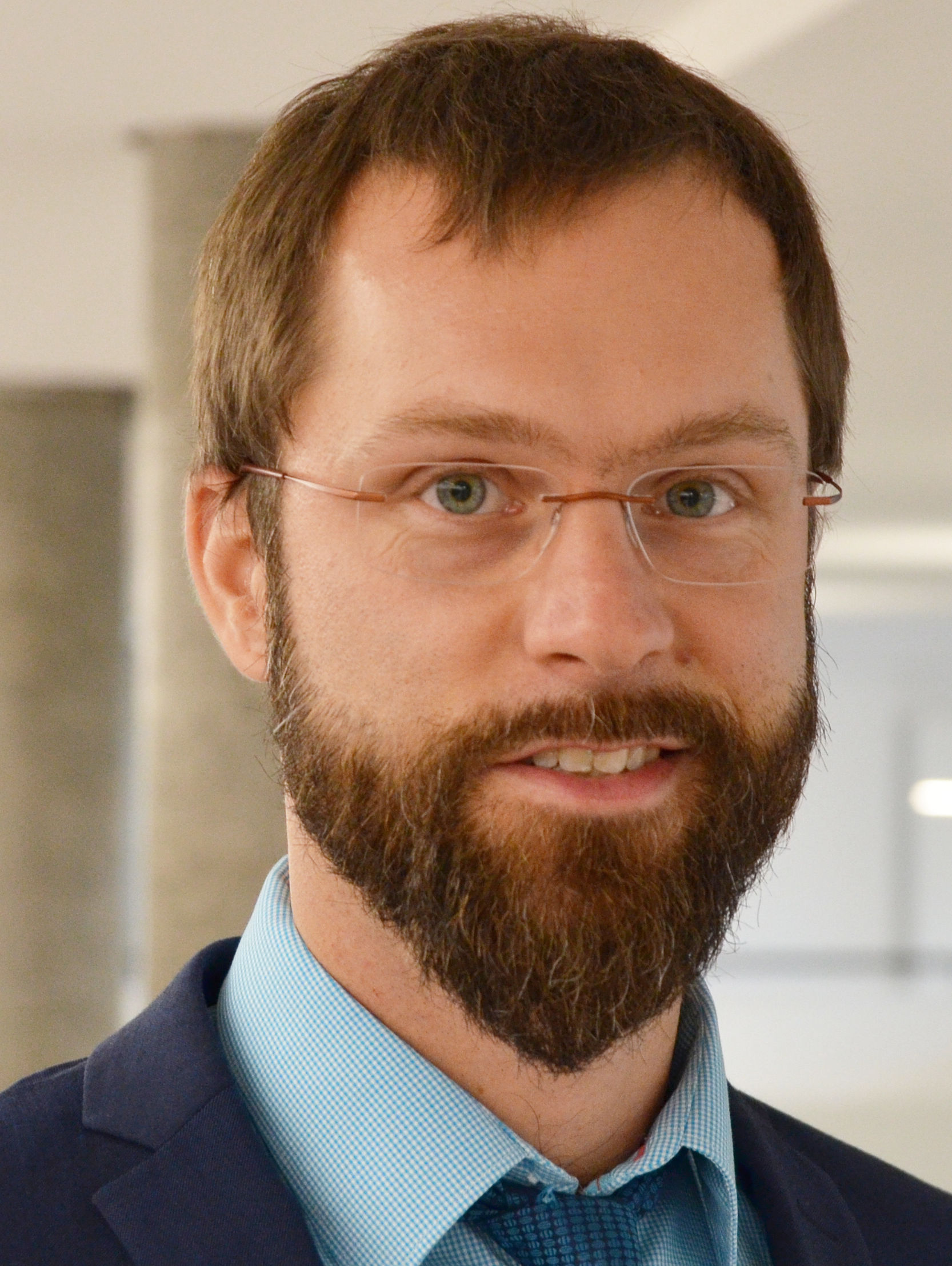}}]{Michael Muma}  (S’10–M’14) received the Dipl.-Ing (2009) and the Dr.-Ing. (summa cum laude, 2014) degree in electrical engineering and information technology both from Technische Universität Darmstadt, Darmstadt, Germany. He is currently an Independent Junior Research Group Leader at TU Darmstadt, and a scientist at the interdisciplinary and multi-site research cooperation emergenCITY. He is devoted to developing new data science approaches, theory and tools for robust signal processing and statistical learning with applications in biomedicine, radar, and sensor networks. For his contributions to robust signal processing and statistical learning, he received the prestigious EURASIP Early Career Award (2021). He was co-author of the 2020 IEEE Radar Conference Student Contest Best Paper Award. Since 2019, he is Associate Editor for IEEE Transactions on Signal Processing, and Elected Member EURASIP Technical Area Committee on Theoretical and Methodological Trends in Signal Processing (TMTSP). Together with his coauthors, he received the 2017 IEEE Signal Processing Magazine Best Paper Award for the paper entitled “Robust Estimation in Signal Processing: A Tutorial-Style Treatment of Fundamental Concepts.” In 2017, he also received the Athene Young Investigator Award for Exceptionally Qualified Early Career Researchers, and in 2015, he was the Supervisor of the Technische Universität Darmstadt student team who won the international IEEE Signal Processing Cup.
\end{IEEEbiography}


\begin{IEEEbiography}[{\includegraphics[width=1in,height=1.25in,clip,keepaspectratio]{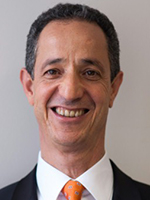}}]{Abdelhak M. Zoubir} (S’87–M’92–SM’97–F’08) received the Dr.-Ing. degree from Ruhr-Universität
Bochum, Bochum, Germany, in 1992. From 1992 to
1998, he was an Associate Professor with the Queensland University of Technology, Brisbane, QLD,
Australia. In 1999, he joined the Curtin University
of Technology, Bentley, WA, Australia, as a Professor of Telecommunications and was the Interim Head
of the School of Electrical and Computer Engineering from 2001 to 2003. In 2003, he was a Professor in signal processing and the Head of the Signal Processing Group, Technische Universität Darmstadt, Darmstadt, Germany, where he is currently the Head of the Department of Electrical Engineering and Information Technology (2020-2022). His research interests include statistical
methods for signal processing with emphasis on bootstrap techniques, robust
detection and estimation and array processing applied to telecommunications,
radar, sonar, car engine monitoring and safety, and biomedicine. He authored or
coauthored over 400 journal and conference papers on these areas. He was the
Technical Chair of the 11th IEEE Workshop on Statistical Signal Processing in
2001, General Co-Chair of the 3rd IEEE International Symposium on Signal
Processing and Information Technology in 2003, and of the 5th IEEE Workshop on Sensor Array and Multichannel Signal Processing in 2008. He was the
General Co-Chair of the European Conference on Signal Processing in 2013
in Marrakesh (Morocco), and the 14th IEEE Workshop on Signal Processing
Advances in Wireless Communications in 2013 in Darmstadt, Germany. He
was the Technical Co-Chair of ICASSP-14, held in Florence. He was an Associate Editor for the IEEE TRANSACTIONS ON SIGNAL PROCESSING from 1999 to 2005 and a Member of the Senior Editorial Board for the IEEE JOURNAL ON
SELECTED TOPICS IN SIGNAL PROCESSING from 2009 to 2011. From 2005 to 2013, he served on the Editorial Board of the EURASIP Journal Signal Processing, and was the Editor-in-Chief for the IEEE SIGNAL PROCESSING MAGAZINE (2012–2014). He was the Chair (2010–2011), the Vice-Chair (2008–2009), and a member (2002–2008) of the IEEE SPS Technical Committee Signal Processing Theory and Methods, and a member (2007–2012) of the IEEE SPS Technical Committee Sensor Array and Multichannel Signal Processing. He has served on the Board of Governors of the IEEE SPS as an elected Member-at-Large (2015–2017), and was an elected Member of the Board of Directors of the EURASIP from 2008 to 2016. He was the President of EURASIP (2017–2018), and a Member of the IEEE Fellow Evaluation Committee of the IEEE SPS (2018-2020). He is currently an elected Member of the Board of Directors of the EURASIP (2019-2021). He is an IEEE Distinguished
Lecturer (Class 2010–2011).
\end{IEEEbiography}






\end{document}


%
\title{Supplementary Information: Robust Regularized Locality Preserving Indexing for Fiedler Vector Estimation}

\author{Aylin~Ta{\c{s}}tan,~\IEEEmembership{Student Member,~IEEE,}
        Michael~Muma,~\IEEEmembership{Member,~IEEE,}
        and~Abdelhak~M.~Zoubir,~\IEEEmembership{Fellow,~IEEE}
\thanks{The authors are with the Signal Processing Group, Technische Universität
Darmstadt, Darmstadt, Germany (e-mail: atastan@spg.tu-darmstadt.de; muma@spg.tu-darmstadt.de; zoubir@spg.tu-darmstadt.de).\protect\\
}
}

\markboth{}%
{Shell \MakeLowercase{\textit{et al.}}: Robust Regularized Locality Preserving Indexing for Fiedler Vector Estimation}

\maketitle

\IEEEdisplaynontitleabstractindextext
\IEEEpeerreviewmaketitle

\ifCLASSOPTIONcompsoc
\IEEEraisesectionheading{\section{Structure}\label{sec:introduction}}
\else
\section{Structure}
\label{sec:introduction}
\fi

The Supplementary Information for the paper Robust Regularized Locality Preserving Indexing for Fiedler Vector Estimation is organized as follows: In Appendix A and B, the theorems
by reference to outlier effects on eigen-decomposition and robust Fiedler vector estimation are proved, respectively. Appendix C contains the detailed eigen-decompositions used to analyse the effects of outliers and in Appendix D additional experimental results are provided to show the outlier effects on eigen-decompositon, cluster enumeration and image segmentation, respectively.
\setcounter{secnumdepth}{0}
\vspace{-1.5mm}
\section{Appendix A~: Outlier Effects on Eigen-decomposition}
\subsection{A.1~Proof of Theorem 2}\label{app:ProofTheorem2}
\vspace{-1mm}
\begin{proof}[\unskip\nopunct]
Suppose that $\mathbf{W}\in\mathbb{R}^{n\times n}$ denotes a symmetric block diagonal affinity matrix that can be written for $k$ blocks as
\vspace{-1mm}
\begin{align*}
\begin{bmatrix}
\mathbf{W}_1 & &  \dots &0   \\
 & \mathbf{W}_2 &  & \vphantom{\int\limits^x}\smash{\vdots} \\
\vphantom{\int\limits^x}\smash{\vdots} &  & \ddots  & \vphantom{\int\limits^x}\smash{\vdots} \\
 0&\dots  &   & \mathbf{W}_k \\
\end{bmatrix} 
\end{align*}
where all values outside the block matrices $\mathbf{W}_1,\dots,\mathbf{W}_k$ are equal to zero. Accordingly, the Laplacian is also a block diagonal matrix with blocks $\mathbf{L}_1,\dots,\mathbf{L}_k$ on the diagonal. To study the effects of outlier contamination on eigen-decomposition, let a single outlier, which does not have considerable correlations with any of the feature vectors, be added to the affinity matrix. The corresponding corrupted affinity matrix, weight matrix and Laplacian matrix are denoted as ${\Tilde{\mathbf{W}}\in\mathbb{R}^{(n+1)\times (n+1)}}$, ${\Tilde{\mathbf{D}}\in\mathbb{R}^{(n+1)\times (n+1)}}$ and ${\Tilde{\mathbf{L}}\in\mathbb{R}^{(n+1)\times (n+1)}}$, respectively. If the outlier does not correlate with any of the blocks, it can be written as a single-element block $\tilde{w}_0$, i.e.,
\vspace{-1mm}
\begin{align*}
\begin{bmatrix}
\mathbf{W}_1 &  & & \dots  & 0  \\
 &\mathbf{W}_2 &  & & \vphantom{\int\limits^x}\smash{\vdots} \\
 & & \ddots  &   &  \\
\vphantom{\int\limits^x}\smash{\vdots}&&   &\mathbf{W}_k &  \\
0 &\dots & &  & \Tilde{w}_o \\
\end{bmatrix}
\end{align*}
Similarly, the corrupted Laplacian $\tilde{\mathbf{L}}$ has blocks $\mathbf{L}_1,\dots,\mathbf{L}_k$ and $\Tilde{l}_o$, respectively, on the diagonal. Using determinant properties of block matrices, it follows that
\begin{align*}
    \mathrm{det}(\Tilde{\mathbf{L}}\hspace{-0.4mm}-\hspace{-0.4mm}\Tilde{\lambda}\Tilde{\mathbf{D}})\hspace{-0.8mm}=\hspace{-0.8mm}\mathrm{det}(\mathbf{L}_1\hspace{-0.4mm}-\hspace{-0.4mm}\lambda\mathbf{D}_1)\hspace{-0.5mm}\dots\hspace{-0.5mm}\mathrm{det}(\mathbf{L}_k\hspace{-0.4mm}-\hspace{-0.4mm}\lambda\mathbf{D}_k)\mathrm{det}(\Tilde{l}_o\hspace{-0.4mm}-\hspace{-0.4mm}\Tilde{\lambda} \Tilde{d}_o),
\end{align*}
where $\Tilde{d}_o$ denotes overall edge weight attached to the outlier for ${\Tilde{d}_{o}=\sum_j \Tilde{w}_{j,o}}$, and $\mathbf{L}_k$ and $\mathbf{D}_k$ are the Laplacian and the edge weight matrix of the $k$th block, respectively. Because the affinity matrix $\Tilde{\mathbf{W}}$ has eigenvalues ${0\leq\lambda_0\leq \lambda_1\leq\dots\leq\lambda_{n-1}}$ for blocks $\mathbf{L}_1,\dots,\mathbf{L}_k$ on the diagonal, the expression $\mathrm{det}(\Tilde{l}_o-\Tilde{\lambda} \Tilde{d}_o)=0$ yields
 \begin{align*}
\Tilde{d}_o-\Tilde{w}_o-\Tilde{\lambda} \Tilde{d}_o&=0\\
\Tilde{d}_o(1-\Tilde{\lambda})&=\Tilde{w}_o
 \end{align*}
 and, for $\Tilde{d}_o\approx\Tilde{w}_o$, it follows that $\Tilde{\lambda}=0$.
This result is straight-forwardly generalized to the case of adding $N_o$ outliers with negligibly small correlation coefficients, which results in $N_o$ single outlier blocks and $N_o$ zero eigenvalues.
\vspace{0.5mm}
\end{proof}

\subsection{A.2~Proof of Corollary 2.1}\label{app:ProofCorollary2.1}
\vspace{-1mm}
\begin{proof}[\unskip\nopunct]
Consider an affinity matrix $\mathbf{W}\in\mathbb{R}^{n\times n}$ with eigenvalues ${0\leq\lambda_0\leq \lambda_1\leq\dots\leq\lambda_{n-1}}$. Assuming that a single Type I outlier is added to $\mathbf{W}$, the corrupted affinity matrix ${\Tilde{\mathbf{W}}\in\mathbb{R}^{(n+1)\times (n+1)}}$ is written as two blocks
\begin{align*}
\begin{bmatrix}
\mathbf{W} & 0 \\
 0 & \Tilde{w}_o \\
\end{bmatrix} 
\end{align*}
where $\mathbf{W}$ is the first block that is equal to an initial affinity matrix and $\Tilde{w}_o$ denotes single element block of the outlier. The corrupted Laplacian $\Tilde{\mathbf{L}}\in\mathbb{R}^{(n+1)\times (n+1)}$ of such a matrix also has two blocks, i.e. $\mathbf{L}$ and $\Tilde{l}_o$. The first block has eigenvalues ${0\leq\lambda_0\leq \lambda_1\leq\dots\leq\lambda_{n-1}}$. Because $\mathrm{det}(\Tilde{l}_o-\Tilde{\lambda} \Tilde{d}_o)$, gives an additional zero eigenvalue.
\end{proof}

\subsection{A.3~Proof of Corollary 2.2}\label{app:ProofCorollary2.2}
\vspace{-1mm}
\begin{proof}[\unskip\nopunct]
For simplicity, let $\mathbf{W}\in\mathbb{R}^{n\times n}$ and the corresponding Laplacian $\mathbf{L}\in\mathbb{R}^{n\times n}$ consist of $k=2$ blocks.  To estimate eigenvectors, the $\mathbf{y}_i\in\{\mathbf{y}_0,\dots,\mathbf{y}_3  \}$, where $\mathbf{y}_i$ is the eigenvector associated with the $i$th eigenvalue $\lambda_i$ for $\mathbf{L}$, we consider the eigen-decompositions
$\mathbf{L}\mathbf{y}_i=\lambda_i\mathbf{y}_i$ and $\mathbf{L}\mathbf{y}_i=\lambda_i\mathbf{D}\mathbf{y}_i$ whose corresponding eigenvectors are equivalent and can be written as\footnote{For a discussion, see Sections~C.1 and C.3.}
\begin{align*}
\begin{bmatrix}
y_{1,0}\\
y_{1,0}\\
y_{3,0}\\
y_{3,0}\\
\end{bmatrix}\hspace{5mm}
\begin{bmatrix}
y_{1,1}\\
y_{1,1}\\
y_{3,1}\\
y_{3,1}\\
\end{bmatrix}\hspace{5mm}
\begin{bmatrix}
\quad y_{1,2}\\
-y_{1,2}\\
\quad y_{3,2}\\
-y_{3,2}\\
\end{bmatrix}\hspace{5mm}
\begin{bmatrix}
\quad y_{1,3}\\
-y_{1,3}\\
\quad y_{3,3}\\
-y_{3,3}\\
\end{bmatrix}
\end{align*}
where $y_{j,i}$ denotes the $j$th embedding results in the $i$th eigenvector. Adding a single outlier to the affinity matrix $\mathbf{W}$ leads to the corrupted affinity matrix, weight matrix and Laplacian matrix denoted as ${\Tilde{\mathbf{W}}\in\mathbb{R}^{(n+1)\times (n+1)}}$, ${\Tilde{\mathbf{D}}\in\mathbb{R}^{(n+1)\times (n+1)}}$ and ${\Tilde{\mathbf{L}}\in\mathbb{R}^{(n+1)\times (n+1)}}$, respectively. The eigenvalues of  ${\Tilde{\mathbf{L}}}$ for the standard eigen-decomposition $\mathbf{L}\mathbf{y}_{i}=\lambda_{i}\mathbf{y}_{i}$ are $\Tilde{\lambda}_{0,1,2}=0$, $\Tilde{\lambda}_{3}=2w_1$ and $\Tilde{\lambda}_{4}=2w_2$; and for the generalized eigen-decomposition $\mathbf{L}\mathbf{y}_{i}=\lambda_{i}\mathbf{D}\mathbf{y}_{i}$ they are $\Tilde{\lambda}_{0,1}=0$, $\Tilde{\lambda}_{2,3}=2$ and $\Tilde{\lambda}_{4}=\lambda$ where $\lambda$ can take any value if $\Tilde{w}_o=0$. \color{black}Even though the objective function results in different eigenvalues, the eigenvectors are equivalent for both eigen-decompositions
\begin{align*}
\begin{bmatrix}
\Tilde{y}_{1,0}\\
\Tilde{y}_{1,0}\\
\Tilde{y}_{3,0}\\
\Tilde{y}_{3,0}\\
\Tilde{y}_{o,0}\\
\end{bmatrix}\hspace{5mm}
\begin{bmatrix}
\Tilde{y}_{1,1}\\
\Tilde{y}_{1,1}\\
\Tilde{y}_{3,1}\\
\Tilde{y}_{3,1}\\
\Tilde{y}_{o,1}\\
\end{bmatrix}\hspace{5mm}
\begin{bmatrix}
\Tilde{y}_{1,2}\\
\Tilde{y}_{1,2}\\
\Tilde{y}_{3,2}\\
\Tilde{y}_{3,2}\\
\Tilde{y}_{o,2}\\
\end{bmatrix}\hspace{5mm}
\begin{bmatrix}
\quad \Tilde{y}_{1,3}\\
-\Tilde{y}_{1,3}\\
\quad \Tilde{y}_{3,3}\\
-\Tilde{y}_{3,3}\\
\quad \Tilde{y}_{o,3}\\
\end{bmatrix}\hspace{5mm}
\begin{bmatrix}
\quad \Tilde{y}_{1,4}\\
-\Tilde{y}_{1,4}\\
\quad \Tilde{y}_{3,4}\\
-\Tilde{y}_{3,4}\\
\quad \Tilde{y}_{o,4}\\
\end{bmatrix}
\end{align*}
where $\Tilde{y}_{j,i}$ denotes the $j$th  embedding result in the $i$th  eigenvector and $\Tilde{y}_{o,i}$ is the embedding result of the outlier in the $i$th eigenvector of the corrupted Laplacian matrix $\Tilde{\mathbf{L}}$. As can be seen when ${| \Tilde{y}_{o,i}|\rightarrow 1}$, the remaining  embedding results associated with different blocks become small-valued to satisfy $\|\Tilde{\mathbf{y}}_{i}\|_2^2=1$. This leads to a decrease in the $\ell_2$ norm between mapping points associated with different blocks based on the first $k$ eigenvectors.
\end{proof}
\vspace{-5mm}
\subsection{A.4~Proof of Theorem 3}\label{app:ProofTheorem3}
\vspace{-1mm}
\begin{proof}[\unskip\nopunct]
Let $\mathbf{W}\in\mathbb{R}^{n\times n}$ denote a block zero-diagonal symmetric
affinity matrix for $k$ blocks which does not include correlations between different blocks, i.e., 
\begin{align*}
\begin{bmatrix}
0 & w_1 &\dots  &  &  & &\\
w_1 & 0 &\dots  &  &  & &\\
\vphantom{\int\limits^x}\smash{\vdots} &  \vphantom{\int\limits^x}\smash{\vdots}& 0 & w_2 &  & &\\
 &  & w_2  & 0& &  & \\
  &   &   &  & \ddots &\vphantom{\int\limits^x}\smash{\vdots} &\vphantom{\int\limits^x}\smash{\vdots}\\
  &  &   &  &\dots  &0 &w_k\\
  &   &   &  &\dots &w_k &0\\
\end{bmatrix}
\end{align*}
The associated Laplacian has the same block structure for $\mathbf{L}_1,\dots,\mathbf{L}_k$. To estimate eigenvalues of $\mathbf{L}$, $\mathrm{det}(\mathbf{L}-\lambda\mathbf{D})$ can be written in terms of the distinct determinants as
\begin{align*}
    \mathrm{det}(\mathbf{L}-\lambda\mathbf{D})=\mathrm{det}(\mathbf{L}_1-\lambda\mathbf{D}_1)\dots\mathrm{det}(\mathbf{L}_k-\lambda\mathbf{D}_k)
\end{align*}
where $\mathbf{L}_k$ and $\mathbf{D}_k$, respectively,  denote the Laplacian and the weight matrix of the $k$th block. After solving this equation for the $i$th block, the eigenvalues are equal to ${\lambda_0=0}$ and ${\lambda_1=2}$. Now assume that the first and second blocks are correlated because of Type II outliers that have been added to the data set. Because the correlations are the effect of the outliers, we call them undesired correlations. The resulting corrupted affinity matrix $\Tilde{\mathbf{W}}\in\mathbb{R}^{n\times n}$ consists of $k$ blocks, i.e.,
\vspace{-1mm}
\begin{align*}
\begin{bmatrix}
0 & w_1 & \Tilde{w}_u &\Tilde{w}_u  &\dots  & &\\
w_1 & 0 & \Tilde{w}_u& \Tilde{w}_u &\dots   & &\\
\Tilde{w}_u & \Tilde{w}_u & 0 & w_2 &  & &\\
\Tilde{w}_u & \Tilde{w}_u & w_2  & 0& &  & \\
\vphantom{\int\limits^x}\smash{\vdots}  & \vphantom{\int\limits^x}\smash{\vdots}  &   &  & \ddots &\vphantom{\int\limits^x}\smash{\vdots} &\vphantom{\int\limits^x}\smash{\vdots}\\
  &  &   &  &\dots  &0 &w_k\\
  &   &   &  &\dots &w_k &0\\
\end{bmatrix}
\end{align*}
where $\Tilde{w}_u$ denotes correlation coefficient belonging to the undesired correlations for $\Tilde{w}_u\neq 0$. For simplicity, let $\Tilde{\mathbf{L}}_{12}$ and $\Tilde{\mathbf{D}}_{12}$ be the Laplacian and the weight matrix of correlated blocks, respectively. Hence, the eigenvalues of the corrupted Laplacian $\Tilde{\mathbf{L}}$ are estimated using distinct determinants as
\begin{align*}
    \mathrm{det}(\Tilde{\mathbf{L}}-\Tilde{\lambda}\Tilde{\mathbf{D}})=\mathrm{det}(\Tilde{\mathbf{L}}_{12}-\Tilde{\lambda}\Tilde{\mathbf{D}}_{12})\dots\mathrm{det}(\mathbf{L}_k-\lambda\mathbf{D}_k)
\end{align*}
where $\Tilde{\mathbf{L}}_{12}$ and $\Tilde{\mathbf{D}}_{12}$, respectively, denote the Laplacian and the weight matrix associated with the block that is a combination of blocks $i=1$ and $j=2$. To study the effect of outliers on the eigen-decomposition, the above equation shows that analyzing ${\mathrm{det}(\Tilde{\mathbf{L}}_{12}-\Tilde{\lambda}\Tilde{\mathbf{D}}_{12})=0}$ is sufficient, which can equivalently be written as
\begin{align*}
\begin{vmatrix}
w_1-w_1\Tilde{\lambda}+c & -w_1 & -\Tilde{w}_u & -\Tilde{w}_u\\
-w_1 &w_1-w_1\Tilde{\lambda}+c & -\Tilde{w}_u & -\Tilde{w}_u\\
-\Tilde{w}_u & -\Tilde{w}_u & w_2-w_2\Tilde{\lambda}+c & -w_2\\
-\Tilde{w}_u & -\Tilde{w}_u & -w_2 &w_2-w_2\Tilde{\lambda}+c\\
\end{vmatrix}=\hspace{-0.5mm}0 
\end{align*}
where $c=2\Tilde{w}_u-2\Tilde{w}_u\Tilde{\lambda}$. To simplify the determinant problem for block matrices as in \cite{detblockmatrices}, the matrix $\Tilde{\mathbf{L}}_{12}-\Tilde{\lambda}\Tilde{\mathbf{D}}_{12}$ must be commutative such that ${(-\Tilde{\mathbf{W}}_{u})(\Tilde{\mathbf{L}}_{2}-\Tilde{\lambda}\Tilde{\mathbf{D}}_2)=(\Tilde{\mathbf{L}}_{2}-\Tilde{\lambda}\Tilde{\mathbf{D}}_2)(-\Tilde{\mathbf{W}}_{u})}$. Accordingly, the matrix $\Tilde{\mathbf{L}}_{12}-\Tilde{\lambda}\Tilde{\mathbf{D}}_{12}$ is commutative if it satisfies 
\begin{align*}
\begin{bmatrix}
-\Tilde{w}_u & -\Tilde{w}_u \\
-\Tilde{w}_u & -\Tilde{w}_u \\
\end{bmatrix}
\begin{bmatrix}
z_2 & -w_2 \\
-w_2 & z_2 \\
\end{bmatrix}=
\begin{bmatrix}
z_2 & -w_2 \\
-w_2 & z_2 \\
\end{bmatrix}
\begin{bmatrix}
-\Tilde{w}_u & -\Tilde{w}_u \\
-\Tilde{w}_u & -\Tilde{w}_u \\
\end{bmatrix}
\end{align*}
where $z_2=w_2-w_2\Tilde{\lambda}+c$. The equality is shown as
\begin{align*}
(-\Tilde{\mathbf{W}}_{u})(\Tilde{\mathbf{L}}_{2}-\Tilde{\lambda}\Tilde{\mathbf{D}}_2)&=
\begin{bmatrix}
-\Tilde{w}_u & -\Tilde{w}_u \\
-\Tilde{w}_u & -\Tilde{w}_u 
\end{bmatrix}
\begin{bmatrix}
z_2 & -w_2 \\
-w_2 & z_2 
\end{bmatrix}\\
&=\begin{bmatrix}
-\Tilde{w}_u z_2+w_2\Tilde{w}_u & w_2\Tilde{w}_u-\Tilde{w}_u z_2 \\
-\Tilde{w}_u z_2+w_2\Tilde{w}_u & w_2\Tilde{w}_u-\Tilde{w}_u z_2
\end{bmatrix}\\
&=
\begin{bmatrix}
z_2 & -w_2 \\
-w_2 & z_2 
\end{bmatrix}
\begin{bmatrix}
-\Tilde{w}_u & -\Tilde{w}_u \\
-\Tilde{w}_u & -\Tilde{w}_u 
\end{bmatrix}\\
&=(\Tilde{\mathbf{L}}_{2}-\lambda\Tilde{\mathbf{D}}_2)(-\Tilde{\mathbf{W}}_{u})
\end{align*}
For a commutative matrix, it holds that
\begin{align*}
    \mathrm{det}(\Tilde{\mathbf{L}}_{12}-\Tilde{\lambda}\Tilde{\mathbf{D}}_{12})\hspace{-0.5mm}=\hspace{-0.5mm}\mathrm{det}\Big((\Tilde{\mathbf{L}}_{1}-\Tilde{\lambda}\Tilde{\mathbf{D}}_1)(\Tilde{\mathbf{L}}_{2}-\Tilde{\lambda}\Tilde{\mathbf{D}}_2)-\Tilde{\mathbf{W}}_u^2\Big)\hspace{-0.5mm}=\hspace{-0.5mm}0
\end{align*}
The solution results in a coupled set of equations as follows
\begin{align*}
0\hspace{-0.5mm} &=\hspace{-0.5mm}\bigg|
\begin{bmatrix}
z_1 & -w_1\\
-w_1  & z_1
\end{bmatrix}
\begin{bmatrix}
z_2 & -w_2\\
-w_2  & z_2
\end{bmatrix}-
\begin{bmatrix}
2\Tilde{w}_u^2 & 2\Tilde{w}_u^2\\
2\Tilde{w}_u^2 & 2\Tilde{w}_u^2
\end{bmatrix}\bigg|\\
&=
\bigg|
\begin{bmatrix}
z_1z_2+w_1w_2-2\Tilde{w}_u^2 & -w_2z_1-w_1z_2-2\Tilde{w}_u^2\\
-w_1z_2-w_2z_1-2\Tilde{w}_u^2& w_1w_2+z_1z_2-2\Tilde{w}_u^2
\end{bmatrix}
\bigg|\\
&=
\big(z_1z_2+w_1w_2-2\Tilde{w}_u^2\big)^2-\big(-w_2z_1-w_1z_2-2\Tilde{w}_u^2\big)^2\\
&=\big(z_1z_2+w_1w_2-2\Tilde{w}_u^2-(-w_2z_1-w_1z_2-2\Tilde{w}_u^2)\big)\\
&\quad\big(z_1z_2+w_1w_2-2\Tilde{w}_u^2+\big(-w_2z_1-w_1z_2-2\Tilde{w}_u^2)\big)\\
&=\big(w_1w_2\Tilde{\lambda}^2-4w_1w_2\Tilde{\lambda}+4w_1w_2-w_1\Tilde{\lambda} c-w_2\Tilde{\lambda} c+2w_1c\\&\quad+2w_2c+c^2 \big)\big(w_1w_2\Tilde{\lambda}^2-w_1\Tilde{\lambda} c-w_2\Tilde{\lambda} c+c^2-4\Tilde{w}_u^2 \big)\\
&=(w_1\Tilde{\lambda}-2w_1-c)(w_2\Tilde{\lambda}-2w_2-c)\\&\quad\big((w_1\Tilde{\lambda}-c)(w_2\Tilde{\lambda}-c)-4\Tilde{w}_u^2 \big)
\end{align*}
where $z_1=w_1-w_1\Tilde{\lambda}+c$, $z_2=w_2-w_2\Tilde{\lambda}+c$ and ${c=2\Tilde{w}_u-2\Tilde{w}_u\Tilde{\lambda}}$.
Now, substituting ${c=2\Tilde{w}_u-2\Tilde{w}_u\Tilde{\lambda}}$ into the first and second expressions yields
\begin{align*}
        0&=w_1\Tilde{\lambda}-2w_1-2\Tilde{w}_u+2\Tilde{w}_u\Tilde{\lambda}\\
      \Tilde{\lambda}&=\frac{2(w_1+\Tilde{w}_u)}{w_1+2\Tilde{w}_u}
\end{align*}
and
\begin{align*}
    0&=w_2\Tilde{\lambda}-2w_2-2\Tilde{w}_u+2\Tilde{w}_u\Tilde{\lambda}\\
    \Tilde{\lambda}&=\frac{2(w_2+\Tilde{w}_u)}{w_2+2\Tilde{w}_u}.
\end{align*}
Further, the third case, i.e.,  $(w_1\Tilde{\lambda}-c)(w_2\Tilde{\lambda}-c)-4\Tilde{w}_u^2=0$ must be examined by substituting ${c=2\Tilde{w}_u-2\Tilde{w}_u\Tilde{\lambda}}$:
\begin{align*}
  \begin{split}
      0=&w_1w_2\Tilde{\lambda}^2-w_1\Tilde{\lambda}(2\Tilde{w}_u-2\Tilde{w}_u\Tilde{\lambda})-w_2\Tilde{\lambda}(2\Tilde{w}_u-2\Tilde{w}_u\Tilde{\lambda})\\&+(2\Tilde{w}_u-2\Tilde{w}_u\Tilde{\lambda})^2-4\Tilde{w}_u^2\\
      0=&\Tilde{\lambda}(w_1w_2\Tilde{\lambda}-2\Tilde{w}_uw_1+2\Tilde{w}_uw_1\Tilde{\lambda}-2\Tilde{w}_uw_2\\&+2\Tilde{w}_uw_2\Tilde{\lambda}-8\Tilde{w}_u^2+4\Tilde{w}_u^2\Tilde{\lambda})
  \end{split}
\end{align*}
where the roots of the equation are the two smallest eigenvalues\footnote{For a detailed discussion about the ordering, see Section~C.2.}
\begin{align*}
    \Tilde{\lambda}_0=0\hspace{3mm}\mathrm{and}\hspace{3mm}
    \Tilde{\lambda}_1=\frac{2(w_1\Tilde{w}_u+w_2\Tilde{w}_u+4\Tilde{w}_u^2)}{(w_1+2\Tilde{w}_u)(w_2+2\Tilde{w}_u)}
\end{align*}
As a result, if $\Tilde{w}_u>0$, the second smallest eigenvalue corresponds to the Laplacian $\Tilde{\mathbf{L}}_{12}$ is greater than zero while the second smallest eigenvalue is equal to zero for $\mathbf{L}_{12}$.
\end{proof}
\vspace{-3mm}
\subsection{A.5~Proof of Theorem 4}\label{app:ProofTheorem4}
\begin{proof}[\unskip\nopunct]
The eigenvectors of Laplacian matrix $\mathbf{L}\in\mathbb{R}^{n\times n}$ are identical for both eigen-decompositions $\mathbf{L}\mathbf{y}_i=\lambda_i\mathbf{y}_i$ and $\mathbf{L}\mathbf{y}_i=\lambda_i\mathbf{D}\mathbf{y}_i$.\footnote{For a detailed information about the eigenvectors of $\mathbf{L}$, see Section~C.1 and C.3.} Then, the matrix ${\Tilde{\mathbf{Y}}=[\Tilde{\mathbf{y}}_0,\dots,\Tilde{\mathbf{y}}_{K-1}]\in\mathbb{R}^{n\times K}}$, i.e.
\begin{align*}
\mathbf{Y}=
\begin{bmatrix}
0 & y_{1,1}\\
0 & y_{1,1}\\
y_{3,0} & 0\\
y_{3,0} & 0\\
\end{bmatrix}.
\end{align*}
The next step is to design the matrix ${\Tilde{\mathbf{Y}}=[\Tilde{\mathbf{y}}_0,\dots,\Tilde{\mathbf{y}}_{K-1}]\in\mathbb{R}^{n\times K}}$ using the eigenvectors of the corrupted Laplacian matrix $\Tilde{\mathbf{L}}\in\mathbb{R}^{n\times n}$ based on, respectively, the standard eigen-decomposition $\mathbf{L}\mathbf{y}_i=\lambda_i\mathbf{y}_i$ and the generalized one $\mathbf{L}\mathbf{y}_i=\lambda_i\mathbf{D}\mathbf{y}_i$ as\footnote{For a detailed information about the eigenvectors of $\Tilde{\mathbf{L}}$, see Section~C.2 and C.4.}
\begin{align*}
\Tilde{\mathbf{Y}}=
\begin{bmatrix}
\Tilde{y}_{1,0} & \quad \Tilde{y}_{1,1}\\
\Tilde{y}_{1,0} & \quad \Tilde{y}_{1,1}\\
\Tilde{y}_{1,0} & - \Tilde{y}_{1,1}\\
\Tilde{y}_{1,0} & - \Tilde{y}_{1,1}\\
\end{bmatrix}\hspace{5mm}
\Tilde{\mathbf{Y}}=
\begin{bmatrix}
\Tilde{y}_{1,0} & \Tilde{y}_{1,1}\\\vspace{0.01mm}
\Tilde{y}_{1,0} & \Tilde{y}_{1,1}\\\vspace{0.01mm}
\Tilde{y}_{1,0} & \Tilde{y}_{3,1}\\\vspace{0.01mm}
\Tilde{y}_{1,0} & \Tilde{y}_{3,1}\\
\end{bmatrix}\hspace{5mm}
\end{align*}

By looking at the first column vector of ${\Tilde{\mathbf{Y}}}$ associated with the smallest eigenvalue $\Tilde{\lambda}_0$, it is evident that all feature vectors are mapped onto the same location $\Tilde{y}_{1,1}$ for the eigen-decompositions $\mathbf{L}\mathbf{y}_i=\lambda_i\mathbf{y}_i$ and $\mathbf{L}\mathbf{y}_i=\lambda_i\mathbf{D}\mathbf{y}_i$.
\end{proof}

\vspace{-3mm}
\subsection{A.6~Proof of Corollary 4.1}\label{app:ProofCorollary4.1}
\vspace{-1mm}
\begin{proof}[\unskip\nopunct]
Assuming that the column vectors of the matrices ${\mathbf{Y}}$ and ${\Tilde{\mathbf{Y}}}$ take values in a range ${\{y_{\mathrm{min}},y_{\mathrm{max}}\}}$, without loss of generalization, the matrices can be rewritten by defining the scaled eigenvectors of $\mathbf{L}\mathbf{y}_i=\lambda_i\mathbf{y}_i$ as
\begin{align*}
\mathbf{Y}=
\begin{bmatrix}
s(0) & s(y_{1,1})\\
s(0) & s(y_{1,1})\\
s(y_{3,0}) & s(0)\\
s(y_{3,0}) & s(0)\\
\end{bmatrix}\hspace{5mm}
\Tilde{\mathbf{Y}}=
\begin{bmatrix}
s(\Tilde{y}_{1,0}) & s(\Tilde{y}_{1,1})\\
s(\Tilde{y}_{1,0}) & s(\Tilde{y}_{1,1})\\
s(\Tilde{y}_{1,0}) & s(-\Tilde{y}_{1,1})\\
s(\Tilde{y}_{1,0}) & s(-\Tilde{y}_{1,1})\\
\end{bmatrix},
\end{align*}
where $\mathit{s}(y_{j,i})$ and $\mathit{s}(\Tilde{y}_{j,i})$ respectively denote the scaled $j$th embedding result in the $i$th eigenvector associated with Laplacian matrices $\mathbf{L}$ and $\Tilde{\mathbf{L}}$, and ${\mathit{s}(y_{j,i})=y_\mathrm{min}+\frac{y_{j,i}-\mathrm{min}(\mathbf{y}_i)}{\mathrm{max}(\mathbf{y}_i)-\mathrm{min}(\mathbf{y}_i)}(y_\mathrm{max}-y_\mathrm{min})}$ is the scaling function with $\mathrm{min}(\mathbf{y}_i)$ and $\mathrm{max}(\mathbf{y}_i)$ denoting the minimum and the maximum valued mapping points in the eigenvector $\mathbf{y}_i$, respectively. To cluster the row vectors of the matrix ${\mathbf{Y}}$, the squared Euclidean distances within the blocks and between different blocks are evaluated as
\begin{align*}
    d_{0,1}=&d_{2,3}=0\\
    d_{0,2}=&\big(\mathit{s}(0)-\mathit{s}(y_{3,0})\big)^2+\big(\mathit{s}(y_{1,1})-\mathit{s}(0)\big)^2\\
    =&2\big(y_\mathrm{max}-y_\mathrm{min}\big)^2,
\end{align*}
where $d_{i,j}$ denotes the distance between the $i$th and the $j$th feature vector with $d_{i,j}=d_{j,i}$ and $d_{0,2}=d_{0,3}=d_{1,2}=d_{1,3}$. 

Then, the distances within the blocks and between different blocks are evaluated for $\mathbf{L}\mathbf{y}_i=\lambda_i\mathbf{y}_i$ based ${\Tilde{\mathbf{Y}}}$ as
\begin{align*}
    d_{0,1}=&d_{2,3}=0\\
    d_{0,2}=&\big(\mathit{s}(\Tilde{y}_{1,0})-\mathit{s}(\Tilde{y}_{1,0})\big)^2+\big(\mathit{s}(\Tilde{y}_{1,1})-\mathit{s}(-\Tilde{y}_{1,1})\big)^2\\
    =&\big(y_\mathrm{max}-y_\mathrm{min}\big)^2.
\end{align*} 
The next step is to examine the matrices ${\mathbf{Y}}$ and ${\Tilde{\mathbf{Y}}}$ for the scaled eigenvectors of $\mathbf{L}\mathbf{y}_i=\lambda_i\mathbf{D}\mathbf{y}_i$, as
\begin{align*}
\mathbf{Y}=
\begin{bmatrix}
s(0) & s(y_{1,1})\\
s(0) & s(y_{1,1})\\
s(y_{3,0}) & s(0)\\
s(y_{3,0}) & s(0)\\
\end{bmatrix}\hspace{5mm}
\Tilde{\mathbf{Y}}=
\begin{bmatrix}
s(\Tilde{y}_{1,0}) & s(\Tilde{y}_{1,1})\\
s(\Tilde{y}_{1,0}) & s(\Tilde{y}_{1,1})\\
s(\Tilde{y}_{1,0}) & s(\Tilde{y}_{3,1})\\
s(\Tilde{y}_{1,0}) & s(\Tilde{y}_{3,1})\\
\end{bmatrix}\hspace{5mm}
\end{align*}
As the eigenvectors of $\mathbf{L}\mathbf{y}_i=\lambda_i\mathbf{D}\mathbf{y}_i$ are equivalent to $\mathbf{L}\mathbf{y}_i=\lambda_i\mathbf{y}_i$ for the Laplacian matrix $\mathbf{L}$, the distances associated with $\mathbf{Y}$ are also equal. Further, based on the knowledge that ${\Tilde{y}_{1,1}=\frac{-w_2+2\Tilde{w}_u}{w_1+2\Tilde{w}_u}\Tilde{y}_{3,1}}$, the scaled mapping points yield $s(\Tilde{y}_{1,1})=y_\mathrm{min}$ and $s(\Tilde{y}_{3,1})=y_\mathrm{max}$.\footnote{For a detailed discussion, see Section C.2.} Thus, the distances are computed for $\Tilde{\mathbf{Y}}$ as
\begin{align*}
    d_{0,1}=&d_{2,3}=0\\
    d_{0,2}=&\big(\mathit{s}(\Tilde{y}_{1,0})-\mathit{s}(\Tilde{y}_{1,0})\big)^2+\big(\mathit{s}(\Tilde{y}_{1,1})-\mathit{s}(-\Tilde{y}_{3,1})\big)^2\\
    =&\big(y_\mathrm{max}-y_\mathrm{min}\big)^2.
\end{align*}
This concludes the proof that the distance between the row vectors of the $\mathbf{Y}$ associated with different blocks is greater than that of 
$\Tilde{\mathbf{Y}}$.
\end{proof}
\vspace{-4mm}
\setcounter{secnumdepth}{0}
\section{Appendix B~: Robust Fiedler Vector Estimation}
\vspace{-1mm}
\subsection{B.1~Proof of Theorem 5}\label{app:ProofTheorem5}
\begin{proof}[\unskip\nopunct]
The objective function of RLPI, for the transformation vector $\boldsymbol{\beta}_i$, is given by
\begin{align*}
   \mathcal{L}(\mathbf{X},\hat{\mathbf{y}_i},\hat{\boldsymbol{\beta}_i)}=\|\mathbf{y}_i-\mathbf{X}^{\top}\boldsymbol{\beta}_i\|^2+\gamma\|\boldsymbol{\beta}_i\|^2.
\end{align*}
Introducing the weight function ${\boldsymbol{\Omega}\in \mathbb{R}^{n\times n}}$ leads to:
\begin{align*}
\begin{split}
   \mathcal{L}(\mathbf{X},\hat{\mathbf{y}}_i,\hat{\boldsymbol{\beta}}_i)&=\boldsymbol{\Omega}\|\mathbf{y}_i-\mathbf{X}^{\top}\boldsymbol{\beta}_i\|^2+\gamma\|\boldsymbol{\beta}_i\|^2\\
    =&\mathrm{tr}[\boldsymbol{\Omega}\mathbf{y}_i\mathbf{y}_i^{\top}-2\boldsymbol{\Omega}\mathbf{y}_i(\mathbf{X}^{\top}\boldsymbol{\beta}_i)^{\top}+\boldsymbol{\Omega}\mathbf{X}^{\top}\boldsymbol{\beta}_i(\mathbf{X}^{\top}\boldsymbol{\beta}_i)^{\top}]\\&+\gamma\mathrm{tr}(\boldsymbol{\beta}_i\boldsymbol{\beta}_i^{\top})\\
    =&\mathrm{tr}[\boldsymbol{\Omega}\mathbf{y}\mathbf{y}_i^{\top}-2\boldsymbol{\Omega}\mathbf{y}_i\boldsymbol{\beta}_i^{\top}\mathbf{X}+\boldsymbol{\Omega}\mathbf{X}^{\top}\boldsymbol{\beta}_i\boldsymbol{\beta}_i^{\top}\mathbf{X}]\\&+\gamma\mathrm{tr}(\boldsymbol{\beta}_i\boldsymbol{\beta}_i^{\top})\\
    =&\mathrm{tr}[\boldsymbol{\Omega}\mathbf{y}_i\mathbf{y}_i^{\top}-2\boldsymbol{\Omega}\mathbf{y}\boldsymbol{\beta}_i^{\top}\mathbf{X}+\boldsymbol{\beta}_i\boldsymbol{\beta}_i^{\top}(\mathbf{X}\boldsymbol{\Omega}\mathbf{X}^{\top}+\gamma\mathbf{I})]
\end{split}
\end{align*}
Setting the derivative of the right hand side with respect to $\boldsymbol{\beta}_i$ to zero, yields
\begin{align*}
   -2\mathbf{X}\boldsymbol{\Omega}\mathbf{y}_i+2\boldsymbol{\beta}_i(\mathbf{X}\boldsymbol{\Omega}\mathbf{X}^{\top}+\gamma\mathbf{I})=0.
\end{align*}
Thus,
\begin{align*}
    \hat{\boldsymbol{\beta}}_i=(\mathbf{X}\boldsymbol{\Omega}\mathbf{X}^{\top}+\gamma\mathbf{I})^{-1}\mathbf{X}\boldsymbol{\Omega}\mathbf{y}_i
\end{align*}For a Fiedler vector estimate $\mathbf{y}_i=\mathbf{y}_F$, substituting $\boldsymbol{\Omega}=\mathbf{I}$ in Eq.~(11) shows that 
RLPI is a special case of RRLPI and substituting $\boldsymbol{\Omega}=\mathbf{I}$ and $\gamma\rightarrow 0$ leads to LPI. 
\end{proof}
\vspace{-4mm}
\subsection{B.2~Proof of Theorem
6}\label{app:ProofTheorem6}
\vspace{-1mm}
\begin{proof}[\unskip\nopunct]
Suppose that $\mathrm{rank}(\mathbf{X})=\tau$, the SVD of $\mathbf{X}$ is
\begin{align*}
    \mathbf{X}=\mathbf{U}\boldsymbol{\Sigma}\mathbf{V}^{\top}
\end{align*}
where $\boldsymbol{\Sigma}=\mathrm{diag}(\Sigma_1,\dots,\Sigma_{\tau})$, $\mathbf{U}\in\mathbb{R}^{m\times\tau}$, $\mathbf{V}\in\mathbb{R}^{n\times\tau}$ and $\mathbf{U}^{\top}\mathbf{U}=\mathbf{V}^{\top}\mathbf{V}=\mathbf{I}$. 
This can be generalized using the weighted singular value decomposition\cite{WSVD}, i.e., 
\begin{align*}
    \mathbf{X}^{\ast}=\mathbf{U}\boldsymbol{\Sigma}\mathbf{V}^{\top}\boldsymbol{\Omega}.
\end{align*}
where $\boldsymbol{\Omega}\in\mathbb{R}^{n\times n}$ is a square positive definite symmetric weight matrix such that $\mathbf{V}^{\top}\boldsymbol{\Omega}\mathbf{V}=\mathbf{I}$.

Let $\mathbf{V}^{\ast}$ be a weighted matrix whose columns are weighted orthonormal eigenvectors of $\mathbf{V}$, i.e., $\mathbf{V}^{\ast}=\boldsymbol{\Omega}\mathbf{V}$. Then, the orthogonality term can be equivalently written as ${\mathbf{V}^{\top}\mathbf{V}^{\ast}=\mathbf{I}}$. The Fiedler vector $\mathbf{y}_F$ is in the space spanned by the column vectors of $\mathbf{V}^{\ast}$ because $\mathbf{y}_F$ is spanned by row vectors of weighted data matrix $\mathbf{X}^{\ast}$. Accordingly, the Fiedler vector $\mathbf{y}_F$ can be represented as a unique linear combination of the linearly independent column vectors of $\mathbf{V}^{\ast}$. For a set of combination coefficients $\mathbf{b}\in\mathbb{R}^{\tau}$,
\begin{align*}
\begin{split}
    \mathbf{V}^{\ast}\mathbf{b}=&\mathbf{y}_F\\\boldsymbol{\Omega}\mathbf{V}\mathbf{b}=&\mathbf{y}_F \\ \mathbf{V}^{\top}\boldsymbol{\Omega}\mathbf{V}\mathbf{b}=&\mathbf{V}^{\top}\mathbf{y}_F\\\mathbf{b}=&\mathbf{V}^{\top}\mathbf{y}_F
\end{split}
\end{align*}
Substituting $\mathbf{b}=\mathbf{V}^{\top}\mathbf{y}_F$ into $ \mathbf{V}^{\ast}\mathbf{b}=\mathbf{y}_F$, yields $\mathbf{V}^{\ast}\mathbf{V}^{\top}\mathbf{y}=\mathbf{y}_F$.
To continue, using the pseudo inverse of data matrix $\mathbf{X}^{\dagger}$ and weighted data matrix  $(\mathbf{X}^{\ast})^{\dagger}$ which can be written as
\begin{align*}
    \mathbf{X}^{\dagger}=\mathbf{V}\boldsymbol{\Sigma}^{-1}\mathbf{U}^{\top}
\end{align*}
and
\begin{align*}
    (\mathbf{X}^{\ast})^{\dagger}=\mathbf{V}\boldsymbol{\Sigma}^{-1}\mathbf{U}^{\top}\boldsymbol{\Psi},  
\end{align*}
it follows that
\begin{align*}
    \hat{\boldsymbol{\beta}}_F=(\mathbf{X}\boldsymbol{\Omega}\mathbf{X}^{\top}+\gamma\sigma^2\mathbf{I})^{-1}\mathbf{X}\boldsymbol{\Omega}\mathbf{y}_F=(\mathbf{X}^{\ast}\mathbf{X}^{\top}+\gamma\sigma^2\mathbf{I})^{-1}\mathbf{X}^{\ast}\mathbf{y}_F
\end{align*}
For a discussion, see Section~B.4. For $\gamma\rightarrow 0$
\begin{align*}
\hat{\boldsymbol{\beta}}_F=&(\mathbf{X}^{\ast}\mathbf{X}^{\top}+\gamma\sigma^2\mathbf{I})^{-1}\mathbf{X}^{\ast}\mathbf{y}_F\\
=&(\mathbf{X}^{\top})^{-1}(\mathbf{X}^{\ast})^{-1}\mathbf{X}^{\ast}\mathbf{y}_F\\
=&\mathbf{U}\boldsymbol{\Sigma}^{-1}\mathbf{V}^{\top}\mathbf{V}\boldsymbol{{\Sigma}}^{-1}\mathbf{U}^{\top}\boldsymbol{\Psi}\mathbf{U}\boldsymbol{\Sigma}\mathbf{V}^{\top}\boldsymbol{\Omega}\mathbf{y}_F\\
=&\mathbf{U}\boldsymbol{\Sigma}^{-1}\mathbf{V}^{\top}\mathbf{V}(\mathbf{V})^{-1}\mathbf{y}_F\\
=&\mathbf{U}\boldsymbol{\Sigma}^{-1}\mathbf{V}^{\top}\mathbf{y}_F
\end{align*}
Further, inserting the equation for $\boldsymbol{\beta}_F$ into $\hat{\mathbf{y}}_F=\mathbf{X}^{\top}\hat{\boldsymbol{\beta}}_F$ leads to
\begin{align*}
\begin{split}
    \hat{\mathbf{y}}_F&=\mathbf{X}^{\top}\hat{\boldsymbol{\beta}}_F=\mathbf{V}\boldsymbol{\Sigma}\mathbf{U}^{\top}\mathbf{U}\boldsymbol{\Sigma}^{-1}\mathbf{V}^{\top}\mathbf{y}_F=\mathbf{y}_F.
\end{split}
\end{align*}
 This shows that $\hat{\boldsymbol{\beta}}_F$ is the eigenvector of eigen-problem in Eq.~(4) for $\gamma\rightarrow 0$.
 \end{proof}
\vspace{-2.5mm}
\subsection{B.3~Proof of Corollary 6.1}\label{app:ProofCorollary6.1}
\begin{proof}[\unskip\nopunct]
Based on Theorem 6, it holds that for $\mathrm{rank}(\mathbf{X})=n$ all eigenvectors $\mathbf{y}_0,\dots,\mathbf{y}_{n-1}$ associated with $\lambda_0\leq\lambda_1\leq\dots\lambda_{n-1}$ are in the space spanned by row vectors of $\mathbf{X}$ and the transformation vector for the $i$th eigenvector is
\begin{align*}
   \hat{\boldsymbol{\beta}}_i^{\mathrm{(RRLPI)}} =\mathbf{U}\boldsymbol{\Sigma}^{-1}\mathbf{V}^{\top}\mathbf{y}_i
\end{align*}
where $\gamma\rightarrow0$, $\mathbf{U}^{\top}\boldsymbol{\Psi}\mathbf{U}=\mathbf{I}$ and $\mathbf{V}^{\top}\boldsymbol{\Omega}\mathbf{V}=\mathbf{I}$. If all feature vectors are linearly independent, each transformation vector is unique and is equal to the transformation functions of LPI, i.e., 
\begin{align*}
     \hat{\boldsymbol{\beta}}_i^{(\mathrm{LPI})}=\mathbf{U}\boldsymbol{\Sigma}^{-1}\mathbf{V}^{\top}\mathbf{y}_i=\boldsymbol{\beta}_i^{\mathrm{(RRLPI)}}.
\end{align*}
\end{proof}
\newpage
\subsection{B.4~Moore-Penrose Inverse of Weighted Data Matrix}

Let $\mathbf{X}\in\mathbb{R}^{m\times n}$, $\boldsymbol{\Omega}\in\mathbb{R}^{n\times n}$ and $\boldsymbol{\Psi}\in\mathbb{R}^{m\times m}$ denote a data matrix, and positive definite weight matrices, respectively. If there exist matrices $\boldsymbol{\Omega}$ and $\boldsymbol{\Psi}$ that satisfy $\mathbf{U}^{\top}\boldsymbol{\Psi}\mathbf{U}=\mathbf{I}$ and $\mathbf{V}^{\top}\boldsymbol{\Omega}\mathbf{V}=\mathbf{I}$ such that $\mathbf{X}^{\ast}=\mathbf{U}\boldsymbol{\Sigma}\mathbf{V}^{\top}\boldsymbol{\Omega}$, the weighted Moore-Penrose inverse $(\mathbf{X}^{\ast})^{\dagger}$ can be written as 
\begin{align*}
                 (\mathbf{X}^{\ast})^{\dagger}=\mathbf{V}\boldsymbol{\Sigma}^{-1}\mathbf{U}^{\top}\boldsymbol{\Psi}
\end{align*}
if it satisfies the following conditions:
\begin{align*}
      \mathbf{X}^{\ast}(\mathbf{X}^{\ast})^{\dagger}\mathbf{X}^{\ast}&=\mathbf{X}^{\ast}\\
      (\mathbf{X}^{\ast})^{\dagger}\mathbf{X}^{\ast}(\mathbf{X}^{\ast})^{\dagger}&=(\mathbf{X}^{\ast})^{\dagger}\\
      (\boldsymbol{\Psi}\mathbf{X}^{\ast}(\mathbf{X}^{\ast})^{\dagger})^{\top}&=\boldsymbol{\Psi}\mathbf{X}^{\ast}(\mathbf{X}^{\ast})^{\dagger}\\
      (\boldsymbol{\Omega}(\mathbf{X}^{\ast})^{\dagger}\mathbf{X}^{\ast})^{\top}&=\boldsymbol{\Omega}(\mathbf{X}^{\ast})^{\dagger}\mathbf{X}^{\ast}
\end{align*}.

The four conditions are now examined as

\begin{itemize}
    \item $\mathbf{X}^{\ast}(\mathbf{X}^{\ast})^{\dagger}\mathbf{X}^{\ast}\stackrel{?}{=}\mathbf{X}^{\ast}$
    \begin{align*}
        \mathbf{X}^{\ast}(\mathbf{X}^{\ast})^{\dagger}\mathbf{X}^{\ast}&=\mathbf{U}\boldsymbol{\Sigma}\mathbf{V}^{\top}\boldsymbol{\Omega}  \mathbf{V}\boldsymbol{\Sigma}^{-1}\mathbf{U}^{\top}\boldsymbol{\Psi}  \mathbf{U}\boldsymbol{\Sigma}\mathbf{V}^{\top}\boldsymbol{\Omega}\\
        &=\mathbf{U}\boldsymbol{\Sigma}\boldsymbol{\Sigma}^{-1}\boldsymbol{\Sigma}\mathbf{V}^{\top}\boldsymbol{\Omega}\\
        &=\mathbf{U}\boldsymbol{\Sigma}\mathbf{V}^{\top}\boldsymbol{\Omega}\\
        &=\mathbf{X}^{\ast}
    \end{align*}
    \item $(\mathbf{X}^{\ast})^{\dagger}\mathbf{X}^{\ast}(\mathbf{X}^{\ast})^{\dagger}\stackrel{?}{=}(\mathbf{X}^{\ast})^{\dagger}$
    \begin{align*}
        (\mathbf{X}^{\ast})^{\dagger}\mathbf{X}^{\ast}(\mathbf{X}^{\ast})^{\dagger}&=\mathbf{V}\boldsymbol{\Sigma}^{-1}\mathbf{U}^{\top}\boldsymbol{\Psi}\mathbf{U}\boldsymbol{\Sigma}\mathbf{V}^{\top}\boldsymbol{\Omega}    \mathbf{V}\boldsymbol{\Sigma}^{-1}\mathbf{U}^{\top}\boldsymbol{\Psi}\\
        &=\mathbf{V}\boldsymbol{\Sigma}^{-1}\boldsymbol{\Sigma}\boldsymbol{\Sigma}^{-1}\mathbf{U}^{\top}\boldsymbol{\Psi}\\
        &=\mathbf{V}\boldsymbol{\Sigma}^{-1}\mathbf{U}^{\top}\boldsymbol{\Psi}\\
        &=(\mathbf{X}^{\ast})^{\dagger}
    \end{align*}
    \item $(\boldsymbol{\Psi}\mathbf{X}^{\ast}(\mathbf{X}^{\ast})^{\dagger})^{\top}\stackrel{?}{=}\boldsymbol{\Psi}\mathbf{X}^{\ast}(\mathbf{X}^{\ast})^{\dagger}$
    \begin{align*}
        (\boldsymbol{\Psi}\mathbf{X}^{\ast}(\mathbf{X}^{\ast})^{\dagger})^{\top}&=(\boldsymbol{\Psi}\mathbf{U}\boldsymbol{\Sigma}\mathbf{V}^{\top}\boldsymbol{\Omega} \mathbf{V}\boldsymbol{\Sigma}^{-1}\mathbf{U}^{\top}\boldsymbol{\Psi})^{\top}\\
        &=(\boldsymbol{\Psi}\mathbf{U}\mathbf{U}^{\top}\boldsymbol{\Psi})^{\top}\\
        &=((\mathbf{U}^{\top})^{-1}\mathbf{U}^{-1})^{\top}\\
        &=(\mathbf{U}^{\top})^{-1}\mathbf{U}^{-1}\\
        &=\boldsymbol{\Psi}\mathbf{X}^{\ast}(\mathbf{X}^{\ast})^{\dagger}
    \end{align*}
    \item $(\boldsymbol{\Omega}(\mathbf{X}^{\ast})^{\dagger}\mathbf{X}^{\ast})^{\top}\stackrel{?}{=}\boldsymbol{\Omega}(\mathbf{X}^{\ast})^{\dagger}\mathbf{X}^{\ast}$
    \begin{align*}
        (\boldsymbol{\Omega}(\mathbf{X}^{\ast})^{\dagger}\mathbf{X}^{\ast})^{\top}&=(\boldsymbol{\Omega}\mathbf{V}\boldsymbol{\Sigma}^{-1}\mathbf{U}^{\top}\boldsymbol{\Psi}\mathbf{U}\boldsymbol{\Sigma}\mathbf{V}^{\top}\boldsymbol{\Omega})^{\top}\\
        &=(\boldsymbol{\Omega}\mathbf{V}\mathbf{V}^{\top}\boldsymbol{\Omega})^{\top}\\
        &=((\mathbf{V}^{\top})^{-1}\mathbf{V}^{-1})^{\top}\\
        &=(\mathbf{V}^{\top})^{-1}\mathbf{V}^{-1}\\
        &=\boldsymbol{\Omega}(\mathbf{X}^{\ast})^{\dagger}\mathbf{X}^{\ast}
    \end{align*}
    \end{itemize}
    Thus, for weight matrices $\boldsymbol{\Omega}$ and $\boldsymbol{\Psi}$, which satisfy $\mathbf{U}^{\top}\boldsymbol{\Psi}\mathbf{U}=\mathbf{I}$ and $\mathbf{V}^{\top}\boldsymbol{\Omega}\mathbf{V}=\mathbf{I}$ such that $\mathbf{X}^{\ast}=\mathbf{U}\boldsymbol{\Sigma}\mathbf{V}^{\top}\boldsymbol{\Omega}$, there exist Moore-Penrose inverse of the weighted data matrix $(\mathbf{X}^{\ast})^{\dagger}$.


\newpage
\setcounter{secnumdepth}{0}
\section{Appendix C~: Detailed Solutions for the Generalized and Standard Eigen-decompositions}

\subsection{C.1~Generalized Eigen-decomposotion for Affinity Matrix of Uncorrelated Blocks\label{app:GeneralizedEigenDecUncorrelated}}

For simplicity, let $\mathbf{W}\in\mathbb{R}^{n\times n}$ and the corresponding Laplacian $\mathbf{W}\in\mathbb{R}^{n\times n}$ consist of $k=2$ blocks, i.e.,
\begin{align*}
\mathbf{W}=
\begin{bmatrix}
\vspace{1mm}
0 & w_1 & 0&0\\
w_1 & 0 &0 &0\\
 0&  0& 0 &w_2\\
0 & 0 & w_2& 0
\end{bmatrix}
\hspace{5mm}
\mathbf{L}=
\begin{bmatrix}
\vspace{1mm}
w_1 & -w_1 & 0&0\\
-w_1 & w_1 &0 &0\\
 0&  0& w_2 &-w_2\\
0 & 0 & -w_2& w_2
\end{bmatrix},
\end{align*}
where $w_i$,$i=1,2$ denotes correlation coefficient of the $i$th block. To estimate the eigenvalues, the following expression is considered
\begin{align*}
    \mathrm{det}(\mathbf{L}-\lambda\mathbf{D})=0,
\end{align*}
which can equivalently be written in matrix form as
\begin{align*}
\begin{vmatrix}
\vspace{1mm}
w_1-\lambda w_1 & -w_1 & 0&0\\
-w_1 & w_1-\lambda w_1 &0 &0\\
 0&  0& w_2-\lambda w_2 &-w_2\\
0 & 0 & -w_2& w_2-\lambda w_2
\end{vmatrix}=0.
\end{align*}
Using determinant properties of block matrices, it follows that [62],
\begin{align*}
\begin{vmatrix}
w_1-\lambda w_1 & -w_1\\
-w_1 & w_1-\lambda w_1\\
\end{vmatrix}
\begin{vmatrix}
w_2-\lambda w_2 & -w_2\\
-w_2 & w_2-\lambda w_2\\
\end{vmatrix}=0.
\end{align*}
After solving this equation, the eigenvalues are equal to $\lambda_{0,1}=0$ and $\lambda_{2,3}=2$. To compute the eigenvectors, the following generalized eigen-decomposition is considered
\begin{align*}
    \mathbf{L}\mathbf{y}_i=\lambda_i\mathbf{D}\mathbf{y}_i,
\end{align*}
where $\mathbf{y}_i$ denotes the eigenvector associated with the $i$th eigenvalue $\lambda_i$. The eigen-problem can equivalently be written in the matrix form as
\begin{align*}
\begin{bmatrix}
\vspace{1mm}
w_1-\lambda_i w_1 & -w_1 & 0&0\\\vspace{1mm}
-w_1 & w_1-\lambda_i w_1 &0 &0\\\vspace{1mm}
 0&  0& w_2-\lambda_i w_2 &-w_2\\\vspace{1mm}
0 & 0 & -w_2& w_2-\lambda_i w_2
\end{bmatrix} 
\begin{bmatrix}
\vspace{1mm}
y_{1,i}\\\vspace{1mm}
y_{2,i}\\\vspace{1mm}
y_{3,i}\\\vspace{1mm}
y_{4,i}\\
\end{bmatrix}=0,
\end{align*}
where $y_{j,i}$, $j=1,\dots,4$ denotes the $j$th mapping point in the $i$th eigenvector $\mathbf{y}_i$. Substituting each $\lambda_i\in[\lambda_0,\dots,\lambda_3]$ leads to the eigenvectors
\begin{align*}
\mathbf{y}_0=
\begin{bmatrix}
\vspace{1mm}
y_{1,0}\\\vspace{1mm}
y_{1,0}\\\vspace{1mm}
y_{3,0}\\\vspace{1mm}
y_{3,0}\\
\end{bmatrix}\hspace{5mm}
\mathbf{y}_1=
\begin{bmatrix}
\vspace{1mm}
y_{1,1}\\\vspace{1mm}
y_{1,1}\\\vspace{1mm}
y_{3,1}\\\vspace{1mm}
y_{3,1}\\
\end{bmatrix}\hspace{5mm}
\mathbf{y}_2=
\begin{bmatrix}
\vspace{1mm}
\quad y_{1,2}\\\vspace{1mm}
-y_{1,2}\\\vspace{1mm}
\quad y_{3,2}\\\vspace{1mm}
-y_{3,2}\\
\end{bmatrix}\hspace{5mm}
\mathbf{y}_3=
\begin{bmatrix}
\vspace{1mm}
\quad y_{1,3}\\\vspace{1mm}
-y_{1,3}\\\vspace{1mm}
\quad y_{3,3}\\\vspace{1mm}
-y_{3,3}\\
\end{bmatrix}.
\end{align*}
According to the definition stating that
for a real and symmetric matrix the eigenvectors are orthogonal, the eigenvectors in four possible form yields
\begin{itemize}
    \item 
    \begin{align*}
\mathbf{y}_0=
\begin{bmatrix}
\vspace{1mm}
0\\\vspace{1mm}
0\\\vspace{1mm}
y_{3,0}\\\vspace{1mm}
y_{3,0}\\
\end{bmatrix}\hspace{5mm}
\mathbf{y}_1=
\begin{bmatrix}
\vspace{1mm}
y_{1,1}\\\vspace{1mm}
y_{1,1}\\\vspace{1mm}
0\\\vspace{1mm}
0\\
\end{bmatrix}\hspace{5mm}
\mathbf{y}_2=
\begin{bmatrix}
\vspace{1mm}
\quad0\\\vspace{1mm}
\quad0\\\vspace{1mm}
\quad y_{3,2}\\\vspace{1mm}
-y_{3,2}\\
\end{bmatrix}\hspace{5mm}
\mathbf{y}_3=
\begin{bmatrix}
\vspace{1mm}
\quad y_{1,3}\\\vspace{1mm}
-y_{1,3}\\\vspace{1mm}
\quad0\\\vspace{1mm}
\quad0\\
\end{bmatrix}
\end{align*}
    \item 
\begin{align*}
\mathbf{y}_0=
\begin{bmatrix}
\vspace{1mm}
0\\\vspace{1mm}
0\\\vspace{1mm}
y_{3,0}\\\vspace{1mm}
y_{3,0}\\
\end{bmatrix}\hspace{5mm}
\mathbf{y}_1=
\begin{bmatrix}
\vspace{1mm}
y_{1,1}\\\vspace{1mm}
y_{1,1}\\\vspace{1mm}
0\\\vspace{1mm}
0\\
\end{bmatrix}\hspace{5mm}
\mathbf{y}_2=
\begin{bmatrix}
\vspace{1mm}
\quad y_{1,2}\\\vspace{1mm}
-y_{1,2}\\\vspace{1mm}
\quad0\\\vspace{1mm}
\quad0\\
\end{bmatrix}\hspace{5mm}
\mathbf{y}_3=
\begin{bmatrix}
\vspace{1mm}
\quad0\\\vspace{1mm}
\quad0\\\vspace{1mm}
\quad y_{3,3}\\\vspace{1mm}
-y_{3,3}\\
\end{bmatrix}
\end{align*}
  \item 
    \begin{align*}
\mathbf{y}_0=
\begin{bmatrix}
\vspace{1mm}
y_{1,0}\\\vspace{1mm}
y_{1,0}\\\vspace{1mm}
0\\\vspace{1mm}
0\\
\end{bmatrix}\hspace{5mm}
\mathbf{y}_1=
\begin{bmatrix}
\vspace{1mm}
0\\\vspace{1mm}
0\\\vspace{1mm}
y_{3,1}\\\vspace{1mm}
y_{3,1}\\
\end{bmatrix}\hspace{5mm}
\mathbf{y}_2=
\begin{bmatrix}
\vspace{1mm}
\quad0\\\vspace{1mm}
\quad0\\\vspace{1mm}
\quad y_{3,2}\\\vspace{1mm}
-y_{3,2}\\
\end{bmatrix}\hspace{5mm}
\mathbf{y}_3=
\begin{bmatrix}
\vspace{1mm}
\quad y_{1,3}\\\vspace{1mm}
-y_{1,3}\\\vspace{1mm}
\quad0\\\vspace{1mm}
\quad0\\
\end{bmatrix}
\end{align*}
    \item 
\begin{align*}
\mathbf{y}_0=
\begin{bmatrix}
\vspace{1mm}
y_{1,0}\\\vspace{1mm}
y_{1,0}\\\vspace{1mm}
0\\\vspace{1mm}
0\\
\end{bmatrix}\hspace{5mm}
\mathbf{y}_1=
\begin{bmatrix}
\vspace{1mm}
0\\\vspace{1mm}
0\\\vspace{1mm}
y_{3,1}\\\vspace{1mm}
y_{3,1}\\
\end{bmatrix}\hspace{5mm}
\mathbf{y}_2=
\begin{bmatrix}
\vspace{1mm}
\quad y_{1,2}\\\vspace{1mm}
-y_{1,2}\\\vspace{1mm}
\quad0\\\vspace{1mm}
\quad0\\
\end{bmatrix}\hspace{5mm}
\mathbf{y}_3=
\begin{bmatrix}
\vspace{1mm}
\quad0\\\vspace{1mm}
\quad0\\\vspace{1mm}
\quad y_{3,3}\\\vspace{1mm}
-y_{3,3}\\
\end{bmatrix}
\end{align*}.
\end{itemize}

\subsection{C.2~Generalized Eigen-decomposition for Affinity Matrix of Correlated Blocks}\label{sec:GeneralizedEigenDeccorrelated}
Let $\mathbf{W}\in\mathbb{R}^{n\times n}$ denote a block zero-diagonal symmetric
affinity matrix for $k$ blocks which does not include correlations between different blocks, i.e., 
\begin{align*}
\begin{bmatrix}
0 & w_1 &\dots  &  &  & &\\
w_1 & 0 &\dots  &  &  & &\\
\vphantom{\int\limits^x}\smash{\vdots} &  \vphantom{\int\limits^x}\smash{\vdots}& 0 & w_2 &  & &\\
 &  & w_2  & 0& &  & \\
  &   &   &  & \ddots &\vphantom{\int\limits^x}\smash{\vdots} &\vphantom{\int\limits^x}\smash{\vdots}\\
  &  &   &  &\dots  &0 &w_k\\
  &   &   &  &\dots &w_k &0\\
\end{bmatrix}.
\end{align*}
The associated Laplacian has the same block structure for ${\mathbf{L}_1,\dots,\mathbf{L}_k}$. To estimate eigenvalues of $\mathbf{L}$, the expression ${\mathrm{det}(\mathbf{L}-\lambda\mathbf{D})=0}$ can be written in terms of the distinct determinants as
\begin{align*}
    \mathrm{det}(\mathbf{L}-\lambda\mathbf{D})=\mathrm{det}(\mathbf{L}_1-\lambda\mathbf{D}_1)\dots\mathrm{det}(\mathbf{L}_k-\lambda\mathbf{D}_k),
\end{align*}
where $\mathbf{L}_k$ and $\mathbf{D}_k$, respectively,  denote the Laplacian and the weight matrix of $k$th block. After solving this equation, for the $i$th block, the eigenvalues are equal to ${\lambda_0=0}$ and ${\lambda_1=2}$. Now assume that first and second blocks are correlated because of Type II outliers that have been added to the data set. Because the correlations are the effect of the outliers, we call them undesired correlations. The resulting corrupted affinity matrix $\Tilde{\mathbf{W}}\in\mathbb{R}^{n\times n}$ consists of $k$ blocks, i.e.,
\begin{align*}
\begin{bmatrix}
0 & w_1 & \Tilde{w}_u &\Tilde{w}_u  &\dots  & &\\
w_1 & 0 & \Tilde{w}_u& \Tilde{w}_u &\dots   & &\\
\Tilde{w}_u & \Tilde{w}_u & 0 & w_2 &  & &\\
\Tilde{w}_u & \Tilde{w}_u & w_2  & 0& &  & \\
\vdots  & \vdots  &   &  & \ddots &\vdots &\vdots\\
 &  &   &  &\dots  &0 &w_k\\
 &   &   &  &\dots &w_k &0\\
\end{bmatrix},
\end{align*}
where $\Tilde{w}_u$ denotes the correlation coefficient belonging to the undesired correlations for $\Tilde{w}_u\neq 0$. For simplicity, let $\Tilde{\mathbf{L}}_{12}$ and $\Tilde{\mathbf{D}}_{12}$ be the Laplacian and the weight matrix of correlated blocks, respectively. Hence, the eigenvalues of the corrupted Laplacian $\Tilde{\mathbf{L}}$ are estimated using distinct determinants as
\begin{align*}
    \mathrm{det}(\Tilde{\mathbf{L}}-\Tilde{\lambda}\Tilde{\mathbf{D}})=\mathrm{det}(\Tilde{\mathbf{L}}_{12}-\Tilde{\lambda}\Tilde{\mathbf{D}}_{12})\dots\mathrm{det}(\mathbf{L}_k-\lambda\mathbf{D}_k),
\end{align*}
where $\Tilde{\mathbf{L}}_{12}$ and $\Tilde{\mathbf{D}}_{12}$, respectively, denote the Laplacian and the weight matrix associated with the block that is a combination of blocks $i=1$ and $j=2$. To study the effects of outliers on the eigen-decomposition, the above equation shows that analyzing ${\mathrm{det}(\Tilde{\mathbf{L}}_{12}-\Tilde{\lambda}\Tilde{\mathbf{D}}_{12})=0}$ is sufficient, which can equivalently be written as
\begin{align*}
\begin{vmatrix}
w_1-w_1\Tilde{\lambda}+c & -w_1 & -\Tilde{w}_u & -\Tilde{w}_u\\
-w_1 &w_1-w_1\Tilde{\lambda}+c & -\Tilde{w}_u & -\Tilde{w}_u\\
-\Tilde{w}_u & -\Tilde{w}_u & w_2-w_2\Tilde{\lambda}+c & -w_2\\
-\Tilde{w}_u & -\Tilde{w}_u & -w_2 &w_2-w_2\Tilde{\lambda}+c\\
\end{vmatrix}=\hspace{-0.5mm}0,
\end{align*}
where $c=2\Tilde{w}_u-2\Tilde{w}_u\Tilde{\lambda}$. To simplify the determinant problem for block matrices as in [62], the matrix $\Tilde{\mathbf{L}}_{12}-\Tilde{\lambda}\Tilde{\mathbf{D}}_{12}$ must be commutative such that ${(-\Tilde{\mathbf{W}}_{u})(\Tilde{\mathbf{L}}_{2}-\Tilde{\lambda}\Tilde{\mathbf{D}}_2)=(\Tilde{\mathbf{L}}_{2}-\Tilde{\lambda}\Tilde{\mathbf{D}}_2)(-\Tilde{\mathbf{W}}_{u})}$. Accordingly, the matrix $\Tilde{\mathbf{L}}_{12}-\Tilde{\lambda}\Tilde{\mathbf{D}}_{12}$ is commutative if it satisfies
\begin{align*}
\begin{bmatrix}
-\Tilde{w}_u & -\Tilde{w}_u \\
-\Tilde{w}_u & -\Tilde{w}_u \\
\end{bmatrix}
\begin{bmatrix}
z_2 & -w_2 \\
-w_2 & z_2 \\
\end{bmatrix}=
\begin{bmatrix}
z_2 & -w_2 \\
-w_2 & z_2 \\
\end{bmatrix}
\begin{bmatrix}
-\Tilde{w}_u & -\Tilde{w}_u \\
-\Tilde{w}_u & -\Tilde{w}_u \\
\end{bmatrix}
\end{align*}
where $z_2=w_2-w_2\Tilde{\lambda}+c$. The equality is shown as:
\begin{align*}
(-\Tilde{\mathbf{W}}_{u})(\Tilde{\mathbf{L}}_{2}-\Tilde{\lambda}\Tilde{\mathbf{D}}_2)&=
\begin{bmatrix}
-\Tilde{w}_u & -\Tilde{w}_u \\
-\Tilde{w}_u & -\Tilde{w}_u 
\end{bmatrix}
\begin{bmatrix}
z_2 & -w_2 \\
-w_2 & z_2 
\end{bmatrix}\\
&=\begin{bmatrix}
-\Tilde{w}_u z_2+w_2\Tilde{w}_u & w_2\Tilde{w}_u-\Tilde{w}_u z_2 \\
-\Tilde{w}_u z_2+w_2\Tilde{w}_u & w_2\Tilde{w}_u-\Tilde{w}_u z_2
\end{bmatrix}\\
&=
\begin{bmatrix}
z_2 & -w_2 \\
-w_2 & z_2 
\end{bmatrix}
\begin{bmatrix}
-\Tilde{w}_u & -\Tilde{w}_u \\
-\Tilde{w}_u & -\Tilde{w}_u 
\end{bmatrix}\\
&=(\Tilde{\mathbf{L}}_{2}-\lambda\Tilde{\mathbf{D}}_2)(-\Tilde{\mathbf{W}}_{u})
\end{align*}
For a commutative matrix, it holds that
\begin{align*}
    \mathrm{det}(\Tilde{\mathbf{L}}_{12}-\Tilde{\lambda}\Tilde{\mathbf{D}}_{12})\hspace{-0.5mm}=\hspace{-0.5mm}\mathrm{det}\Big((\Tilde{\mathbf{L}}_{1}-\Tilde{\lambda}\Tilde{\mathbf{D}}_1)(\Tilde{\mathbf{L}}_{2}-\Tilde{\lambda}\Tilde{\mathbf{D}}_2)-\Tilde{\mathbf{W}}_u^2\Big)\hspace{-0.5mm}=\hspace{-0.5mm}0.
\end{align*}
The solution results in coupled set of equations as follows
\begin{align*}
0\hspace{-0.5mm} &=\hspace{-0.5mm}\bigg|
\begin{bmatrix}
z_1 & -w_1\\
-w_1  & z_1
\end{bmatrix}
\begin{bmatrix}
z_2 & -w_2\\
-w_2  & z_2
\end{bmatrix}-
\begin{bmatrix}
2\Tilde{w}_u^2 & 2\Tilde{w}_u^2\\
2\Tilde{w}_u^2 & 2\Tilde{w}_u^2
\end{bmatrix}\bigg|\\
&=
\bigg|
\begin{bmatrix}
z_1z_2+w_1w_2-2\Tilde{w}_u^2 & -w_2z_1-w_1z_2-2\Tilde{w}_u^2\\
-w_1z_2-w_2z_1-2\Tilde{w}_u^2& w_1w_2+z_1z_2-2\Tilde{w}_u^2
\end{bmatrix}
\bigg|\\
&=
\big(z_1z_2+w_1w_2-2\Tilde{w}_u^2\big)^2-\big(-w_2z_1-w_1z_2-2\Tilde{w}_u^2\big)^2\\
&=\big(z_1z_2+w_1w_2-2\Tilde{w}_u^2-(-w_2z_1-w_1z_2-2\Tilde{w}_u^2)\big)\\
&\quad\big(z_1z_2+w_1w_2-2\Tilde{w}_u^2+\big(-w_2z_1-w_1z_2-2\Tilde{w}_u^2)\big)\\
&=\big(w_1w_2\Tilde{\lambda}^2-4w_1w_2\Tilde{\lambda}+4w_1w_2-w_1\Tilde{\lambda} c-w_2\Tilde{\lambda} c+2w_1c\\&\quad+2w_2c+c^2 \big)\big(w_1w_2\Tilde{\lambda}^2-w_1\Tilde{\lambda} c-w_2\Tilde{\lambda} c+c^2-4\Tilde{w}_u^2 \big)\\
&=(w_1\Tilde{\lambda}-2w_1-c)(w_2\Tilde{\lambda}-2w_2-c)\\&\quad\big((w_1\Tilde{\lambda}-c)(w_2\Tilde{\lambda}-c)-4\Tilde{w}_u^2 \big),
\end{align*}
where $z_1=w_1-w_1\Tilde{\lambda}+c$, $z_2=w_2-w_2\Tilde{\lambda}+c$ and ${c=2\Tilde{w}_u-2\Tilde{w}_u\Tilde{\lambda}}$.
Now, substituting ${c=2\Tilde{w}_u-2\Tilde{w}_u\Tilde{\lambda}}$ into the first and second expressions yields
\begin{align*}
        0&=w_1\Tilde{\lambda}-2w_1-2\Tilde{w}_u+2\Tilde{w}_u\Tilde{\lambda}\\
      \Tilde{\lambda}&=\frac{2(w_1+\Tilde{w}_u)}{w_1+2\Tilde{w}_u}
\end{align*}
and
\begin{align*}
    0&=w_2\Tilde{\lambda}-2w_2-2\Tilde{w}_u+2\Tilde{w}_u\Tilde{\lambda}\\
    \Tilde{\lambda}&=\frac{2(w_2+\Tilde{w}_u)}{w_2+2\Tilde{w}_u}.
\end{align*}
Further, the third case, i.e.,  $(w_1\Tilde{\lambda}-c)(w_2\Tilde{\lambda}-c)-4\Tilde{w}_u^2=0$ must be examined by substituting ${c=2\Tilde{w}_u-2\Tilde{w}_u\Tilde{\lambda}}$:
\begin{align*}
  \begin{split}
      w_1w_2\Tilde{\lambda}^2-w_1\Tilde{\lambda}(2\Tilde{w}_u-2\Tilde{w}_u\Tilde{\lambda})-w_2\Tilde{\lambda}(2\Tilde{w}_u-2\Tilde{w}_u\Tilde{\lambda})+(2\Tilde{w}_u-2\Tilde{w}_u\Tilde{\lambda})^2-4\Tilde{w}_u^2=&0\\
      \Tilde{\lambda}(w_1w_2\Tilde{\lambda}-2\Tilde{w}_uw_1+2\Tilde{w}_uw_1\Tilde{\lambda}-2\Tilde{w}_uw_2+2\Tilde{w}_uw_2\Tilde{\lambda}-8\Tilde{w}_u^2+4\Tilde{w}_u^2\Tilde{\lambda})=&0,
  \end{split}
\end{align*}
where the roots of the equation are the eigenvalues
\begin{align*}
    \Tilde{\lambda}=0\hspace{3mm}\mathrm{and}\hspace{3mm}
    \Tilde{\lambda}=\frac{2(w_1\Tilde{w}_u+w_2\Tilde{w}_u+4\Tilde{w}_u^2)}{(w_1+2\Tilde{w}_u)(w_2+2\Tilde{w}_u)}.
\end{align*}
To order the eigenvalues $ \Tilde{\lambda}=\frac{2(w_1+\Tilde{w}_u)}{w_1+2\Tilde{w}_u}$ and $\Tilde{\lambda}=\frac{2(w_2+\Tilde{w}_u)}{w_2+2\Tilde{w}_u}$, the relevant information is insufficient. However, these two eigenvalues are compared with the eigenvalue $\Tilde{\lambda}=\frac{2(w_1\Tilde{w}_u+w_2\Tilde{w}_u+4\Tilde{w}_u^2)}{(w_1+2\Tilde{w}_u)(w_2+2\Tilde{w}_u)}$ to reveal the two smallest eigenvalues as
\begin{align*}
\begin{split}
    \frac{2(w_1\Tilde{w}_u+w_2\Tilde{w}_u+4\Tilde{w}_u^2)}{(w_1+2\Tilde{w}_u)(w_2+2\Tilde{w}_u)}\stackrel{?}{<}&\frac{2(w_1+\Tilde{w}_u)}{w_1+2\Tilde{w}_u}\\
    \frac{2(w_1\Tilde{w}_u+w_2\Tilde{w}_u+4\Tilde{w}_u^2)}{(w_1+2\Tilde{w}_u)(w_2+2\Tilde{w}_u)}\stackrel{?}{<}&\frac{2(w_1+\Tilde{w}_u)(w_2+2\Tilde{w}_u)}{(w_1+2\Tilde{w}_u)(w_2+2\Tilde{w}_u)}\\
    2w_1\Tilde{w}_u+2w_2\Tilde{w}_u+8\Tilde{w}_u^2\stackrel{?}{<}&2w_1w_2+4w_1\Tilde{w}_u+2w_2\Tilde{w}_u+4\Tilde{w}_u^2\\
    4\Tilde{w}_u^2\stackrel{?}{<}&2w_1w_2+2w_1\Tilde{w}_u
\end{split}
\end{align*}
and 
\begin{align*}
\begin{split}
    \frac{2(w_1\Tilde{w}_u+w_2\Tilde{w}_u+4\Tilde{w}_u^2)}{(w_1+2\Tilde{w}_u)(w_2+2\Tilde{w}_u)}\stackrel{?}{<}&\frac{2(w_2+\Tilde{w}_u)}{w_2+2\Tilde{w}_u}\\
    \frac{2(w_1\Tilde{w}_u+w_2\Tilde{w}_u+4\Tilde{w}_u^2)}{(w_1+2\Tilde{w}_u)(w_2+2\Tilde{w}_u)}\stackrel{?}{<}&\frac{2(w_2+\Tilde{w}_u)(w_1+2\Tilde{w}_u)}{(w_2+2\Tilde{w}_u)(w_1+2\Tilde{w}_u)}\\
    2w_1\Tilde{w}_u+2w_2\Tilde{w}_u+8\Tilde{w}_u^2\stackrel{?}{<}&2w_1w_2+4w_2\Tilde{w}_u+2w_1\Tilde{w}_u+4\Tilde{w}_u^2\\
    4\Tilde{w}_u^2\stackrel{?}{<}&2w_1w_2+2w_2\Tilde{w}_u.
\end{split}
\end{align*}
For $w_1>\Tilde{w}_u$ and $w_2>\Tilde{w}_u$, it holds that $\frac{2(w_1\Tilde{w}_u+w_2\Tilde{w}_u+4\Tilde{w}_u^2)}{(w_1+2\Tilde{w}_u)(w_2+2\Tilde{w}_u)}<\frac{2(w_1+\Tilde{w}_u)}{w_1+2\Tilde{w}_u}$ and $\frac{2(w_1\Tilde{w}_u+w_2\Tilde{w}_u+4\Tilde{w}_u^2)}{(w_1+2\Tilde{w}_u)(w_2+2\Tilde{w}_u)}<\frac{2(w_2+\Tilde{w}_u)}{w_2+2\Tilde{w}_u}$. To summarize, for $\frac{2(w_2+\Tilde{w}_u)}{w_2+2\Tilde{w}_u}\leq\frac{2(w_1+\Tilde{w}_u)}{w_1+2\Tilde{w}_u}$the eigenvalues in ascending order are:
\begin{align*}
    \Tilde{\lambda}_{0}=0\hspace{5mm} \Tilde{\lambda}_{1}=\frac{2(w_1\Tilde{w}_u+w_2\Tilde{w}_u+4\Tilde{w}_u^2)}{(w_1+2\Tilde{w}_u)(w_2+2\Tilde{w}_u)}\hspace{5mm}
    \Tilde{\lambda}_{2}=\frac{2(w_2+\Tilde{w}_u)}{w_2+2\Tilde{w}_u}\hspace{5mm}
    \Tilde{\lambda}_{3}=\frac{2(w_1+\Tilde{w}_u)}{w_1+2\Tilde{w}_u}
\end{align*}
or are
\begin{align*}
        \Tilde{\lambda}_{0}=0\hspace{5mm} \Tilde{\lambda}_{1}=\frac{2(w_1\Tilde{w}_u+w_2\Tilde{w}_u+4\Tilde{w}_u^2)}{(w_1+2\Tilde{w}_u)(w_2+2\Tilde{w}_u)}\hspace{5mm}
       \Tilde{\lambda}_{2}=\frac{2(w_1+\Tilde{w}_u)}{w_1+2\Tilde{w}_u}
    \hspace{5mm}
    \Tilde{\lambda}_{3}=\frac{2(w_2+\Tilde{w}_u)}{w_2+2\Tilde{w}_u}
\end{align*}
otherwise.
To analyze the effect of outliers on eigenvectors, the generalized eigenvalue problem is considered
\begin{align*}
    \Tilde{\mathbf{L}}\Tilde{\mathbf{y}}_i=\Tilde{\lambda}_i\Tilde{\mathbf{D}}\Tilde{\mathbf{y}}_i,
\end{align*}
where $\Tilde{\mathbf{y}}_i$ denotes the eigenvector associated with the $i$th eigenvalue $\Tilde{\lambda}_i$. The expression can equivalently be written as
\begin{align*}
\begin{bmatrix}
\vspace{1mm}
w_1+2\Tilde{w}_u & -w_1 & -\Tilde{w}_u \hspace{-3mm}& -\Tilde{w}_u\\\vspace{1mm}
-w_1 & w_1+2\Tilde{w}_u & -\Tilde{w}_u & -\Tilde{w}_u\\\vspace{1mm}
-\Tilde{w}_u & -\Tilde{w}_u & w_2+2\Tilde{w}_u & -w_2\\\vspace{1mm}
-\Tilde{w}_u & -\Tilde{w}_u & -w_2 & w_2+2\Tilde{w}_u\\
\end{bmatrix}
\begin{bmatrix}
\vspace{1mm}
\Tilde{y}_{1,i}\\\vspace{1mm}
\Tilde{y}_{2,i}\\\vspace{1mm}
\Tilde{y}_{3,i}\\\vspace{1mm}
\Tilde{y}_{4,i}\\
\end{bmatrix}=
\begin{bmatrix}
\vspace{1mm}
\Tilde{\lambda}_i\Tilde{d}_1& 0 & 0 & 0\\\vspace{1mm}
0 & \Tilde{\lambda}_i\Tilde{d}_1 & 0& 0\\\vspace{1mm}
0 & 0 & \Tilde{\lambda}_i\Tilde{d}_2 & 0\\\vspace{1mm}
0 & 0 & 0 & \Tilde{\lambda}_i\Tilde{d}_2\\
\end{bmatrix}
\begin{bmatrix}
\vspace{1mm}
\Tilde{y}_{1,i}\\\vspace{1mm}
\Tilde{y}_{2,i}\\\vspace{1mm}
\Tilde{y}_{3,i}\\\vspace{1mm}
\Tilde{y}_{4,i}\\
\end{bmatrix}
\end{align*}
and for $\Tilde{d}_1=w_1+2\Tilde{w}_u$ and $\Tilde{d}_2=w_2+2\Tilde{w}_u$, it holds that
\begin{align*}
        (w_1+2\Tilde{w}_u)\Tilde{y}_{1,i}-w_1\Tilde{y}_{2,i}-\Tilde{w}_u(\Tilde{y}_{3,i}+\Tilde{y}_{4,i})=\Tilde{\lambda}_{i}(w_1+2\Tilde{w}_u)\Tilde{y}_{1,i}\\
        -w_1\Tilde{y}_{1,i}+(w_1+2\Tilde{w}_u)\Tilde{y}_{2,i}-\Tilde{w}_u(\Tilde{y}_{3,i}+\Tilde{y}_{4,i})=\Tilde{\lambda}_{i}(w_1+2\Tilde{w}_u)\Tilde{y}_{2,i}\\
        -\Tilde{w}_u(\Tilde{y}_{1,i}+\Tilde{y}_{2,i})+(w_2+2\Tilde{w}_u)\Tilde{y}_{3,i}-w_2\Tilde{y}_{4,i}=\Tilde{\lambda}_{i}(w_2+2\Tilde{w}_u)\Tilde{y}_{3,i}\\
        -\Tilde{w}_u(\Tilde{y}_{1,i}+\Tilde{y}_{2,i})-w_2\Tilde{y}_{3,i}+(w_2+2\Tilde{w}_u)\Tilde{y}_{4,i}=\Tilde{\lambda}_{i}(w_2+2\Tilde{w}_u)\Tilde{y}_{4,i}
\end{align*}
where $\Tilde{y}_{j,i}$ denotes the $j$th mapping point in the eigenvector associated with the $i$th eigenvalue $\Tilde{\lambda}_i$.
To compute the eigenvectors, each $\Tilde{\lambda}_i\in[\Tilde{\lambda}_0,\dots,\Tilde{\lambda}_3]$ is substituted in $\Tilde{\mathbf{L}}\Tilde{\mathbf{y}}_i=\Tilde{\lambda}_i\Tilde{\mathbf{D}}\Tilde{\mathbf{y}}_i$ as follows:
\begin{itemize}
\item  $\Tilde{\mathbf{L}}\Tilde{\mathbf{y}}_0=\Tilde{\lambda}_0\Tilde{\mathbf{D}}\Tilde{\mathbf{y}}_0$ such that $\Tilde{\lambda}_0=0$
    \begin{align*}
        (w_1+2\Tilde{w}_u)\Tilde{y}_{1,0}-w_1\Tilde{y}_{2,0}-\Tilde{w}_u(\Tilde{y}_{3,0}+\Tilde{y}_{4,0})=0\\
        -w_1\Tilde{y}_{1,0}+(w_1+2\Tilde{w}_u)\Tilde{y}_{2,0}-\Tilde{w}_u(\Tilde{y}_{3,0}+\Tilde{y}_{4,0})=0\\
        -\Tilde{w}_u(\Tilde{y}_{1,0}+\Tilde{y}_{2,0})+(w_2+2\Tilde{w}_u)\Tilde{y}_{3,0}-w_2\Tilde{y}_{4,0}=0\\
        -\Tilde{w}_u(\Tilde{y}_{1,0}+\Tilde{y}_{2,0})-w_2\Tilde{y}_{3,0}+(w_2+2\Tilde{w}_u)\Tilde{y}_{4,0}=0\\
    \end{align*}
  The second expression is subtracted from the first one as
  \begin{align*}
      2w_1\Tilde{y}_{1,0}+2\Tilde{w}_u(\Tilde{y}_{1,0}-\Tilde{y}_{2,0})-2w_1\Tilde{y}_{2,0}&=0\\
      2w_1(\Tilde{y}_{1,0}-\Tilde{y}_{2,0})+2\Tilde{w}_u(\Tilde{y}_{1,0}-\Tilde{y}_{2,0})&=0\\
      2(w_1+\Tilde{w}_u)(\Tilde{y}_{1,0}-\Tilde{y}_{2,0})&=0.\\
  \end{align*}
  As $w_1>\Tilde{w}_u>0$, $\Tilde{y}_{1,0}=\Tilde{y}_{2,0}$. Similarly, the fourth expression is subtracted from the third one as:
   \begin{align*}
      2w_2\Tilde{y}_{3,0}+2\Tilde{w}_u(\Tilde{y}_{3,0}-\Tilde{y}_{4,0})-2w_2\Tilde{y}_{4,0}&=0\\
      2w_2(\Tilde{y}_{3,0}-\Tilde{y}_{4,0})+2\Tilde{w}_u(\Tilde{y}_{3,0}-\Tilde{y}_{4,0})&=0\\
      2(w_2+\Tilde{w}_u)(\Tilde{y}_{3,0}-\Tilde{y}_{4,0})&=0\\
  \end{align*}
  Since $w_2>\Tilde{w}_u>0$, it follows that $\Tilde{y}_{3,0}=\Tilde{y}_{4,0}$. Lastly, the first and the second expressions are summed up as follows:
  \begin{align*}
      2\Tilde{w}_u(\Tilde{y}_{1,0}+\Tilde{y}_{2,0})-2\Tilde{w}_u(\Tilde{y}_{3,0}+\Tilde{y}_{4,0})&=0\\
      \Tilde{y}_{1,0}+\Tilde{y}_{2,0}&=\Tilde{y}_{3,0}+\Tilde{y}_{4,0}\\
  \end{align*}
Substituting $\Tilde{y}_{1,0}=\Tilde{y}_{2,0}$ and $\Tilde{y}_{3,0}=\Tilde{y}_{4,0}$ yields $\Tilde{y}_{1,0}=\Tilde{y}_{2,0}=\Tilde{y}_{3,0}=\Tilde{y}_{4,0}$. Thus, the eigenvector associated with the smallest eigenvalue $\Tilde{\lambda}_0$ is
\begin{align*}
\Tilde{\mathbf{y}}_0=
    \begin{bmatrix}
    \vspace{1mm}
\Tilde{y}_{1,0}\\\vspace{1mm}
\Tilde{y}_{1,0}\\\vspace{1mm}
\Tilde{y}_{1,0} \\\vspace{1mm}
\Tilde{y}_{1,0} \\
\end{bmatrix}
\end{align*}

\item $\Tilde{\mathbf{L}}\Tilde{\mathbf{y}}_1=\Tilde{\lambda}_1\Tilde{\mathbf{D}}\Tilde{\mathbf{y}}_1$ such that $\Tilde{\lambda}_1=\frac{2(w_1\Tilde{w}_u+w_2\Tilde{w}_u+4\Tilde{w}_u^2)}{(w_1+2\Tilde{w}_u)(w_2+2\Tilde{w}_u)}$
\begin{align*}
        (w_1+2\Tilde{w}_u)\Tilde{y}_{1,1}-w_1\Tilde{y}_{2,1}-\Tilde{w}_u(\Tilde{y}_{3,1}+\Tilde{y}_{4,1})=\frac{2(w_1\Tilde{w}_u+w_2\Tilde{w}_u+4\Tilde{w}_u^2)}{(w_2+2\Tilde{w}_u)}\Tilde{y}_{1,1}\\
        -w_1\Tilde{y}_{1,1}+(w_1+2\Tilde{w}_u)\Tilde{y}_{2,1}-\Tilde{w}_u(\Tilde{y}_{3,1}+\Tilde{y}_{4,1})=\frac{2(w_1\Tilde{w}_u+w_2\Tilde{w}_u+4\Tilde{w}_u^2)}{(w_2+2\Tilde{w}_u)}\Tilde{y}_{2,1}\\
        -\Tilde{w}_u(\Tilde{y}_{1,1}+\Tilde{y}_{2,1})+(w_2+2\Tilde{w}_u)\Tilde{y}_{3,1}-w_2\Tilde{y}_{4,1}=\frac{2(w_1\Tilde{w}_u+w_2\Tilde{w}_u+4\Tilde{w}_u^2)}{(w_1+2\Tilde{w}_u)}\Tilde{y}_{3,1}\\
        -\Tilde{w}_u(\Tilde{y}_{1,1}+\Tilde{y}_{2,1})-w_2\Tilde{y}_{3,1}+(w_2+2\Tilde{w}_u)\Tilde{y}_{4,1}=\frac{2(w_1\Tilde{w}_u+w_2\Tilde{w}_u+4\Tilde{w}_u^2)}{(w_1+2\Tilde{w}_u)}\Tilde{y}_{4,1}\\
    \end{align*}
 Then, the first two expressions are subtracted yielding:
 \begin{align*}
 \begin{split}
       2w_1\Tilde{y}_{1,1}+2\Tilde{w}_u\Tilde{y}_{1,1}-2w_1\Tilde{y}_{2,1}-2\Tilde{w}_u\Tilde{y}_{2,1}&=(\Tilde{y}_{1,1}-\Tilde{y}_{2,1})\frac{2(w_1\Tilde{w}_u+w_2\Tilde{w}_u+4\Tilde{w}_u^2)}{(w_2+2\Tilde{w}_u)}\\
       2w_1+2\Tilde{w}_u&=\frac{2(w_1\Tilde{w}_u+w_2\Tilde{w}_u+4\Tilde{w}_u^2)}{w_2+2\Tilde{w}_u}\\
 \end{split}
 \end{align*}
Straight-forwardly, subtracting the expressions does not provide further information about mapping points. Thus, the first two expressions are summed up as
\begin{align*}
    \begin{split}
        \cancel{w_1\Tilde{y}_{1,1}}+2\Tilde{w}_u\Tilde{y}_{1,1}-\cancel{w_1\Tilde{y}_{2,1}}-2\Tilde{w}_u(\Tilde{y}_{3,1}+\Tilde{y}_{4,1})-\cancel{w_1\Tilde{y}_{1,1}}+\cancel{w_1\Tilde{y}_{2,1}}+2\Tilde{w}_u\Tilde{y}_{2,1}&=\frac{2(w_1\Tilde{w}_u+w_2\Tilde{w}_u+4\Tilde{w}_u^2)}{(w_2+2\Tilde{w}_u)}(\Tilde{y}_{1,1}+\Tilde{y}_{2,1})\\
        (\Tilde{y}_{1,1}+\Tilde{y}_{2,1})\Big(2\Tilde{w}_u-\frac{2(w_1\Tilde{w}_u+w_2\Tilde{w}_u+4\Tilde{w}_u^2)}{w_2+2\Tilde{w}_u}\Big)&=2\Tilde{w}_u(\Tilde{y}_{3,1}+\Tilde{y}_{4,1})\\
        (\Tilde{y}_{1,1}+\Tilde{y}_{2,1})&=-\frac{w_2+2\Tilde{w}_u}{w_1+2\Tilde{w}_u}(\Tilde{y}_{3,1}+\Tilde{y}_{4,1})\\
    \end{split}
\end{align*}
Identically, summing up the third and fourth expressions yields
\begin{align*}
    \begin{split}
        -2\Tilde{w}_u(\Tilde{y}_{1,1}+\Tilde{y}_{2,1})+2\Tilde{w}_u(\Tilde{y}_{3,1}+\Tilde{y}_{4,1})&=\frac{2(w_1\Tilde{w}_u+w_2\Tilde{w}_u+4\Tilde{w}_u^2)}{(w_2+2\Tilde{w}_u)} (\Tilde{y}_{3,1}+\Tilde{y}_{4,1})\\
        (\Tilde{y}_{1,1}+\Tilde{y}_{2,1})&=-\frac{w_2+2\Tilde{w}_u}{w_1+2\Tilde{w}_u}(\Tilde{y}_{3,1}+\Tilde{y}_{4,1}).
    \end{split}
\end{align*}
The available information is not sufficient to understand the intrinsic relationships of the mappings. Thus, the eigenvector is written based on the available information as:
\begin{align*}
\Tilde{\mathbf{y}}_1=
 \begin{bmatrix}
   \vspace{1mm}
\Tilde{y}_{1,1}\\\vspace{1mm}
\Tilde{y}_{2,1}\\\vspace{1mm}
\Tilde{y}_{3,1}\\\vspace{1mm}
\Tilde{y}_{4,1} \\
\end{bmatrix}
\end{align*}

\item $\Tilde{\mathbf{L}}\Tilde{\mathbf{y}}_2=\Tilde{\lambda}_2\Tilde{\mathbf{D}}\Tilde{\mathbf{y}}_2$ such that $\Tilde{\lambda}_2=\frac{2(w_2+\Tilde{w}_u)}{w_2+2\Tilde{w}_u}$ and $\frac{2(w_2+\Tilde{w}_u)}{w_2+2\Tilde{w}_u}<\frac{2(w_1+\Tilde{w}_u)}{w_1+2\Tilde{w}_u}$

\begin{align*}
(w_2+2\Tilde{w}_u-2(w_2+\Tilde{w}_u))(w_1\Tilde{y}_{1,2}+2\Tilde{w}_u\Tilde{y}_{1,2})-(w_2+2\Tilde{w}_u)w_1\Tilde{y}_{2,2}-\Tilde{w}_u(w_2+2\Tilde{w}_u)(\Tilde{y}_{3,2}+\Tilde{y}_{4,2})&=0\\
-(w_2+2\Tilde{w}_u)w_1\Tilde{y}_{1,2}+(w_2+2\Tilde{w}_u-2(w_2+\Tilde{w}_u))(w_1\Tilde{y}_{2,2}+2\Tilde{w}_u\Tilde{y}_{2,2})-\Tilde{w}_u(w_2+2\Tilde{w}_u)(\Tilde{y}_{3,2}+\Tilde{y}_{4,2})&=0\\
-\Tilde{w}_u(w_2+2\Tilde{w}_u)(\Tilde{y}_{1,2}+\Tilde{y}_{2,2})+(w_2+2\Tilde{w}_u-2(w_2+\Tilde{w}_u))(w_2\Tilde{y}_{3,2}+2\Tilde{w}_u\Tilde{y}_{3,2})-w_2(w_2+2\Tilde{w}_u)\Tilde{y}_{4,2}&=0\\
-\Tilde{w}_u(w_2+2\Tilde{w}_u)(\Tilde{y}_{1,2}+\Tilde{y}_{2,2})-w_2(w_2+2\Tilde{w}_u)\Tilde{y}_{3,2}+(w_2+2\Tilde{w}_u-2(w_2+\Tilde{w}_u))(w_2\Tilde{y}_{4,2}+2\Tilde{w}_u\Tilde{y}_{4,2})&=0\\
\end{align*}
 Summing up the first two expressions yields:
 \begin{align*}
 \begin{split}
     -2w_2w_1(\Tilde{y}_{1,2}+\Tilde{y}_{2,2})-2w_2\Tilde{w}_u(\Tilde{y}_{1,2}+\Tilde{y}_{2,2})-2\Tilde{w}_uw_1(\Tilde{y}_{1,2}+\Tilde{y}_{2,2})&=2\Tilde{w}_u(w_2+2\Tilde{w}_u)(\Tilde{y}_{3,2}+\Tilde{y}_{4,2})\\
     (\Tilde{y}_{1,2}+\Tilde{y}_{2,2})(-2w_2w_1-2w_2\Tilde{w}_u-2w_1\Tilde{w}_u)&=2\Tilde{w}_u(w_2+2\Tilde{w}_u)(\Tilde{y}_{3,2}+\Tilde{y}_{4,2})\\
     (\Tilde{y}_{1,2}+\Tilde{y}_{2,2})&=-\frac{\Tilde{w}_u(w_2+2\Tilde{w}_u)}{w_2w_1+w_2\Tilde{w}_u+\Tilde{w}_uw_1}(\Tilde{y}_{3,2}+\Tilde{y}_{4,2})
\end{split}
 \end{align*}
Similarly, the third and the fourth expressions are summed up, resulting in:
\begin{align*}
    -2w_2^2\Tilde{y}_{3,2}-4w_2\Tilde{w}_u\Tilde{y}_{3,2}-4w_2\Tilde{w}_u\Tilde{y}_{4,2}-2w_2^2\Tilde{y}_{4,2}&=2\Tilde{w}_u(w_2+2\Tilde{w}_u)(\Tilde{y}_{1,2}+\Tilde{y}_{2,2})\\
    -2w_2^2(\Tilde{y}_{3,2}+\Tilde{y}_{4,2})-4w_2\Tilde{w}_u(\Tilde{y}_{3,2}+\Tilde{y}_{4,2})&=2\Tilde{w}_u(w_2+2\Tilde{w}_u)(\Tilde{y}_{1,2}+\Tilde{y}_{2,2})\\
    (\Tilde{y}_{1,2}+\Tilde{y}_{2,2})&=-\frac{w_2}{\Tilde{w}_u}(\Tilde{y}_{3,2}+\Tilde{y}_{4,2})
\end{align*}
There are two possibilities as in previous steps: 
 $\Tilde{y}_{1,2}=-\Tilde{y}_{2,2}$ and $\Tilde{y}_{3,2}=-\Tilde{y}_{4,2}$ or $\frac{w_2}{\Tilde{w}_u}=\frac{\Tilde{w}_u(w_2+2\Tilde{w}_u)}{w_2w_1+w_2\Tilde{w}_u+\Tilde{w}_uw_1}$. Considering the second expression yields
\begin{align*}
    w_2^2w_1+w_2^2\Tilde{w}_u+w_2w_1\Tilde{w}_u-\Tilde{w}_u^2w_2-2\Tilde{w}_u^3&=0\\
   (w_2+\Tilde{w}_u)(w_1w_2+\Tilde{w}_u(w_2-2\Tilde{w}_u))&=0\\
   (w_2+\Tilde{w}_u)(w_1w_2+w_2\Tilde{w}_u-2\Tilde{w}_u^2)&=0.\\
\end{align*}
This equation can be satisfied for $w_1>0$, $w_2>0$ and $\Tilde{w}_u>0$ if and only if
\begin{align*}
    \Tilde{w}_u=\frac{1}{4}(\sqrt{w_2}\sqrt{8w_1+w_2}+w_2).
\end{align*}
For further analysis, the first two expressions are subtracted as:
\begin{align*}
    -2w_2\Tilde{w}_u(\Tilde{y}_{1,2}-\Tilde{y}_{2,2})+2w_1\Tilde{w}_u(\Tilde{y}_{1,2}-\Tilde{y}_{2,2})&=0\\
     (\Tilde{y}_{1,2}-\Tilde{y}_{2,2})(2w_1\Tilde{w}_u-2w_2\Tilde{w}_u)&=0
\end{align*}
Based on the information that $w_1\neq w_2$, the relationship between $\Tilde{y}_{1,2}$ and $\Tilde{y}_{2,2}$ yields $\Tilde{y}_{1,2}=\Tilde{y}_{2,2}$. 
Subtracting the third and fourth expressions does not provide further information about $\Tilde{y}_{3,2}$ and $\Tilde{y}_{4,2}$. Thus, the associated eigenvector can be written with the available information as:
\begin{align*}
\Tilde{\mathbf{y}}_2=
    \begin{bmatrix}
    \vspace{1mm}
\Tilde{y}_{1,2}\\\vspace{1mm}
\Tilde{y}_{1,2}\\\vspace{1mm}
\Tilde{y}_{3,2}\\\vspace{1mm}
\Tilde{y}_{4,2} \\
\end{bmatrix}
\end{align*}
\item $\Tilde{\mathbf{L}}\Tilde{\mathbf{y}}_3=\Tilde{\lambda}_3\Tilde{\mathbf{D}}\Tilde{\mathbf{y}}_3$ such that $\Tilde{\lambda}_3=\frac{2(w_1+\Tilde{w}_u)}{w_1+2\Tilde{w}_u}$ and $\frac{2(w_2+\Tilde{w}_u)}{w_2+2\Tilde{w}_u}<\frac{2(w_1+\Tilde{w}_u)}{w_1+2\Tilde{w}_u}$
    
\begin{align*}
(w_1+2\Tilde{w}_u-2(w_1+\Tilde{w}_u))(w_1\Tilde{y}_{1,3}+2\Tilde{w}_u\Tilde{y}_{1,3})-(w_1+2\Tilde{w}_u)w_1\Tilde{y}_{2,3}-\Tilde{w}_u(w_1+2\Tilde{w}_u)(\Tilde{y}_{3,3}+\Tilde{y}_{4,3})&=0\\
-(w_1+2\Tilde{w}_u)w_1\Tilde{y}_{1,3}+(w_1+2\Tilde{w}_u-2(w_1+\Tilde{w}_u))(w_1\Tilde{y}_{2,3}+2\Tilde{w}_u\Tilde{y}_{2,3})-\Tilde{w}_u(w_1+2\Tilde{w}_u)(\Tilde{y}_{3,3}+\Tilde{y}_{4,3})&=0\\
-\Tilde{w}_u(w_1+2\Tilde{w}_u)(\Tilde{y}_{1,3}+\Tilde{y}_{2,3})+(w_1+2\Tilde{w}_u-2(w_1+\Tilde{w}_u))(w_2\Tilde{y}_{3,3}+2\Tilde{w}_u\Tilde{y}_{3,3})-w_2(w_1+2\Tilde{w}_u)\Tilde{y}_{4,3}&=0\\
-\Tilde{w}_u(w_1+2\Tilde{w}_u)(\Tilde{y}_{1,3}+\Tilde{y}_{2,3})-w_2(w_1+2\Tilde{w}_u)\Tilde{y}_{3,3}+(w_1+2\Tilde{w}_u-2(w_1+\Tilde{w}_u))(w_2\Tilde{y}_{4,3}+2\Tilde{w}_u\Tilde{y}_{4,3})&=0\\
    \end{align*}
The first and the second expressions are summed, yielding:
\begin{align*}
    -2w_1^2(\Tilde{y}_{1,3}+\Tilde{y}_{2,3})-4\Tilde{w}_uw_1(\Tilde{y}_{1,3}+\Tilde{y}_{2,3})&=(2\Tilde{w}_uw_1+4\Tilde{w}_u^2)(\Tilde{y}_{3,3}+\Tilde{y}_{4,3})\\
    (\Tilde{y}_{1,3}+\Tilde{y}_{2,3})(-2w_1^2-4\Tilde{w}_uw_1)&=(2\Tilde{w}_uw_1+4\Tilde{w}_u^2)(\Tilde{y}_{3,3}+\Tilde{y}_{4,3})\\
    (\Tilde{y}_{1,3}+\Tilde{y}_{2,3})&=-\frac{\Tilde{w}_u}{w_1}(\Tilde{y}_{3,3}+\Tilde{y}_{4,3}),
\end{align*}
Then, the third and the fourth expressions are summed in the same way
\begin{align*}
    -2w_1w_2\Tilde{y}_{3,3}-2w_1w_2\Tilde{y}_{4,3}-2w_1\Tilde{w}_u\Tilde{y}_{3,3}-2w_1\Tilde{w}_u\Tilde{y}_{4,3}-2w_2\Tilde{w}_u\Tilde{y}_{4,3}-2w_2\Tilde{w}_u\Tilde{y}_{3,3}&=2\Tilde{w}_u(w_1+2\Tilde{w}_u)(\Tilde{y}_{1,3}+\Tilde{y}_{2,3})\\
    (\Tilde{y}_{3,3}+\Tilde{y}_{4,3})(-2w_1w_2-2w_1\Tilde{w}_u-2w_2\Tilde{w}_u)&=2\Tilde{w}_u(w_1+2\Tilde{w}_u)(\Tilde{y}_{1,3}+\Tilde{y}_{2,3})\\
    (\Tilde{y}_{1,3}+\Tilde{y}_{2,3})&=-\frac{w_1w_2+w_1\Tilde{w}_u+w_2\Tilde{w}_u}{\Tilde{w}_u(w_1+2\Tilde{w}_u)}(\Tilde{y}_{3,3}+\Tilde{y}_{4,3})\\
\end{align*}

The two possibilites can be summarized as: $\Tilde{y}_{1,3}=-\Tilde{y}_{2,3}$ and $\Tilde{y}_{3,3}=-\Tilde{y}_{4,3}$ or $\frac{\Tilde{w}_u}{w_1}=\frac{w_2w_1+w_2\Tilde{w}_u+\Tilde{w}_uw_1}{\Tilde{w}_u(w_1+2\Tilde{w}_u)}$. The second condition yields
\begin{align*}
  0&=w_1^2w_2+w_1^2\Tilde{w}_u+w_2w_1\Tilde{w}_u-\Tilde{w}_u^2w_1-2\Tilde{w}_u^3\\
\end{align*}
whose solution is 
\begin{align*}
      w_1&=\frac{\big(\Tilde{w}_u(w_2+\Tilde{w}_u)-2\Tilde{w}_u^2\big)+\sqrt{(2\Tilde{w}_u^2-\Tilde{w}_u(w_2+\Tilde{w}_u))^2-8\Tilde{w}_u^3(-w_2-\Tilde{w}_u)}}{2(-w_2-\Tilde{w}_u)}\\
\end{align*}
or
\begin{align*}
      w_1&=\frac{\big(\Tilde{w}_u(w_2+\Tilde{w}_u)-2\Tilde{w}_u^2\big)-\sqrt{(2\Tilde{w}_u^2-\Tilde{w}_u(w_2+\Tilde{w}_u))^2-8\Tilde{w}_u^3(-w_2-\Tilde{w}_u)}}{2(-w_2-\Tilde{w}_u)}.\\
\end{align*}
Further, subtracting the first two expressions results in $0=0$ and does not provide further information. To obtain further information, the third and the fourth expressions are subtracted as
\begin{align*}
    -2w_1\Tilde{w}_u(\Tilde{y}_{3,3}-\Tilde{y}_{4,3})+2\Tilde{w}_uw_2(\Tilde{y}_{3,3}-\Tilde{y}_{4,3})&=0\\
    (\Tilde{y}_{3,3}-\Tilde{y}_{4,3})(-2w_1\Tilde{w}_u+2w_2\Tilde{w}_u)&=0
\end{align*}
Based on the knowledge that $w_1\neq w_2$, $\Tilde{y}_{3,3}=\Tilde{y}_{4,3}=0$, it follows that $\Tilde{y}_{1,3}=-\Tilde{y}_{2,3}$. According to the extracted information, the eigenvector can be written as:
\begin{align*}
\Tilde{\mathbf{y}}_3=
    \begin{bmatrix}
    \vspace{1mm}
\quad\Tilde{y}_{1,3}\\\vspace{1mm}
-\Tilde{y}_{1,3}\\\vspace{1mm}
\quad0\\\vspace{1mm}
\quad0 \\
\end{bmatrix}
\end{align*}
\end{itemize}
To summarize, the eigenvectors associated with the eigenvalues $\Tilde{\lambda}_{0}=0$, $\Tilde{\lambda}_{1}=\frac{2(w_1\Tilde{w}_u+w_2\Tilde{w}_u+4\Tilde{w}_u^2)}{(w_1+2\Tilde{w}_u)(w_2+2\Tilde{w}_u)}$, $\Tilde{\lambda}_{2}=\frac{2(w_2+\Tilde{w}_u)}{w_2+2\Tilde{w}_u}$ and $\Tilde{\lambda}_{3}=\frac{2(w_1+\Tilde{w}_u)}{w_1+2\Tilde{w}_u}$ are
\begin{align*}
\Tilde{\mathbf{y}}_0=
\begin{bmatrix}
\vspace{1mm}
\Tilde{y}_{1,0}\\\vspace{1mm}
\Tilde{y}_{1,0}\\\vspace{1mm}
\Tilde{y}_{1,0}\\\vspace{1mm}
\Tilde{y}_{1,0}\\
\end{bmatrix}
\hspace{5mm}\Tilde{\mathbf{y}}_1=
\begin{bmatrix}
\vspace{1mm}
\Tilde{y}_{1,1}\\\vspace{1mm}
\Tilde{y}_{2,1}\\\vspace{1mm}
\Tilde{y}_{3,1}\\\vspace{1mm}
\Tilde{y}_{4,1}\\
\end{bmatrix}
\hspace{5mm}\Tilde{\mathbf{y}}_2=
\begin{bmatrix}
\vspace{1mm}
\Tilde{y}_{1,2}\\\vspace{1mm}
\Tilde{y}_{1,2}\\\vspace{1mm}
\Tilde{y}_{3,2}\\\vspace{1mm}
\Tilde{y}_{4,2}\\
\end{bmatrix}
\hspace{5mm}\Tilde{\mathbf{y}}_3=
\begin{bmatrix}
\vspace{1mm}
\quad\Tilde{y}_{1,3}\\\vspace{1mm}
-\Tilde{y}_{1,3}\\\vspace{1mm}
\quad0\\\vspace{1mm}
\quad0\\
\end{bmatrix}
\end{align*}

By definition, the eigenvectors of a real and symmetric matrix must be orthogonal. To compute the eigenvectors, the orthogonality information is analyzed for couples of eigenvectors as
\begin{itemize}
    \item $\Tilde{\mathbf{y}}_{0}^{\top}\Tilde{\mathbf{y}}_{2}=0$
    \begin{align*}
2\Tilde{y}_{1,0}\Tilde{y}_{1,2}+\Tilde{y}_{1,0}(\Tilde{y}_{3,2}+\Tilde{y}_{4,2})&=0\\
2\Tilde{y}_{1,2}=-(\Tilde{y}_{3,2}+\Tilde{y}_{4,2})
\end{align*}
Here, based on the knowledge that $2\Tilde{y}_{1,2}=-\frac{w_2}{\Tilde{w}_u}(\Tilde{y}_{3,2}+\Tilde{y}_{4,2})$ and $w_2\neq\Tilde{w}_u$ the mappings yields $\Tilde{y}_{1,2}=0$ and $\Tilde{y}_{3,2}=-\Tilde{y}_{4,2}$.

\item $\Tilde{\mathbf{y}}_{1}^{\top}\Tilde{\mathbf{y}}_{3}=0$
\begin{align*}
\Tilde{y}_{1,1}\Tilde{y}_{1,3}-\Tilde{y}_{2,1}\Tilde{y}_{1,3}&=0\\
\Tilde{y}_{1,1}&=\Tilde{y}_{2,1}
\end{align*}

\item $\Tilde{\mathbf{y}}_{1}^{\top}\Tilde{\mathbf{y}}_{2}=0$
\begin{align*}
\Tilde{y}_{3,1}\Tilde{y}_{3,2}-\Tilde{y}_{4,1}\Tilde{y}_{3,2}&=0\\
\Tilde{y}_{3,1}&=\Tilde{y}_{4,1}\\
\end{align*}
\end{itemize}

Consequently, the eigenvectors associated with the eigenvalues in ascending order ${\Tilde{\lambda}_{0}=0}$, ${\Tilde{\lambda}_{1}=\frac{2(w_1\Tilde{w}_u+w_2\Tilde{w}_u+4\Tilde{w}_u^2)}{(w_1+2\Tilde{w}_u)(w_2+2\Tilde{w}_u)}}$, ${\Tilde{\lambda}_{2}=\frac{2(w_2+\Tilde{w}_u)}{w_2+2\Tilde{w}_u}}$ and ${\Tilde{\lambda}_{3}=\frac{2(w_1+\Tilde{w}_u)}{w_1+2\Tilde{w}_u}}$ are 
\begin{align*}
\Tilde{\mathbf{y}}_0=
\begin{bmatrix}
\vspace{1mm}
\Tilde{y}_{1,0}\\\vspace{1mm}
\Tilde{y}_{1,0}\\\vspace{1mm}
\Tilde{y}_{1,0}\\\vspace{1mm}
\Tilde{y}_{1,0}\\
\end{bmatrix}
\hspace{5mm}\Tilde{\mathbf{y}}_1=
\begin{bmatrix}
\vspace{1mm}
\Tilde{y}_{1,1}\\\vspace{1mm}
\Tilde{y}_{1,1}\\\vspace{1mm}
\Tilde{y}_{3,1}\\\vspace{1mm}
\Tilde{y}_{3,1}\\
\end{bmatrix}
\hspace{5mm}\Tilde{\mathbf{y}}_2=
\begin{bmatrix}
\vspace{1mm}
\quad0\\\vspace{1mm}
\quad0\\\vspace{1mm}
\quad\Tilde{y}_{3,2}\\\vspace{1mm}
-\Tilde{y}_{3,2}\\
\end{bmatrix}
\hspace{5mm}\Tilde{\mathbf{y}}_3=
\begin{bmatrix}
\vspace{1mm}
\quad\Tilde{y}_{1,3}\\\vspace{1mm}
-\Tilde{y}_{1,3}\\\vspace{1mm}
\quad0\\\vspace{1mm}
\quad0\\
\end{bmatrix}
\end{align*}
which can straight-forwardly be written for the eigenvalues in ascending order $\Tilde{\lambda}_{0}=0$, ${\Tilde{\lambda}_{1}=\frac{2(w_1\Tilde{w}_u+w_2\Tilde{w}_u+4\Tilde{w}_u^2)}{(w_1+2\Tilde{w}_u)(w_2+2\Tilde{w}_u)}}$,
${\Tilde{\lambda}_{2}=\frac{2(w_1+\Tilde{w}_u)}{w_1+2\Tilde{w}_u}}$ and
${\Tilde{\lambda}_{3}=\frac{2(w_2+\Tilde{w}_u)}{w_2+2\Tilde{w}_u}}$ as
\begin{align*}
\Tilde{\mathbf{y}}_0=
\begin{bmatrix}
\vspace{1mm}
\Tilde{y}_{1,0}\\\vspace{1mm}
\Tilde{y}_{1,0}\\\vspace{1mm}
\Tilde{y}_{1,0}\\\vspace{1mm}
\Tilde{y}_{1,0}\\
\end{bmatrix}
\hspace{5mm}\Tilde{\mathbf{y}}_1=
\begin{bmatrix}
\vspace{1mm}
\Tilde{y}_{1,1}\\\vspace{1mm}
\Tilde{y}_{1,1}\\\vspace{1mm}
\Tilde{y}_{3,1}\\\vspace{1mm}
\Tilde{y}_{3,1}\\
\end{bmatrix}
\hspace{5mm}\Tilde{\mathbf{y}}_2=
\begin{bmatrix}
\vspace{1mm}
\quad\Tilde{y}_{1,2}\\\vspace{1mm}
-\Tilde{y}_{1,2}\\\vspace{1mm}
\quad0\\\vspace{1mm}
\quad0\\
\end{bmatrix}
\hspace{5mm}\Tilde{\mathbf{y}}_3=
\begin{bmatrix}
\vspace{1mm}
\quad0\\\vspace{1mm}
\quad0\\\vspace{1mm}
\quad\Tilde{y}_{3,3}\\\vspace{1mm}
-\Tilde{y}_{3,3}\\
\end{bmatrix}.
\end{align*}

\subsection{C.3~Eigen-decomposotion for Affinity Matrix of Uncorrelated Blocks}\label{sec:EigenDecUncorrelated}
Again, let $\mathbf{W}\in\mathbb{R}^{n\times n}$ and the corresponding Laplacian $\mathbf{W}\in\mathbb{R}^{n\times n}$ consist of $k=2$ blocks, i.e.,
\begin{align*}
\mathbf{W}=
\begin{bmatrix}
\vspace{1mm}
0 & w_1 & 0&0\\
w_1 & 0 &0 &0\\
 0&  0& 0 &w_2\\
0 & 0 & w_2& 0
\end{bmatrix}
\hspace{5mm}
\mathbf{L}=
\begin{bmatrix}
\vspace{1mm}
w_1 & -w_1 & 0&0\\
-w_1 & w_1 &0 &0\\
 0&  0& w_2 &-w_2\\
0 & 0 & -w_2& w_2
\end{bmatrix},
\end{align*}
where $w_i$,$i=1,2$ denotes correlation coefficient of the $i$th block. To estimate the eigenvalues, the following expression is considered
\begin{align*}
    \mathrm{det}(\mathbf{L}-\lambda\mathbf{I})=0,
\end{align*}
which can equivalently be written in matrix form as
\begin{align*}
\begin{vmatrix}
\vspace{1mm}
w_1-\lambda & -w_1 & 0&0\\
-w_1 & w_1-\lambda &0 &0\\
 0&  0& w_2-\lambda &-w_2\\
0 & 0 & -w_2& w_2-\lambda
\end{vmatrix}=0.
\end{align*}
Using determinant properties of block matrices, it follows that [62],
\begin{align*}
\begin{vmatrix}
w_1-\lambda & -w_1\\
-w_1 & w_1-\lambda\\
\end{vmatrix}
\begin{vmatrix}
w_2-\lambda & -w_2\\
-w_2 & w_2-\lambda\\
\end{vmatrix}=0
\end{align*}
After solving this equation, the eigenvalues are equal to $\lambda_{0,1}=0$, $\lambda_2=2w_1$ and $\lambda_3=2w_2$. To compute the eigenvectors, the standard eigen-decomposition is considered
\begin{align*}
    \mathbf{L}\mathbf{y}_i=\lambda_i\mathbf{y}_i,
\end{align*}
where $\mathbf{y}_i$ denotes the eigenvector associated with the $i$th eigenvalue $\lambda_i$. The eigen-problem can equivalently be written in the matrix form as
\begin{align*}
\begin{bmatrix}
\vspace{1mm}
w_1-\lambda_i & -w_1 & 0&0\\\vspace{1mm}
-w_1 & w_1-\lambda_i &0 &0\\\vspace{1mm}
 0&  0& w_2-\lambda_i &-w_2\\\vspace{1mm}
0 & 0 & -w_2& w_2-\lambda_i
\end{bmatrix} 
\begin{bmatrix}
\vspace{1mm}
y_{1,i}\\\vspace{1mm}
y_{2,i}\\\vspace{1mm}
y_{3,i}\\\vspace{1mm}
y_{4,i}
\end{bmatrix}=0,
\end{align*}
where $y_{j,i},j=1,\dots,4$ denotes the $j$th mapping point in the $i$th eigenvector $\mathbf{y}_i$. Substituting each $\lambda_i\in[\lambda_0,\dots,\lambda_3]$ leads to the eigenvectors  
\begin{align*}
\begin{bmatrix}
\vspace{1mm}
y_{1,0}\\\vspace{1mm}
y_{1,0}\\\vspace{1mm}
y_{3,0}\\\vspace{1mm}
y_{3,0}\\
\end{bmatrix}\hspace{5mm}
\begin{bmatrix}
\vspace{1mm}
y_{1,1}\\\vspace{1mm}
y_{1,1}\\\vspace{1mm}
y_{3,1}\\\vspace{1mm}
y_{3,1}\\
\end{bmatrix}\hspace{5mm}
\begin{bmatrix}
\vspace{1mm}
\quad y_{1,2}\\\vspace{1mm}
-y_{1,2}\\\vspace{1mm}
\quad y_{3,2}\\\vspace{1mm}
-y_{3,2}\\
\end{bmatrix}\hspace{5mm}
\begin{bmatrix}
\vspace{1mm}
\quad y_{1,3}\\\vspace{1mm}
-y_{1,3}\\\vspace{1mm}
\quad y_{3,3}\\\vspace{1mm}
-y_{3,3}\\
\end{bmatrix}.
\end{align*}
As can be seen, the solution is identical to the eigenvectors of the generalized eigen-decomposition problem. Therefore, the possible eigenvectors that are computed using the orthogonality information are identical to the ones that are detailed in Section~C.1.

\subsection{C.4~Eigen-decomposition for Affinity Matrix of Correlated Blocks}\label{sec:EigenDeccorrelated}
Let $\mathbf{W}\in\mathbb{R}^{n\times n}$ denote a block zero-diagonal symmetric
affinity matrix for $k$ blocks which does not include correlations between different blocks, i.e., 
\begin{align*}
\begin{bmatrix}
0 & w_1 &\dots  &  &  & &\\
w_1 & 0 &\dots  &  &  & &\\
\vphantom{\int\limits^x}\smash{\vdots} &  \vphantom{\int\limits^x}\smash{\vdots}& 0 & w_2 &  & &\\
 &  & w_2  & 0& &  & \\
  &   &   &  & \ddots &\vphantom{\int\limits^x}\smash{\vdots} &\vphantom{\int\limits^x}\smash{\vdots}\\
  &  &   &  &\dots  &0 &w_k\\
  &   &   &  &\dots &w_k &0\\
\end{bmatrix}
\end{align*}
The associated Laplacian has the same block structure for $\mathbf{L}_1,\dots,\mathbf{L}_k$. To estimate eigenvalues of $\mathbf{L}$, the expression ${\mathrm{det}(\mathbf{L}-\lambda\mathbf{I})=0}$ can be written in terms of the distinct determinants as
\begin{align*}
    \mathrm{det}(\mathbf{L}-\lambda\mathbf{I})=\mathrm{det}(\mathbf{L}_1-\lambda\mathbf{I})\dots\mathrm{det}(\mathbf{L}_k-\lambda\mathbf{I}),
\end{align*}
where $\mathbf{L}_k$ denotes the Laplacian matrix of $k$th block and $\mathbf{I}$ is the identity matrix. After solving this equation, for the $i$th block, the eigenvalues are equal to ${\lambda_0=0}$ and ${\lambda_1=2w_i}$. Now assume that first and second blocks are correlated because of Type II outliers that have been added to the data set. Because the correlations are the effect of the outliers, we call them undesired correlations. The resulting corrupted affinity matrix $\Tilde{\mathbf{W}}\in\mathbb{R}^{n\times n}$ consists of $k$ blocks, i.e.,
\begin{align*}
\begin{bmatrix}
0 & w_1 & \Tilde{w}_u &\Tilde{w}_u  &\dots  & &\\
w_1 & 0 & \Tilde{w}_u& \Tilde{w}_u &\dots   & &\\
\Tilde{w}_u & \Tilde{w}_u & 0 & w_2 &  & &\\
\Tilde{w}_u & \Tilde{w}_u & w_2  & 0& &  & \\
\vdots  & \vdots  &   &  & \ddots &\vdots &\vdots\\
 &  &   &  &\dots  &0 &w_k\\
 &   &   &  &\dots &w_k &0\\
\end{bmatrix},
\end{align*}
where $\Tilde{w}_u$ denotes correlation coefficient belonging to the undesired correlations for $\Tilde{w}_u\neq 0$. For simplicity, let $\Tilde{\mathbf{L}}_{12}$ be the Laplacian matrix of correlated blocks. Hence, the eigenvalues of the corrupted Laplacian $\Tilde{\mathbf{L}}$ are estimated using distinct determinants as
\begin{align*}
    \mathrm{det}(\Tilde{\mathbf{L}}-\Tilde{\lambda}\Tilde{\mathbf{I}})=\mathrm{det}(\Tilde{\mathbf{L}}_{12}-\Tilde{\lambda}\Tilde{\mathbf{I}})\dots\mathrm{det}(\mathbf{L}_k-\lambda\mathbf{I}),
\end{align*}
where $\Tilde{\mathbf{L}}_{12}$ denotes the Laplacian matrix associated with the block that is combination of blocks $i=1$ and $j=2$. To study the effect of outliers on the eigen-decomposition, the above equation shows that analyzing ${\mathrm{det}(\Tilde{\mathbf{L}}_{12}-\Tilde{\lambda}\Tilde{\mathbf{I}})=0}$ is sufficient, which can equivalently be written as
\begin{align*}
\begin{vmatrix}
w_1+2\Tilde{w}_u-\lambda & -w_1 & -\Tilde{w}_u & -\Tilde{w}_u\\
-w_1 & w_1+2\Tilde{w}_u-\lambda & -\Tilde{w}_u & -\Tilde{w}_u\\
-\Tilde{w}_u & -\Tilde{w}_u & w_2+2\Tilde{w}_u-\lambda & -w_2\\
-\Tilde{w}_u & -\Tilde{w}_u & -w_2 & w_2+2\Tilde{w}_u-\lambda\\
\end{vmatrix}
=0.
\end{align*}
To simplify the determinant problem for block matrices as in [62], the matrix $\Tilde{\mathbf{L}}_{12}-\Tilde{\lambda}\Tilde{\mathbf{I}}$ must be commutative such that ${(-\Tilde{\mathbf{W}}_{u})(\Tilde{\mathbf{L}}_{2}-\Tilde{\lambda}\Tilde{\mathbf{I}})=(\Tilde{\mathbf{L}}_{2}-\Tilde{\lambda}\Tilde{\mathbf{I}})(-\Tilde{\mathbf{W}}_{u})}$. Accordingly, the matrix $\Tilde{\mathbf{L}}_{12}-\Tilde{\lambda}\Tilde{\mathbf{I}}$ is commutative if it satisfies
\begin{align*}
\begin{bmatrix}
-\Tilde{w}_u & -\Tilde{w}_u \\
-\Tilde{w}_u & -\Tilde{w}_u \\
\end{bmatrix}
\begin{bmatrix}
w_2+c & -w_2 \\
-w_2 & w_2+c \\
\end{bmatrix}=
\begin{bmatrix}
w_2+c & -w_2 \\
-w_2 & w_2+c \\
\end{bmatrix}
\begin{bmatrix}
-\Tilde{w}_u & -\Tilde{w}_u \\
-\Tilde{w}_u & -\Tilde{w}_u \\
\end{bmatrix}
\end{align*}
where $c=2\Tilde{w}_u-\Tilde{\lambda}$.
After solving both sides the equality yields,
\begin{align*}
\begin{bmatrix}
-\Tilde{w}_u(w_2+c)+w_2\Tilde{w}_u & w_2\Tilde{w}_u-\Tilde{w}_u(w_2+c) \\
-\Tilde{w}_u(w_2+c)+w_2\Tilde{w}_u & w_2\Tilde{w}_u-\Tilde{w}_u(w_2+c) \\
\end{bmatrix}=
\begin{bmatrix}
-\Tilde{w}_u(w_2+c)+w_2\Tilde{w}_u & w_2\Tilde{w}_u-\Tilde{w}_u(w_2+c) \\
-\Tilde{w}_u(w_2+c)+w_2\Tilde{w}_u & w_2\Tilde{w}_u-\Tilde{w}_u(w_2+c) \\
\end{bmatrix}.
\end{align*}
For a commutative matrix, it holds that
\begin{align*}
    \mathrm{det}(\Tilde{\mathbf{L}}_{12}-\Tilde{\lambda}\mathbf{I})=\mathrm{det}\Big((\Tilde{\mathbf{L}}_{1}-\Tilde{\lambda}\mathbf{I})(\Tilde{\mathbf{L}}_{2}-\Tilde{\lambda}\mathbf{I})-\Tilde{\mathbf{W}}_u^2\Big)=0.
\end{align*}
The solution, therefore, results in a coupled set of equations as follows:
\begin{align*}
0 &=\Bigg|
\begin{bmatrix}
w_1+c & -w_1\\
-w_1  & w_1+c\\
\end{bmatrix}
\begin{bmatrix}
w_2+c & -w_2\\
-w_2  & w_2+c\\
\end{bmatrix}-
\begin{bmatrix}
2\Tilde{w}_u^2 & 2\Tilde{w}_u^2\\
2\Tilde{w}_u^2 & 2\Tilde{w}_u^2\\
\end{bmatrix}\Bigg|\\
&=
\Bigg|
\begin{bmatrix}
w_1w_2+w_1c+cw_2+c^2+w_1w_2-2\Tilde{w}_u^2 & -w_1w_2-w_2c-w_1w_2-w_1c-2\Tilde{w}_u^2\\
 -w_1w_2-w_1c-w_1w_2-w_2c-2\Tilde{w}_u^2  & w_1w_2+w_1w_2+w_1c+cw_2+c^2-2\Tilde{w}_u^2\\
\end{bmatrix}
\Bigg|\\
&=
\big(2w_1w_2+w_1c+cw_2+c^2-2\Tilde{w}_u^2\big)^2-\big( -2w_1w_2-w_2c-w_1c-2\Tilde{w}_u^2\big)^2\\
&=
\big(4w_1w_2+2w_1c+2w_2c+c^2\big)\big(c^2-4\Tilde{w}_u^2\big)\\
&=\big(c+2w_1\big)\big(c+2w_2\big)\big(c-2\Tilde{w}_u\big)\big(c+2\Tilde{w}_u)
\end{align*}
For $w_2\leq w_1$, substituting $c=2\Tilde{w}_u-\Tilde{\lambda}$ yields the eigenvalues in ascending order
\begin{align*}
    \Tilde{\lambda}_0=0 \hspace{5mm} \Tilde{\lambda}_1=4\Tilde{w}_u,\hspace{5mm}\Tilde{\lambda}_2=2w_2+2\Tilde{w}_u \hspace{5mm}\Tilde{\lambda}_{3}=2w_1 +2\Tilde{w}_u
\end{align*}
or
\begin{align*}
    \Tilde{\lambda}_0=0 \hspace{5mm} \Tilde{\lambda}_1=4\Tilde{w}_u,\hspace{5mm}\Tilde{\lambda}_2=2w_1+2\Tilde{w}_u \hspace{5mm}\Tilde{\lambda}_{3}=2w_2 +2\Tilde{w}_u
\end{align*}
otherwise.
To analyze the effect of outliers on eigenvectors, the standard eigenvalue problem is considered
\begin{align*}
    \Tilde{\mathbf{L}}\Tilde{\mathbf{y}}_i=\Tilde{\lambda}_i\Tilde{\mathbf{y}}_i
\end{align*}
where $\Tilde{\mathbf{y}}_i$ denotes the eigenvector associated with the $i$th eigenvalue $\Tilde{\lambda}_i$. The expression can equivalently be written as
\begin{align*}
\begin{bmatrix}
\vspace{1mm}
w_1+2\Tilde{w}_u-\Tilde{\lambda}_{i} & -w_1 & -\Tilde{w}_u & -\Tilde{w}_u\\\vspace{1mm}
-w_1 & w_1+2\Tilde{w}_u-\Tilde{\lambda}_{i} & -\Tilde{w}_u & -\Tilde{w}_u\\\vspace{1mm}
-\Tilde{w}_u & -\Tilde{w}_u & w_2+2\Tilde{w}_u-\Tilde{\lambda}_{i} & -w_2\\\vspace{1mm}
-\Tilde{w}_u & -\Tilde{w}_u & -w_2 & w_2+2\Tilde{w}_u-\Tilde{\lambda}_{i}\\
\end{bmatrix}
\begin{bmatrix}
\vspace{1mm}
\Tilde{y}_{1,i}\\\vspace{1mm}
\Tilde{y}_{2,i}\\\vspace{1mm}
\Tilde{y}_{3,i}\\\vspace{1mm}
\Tilde{y}_{4,i}\\
\end{bmatrix}=0,
\end{align*}
where $\Tilde{y}_{j,i}$ denotes the $j$th mapping point in the eigenvector associated with the $i$th eigenvalue $\Tilde{\lambda}_i$. To compute the eigenvectors, each $\Tilde{\lambda}_i\in[\Tilde{\lambda}_0,\dots,\Tilde{\lambda}_3]$ is substituted in $\Tilde{\mathbf{L}}\Tilde{\mathbf{y}}_i=\Tilde{\lambda}_i\Tilde{\mathbf{y}}_i$ as follows:

\begin{itemize}
    \item $\Tilde{\mathbf{L}}\Tilde{\mathbf{y}}_0=\Tilde{\lambda}_0\Tilde{\mathbf{y}}_0$ such that $\lambda_{0}=0$
    \begin{align*}
        (w_1+2\Tilde{w}_u)\Tilde{y}_{1,0}-w_1\Tilde{y}_{2,0}-\Tilde{w}_u(\Tilde{y}_{3,0}+\Tilde{y}_{4,0})=0\\
        -w_1\Tilde{y}_{1,0}+(w_1+2\Tilde{w}_u)\Tilde{y}_{2,0}-\Tilde{w}_u(\Tilde{y}_{3,0}+\Tilde{y}_{4,0})=0\\
        -\Tilde{w}_u(\Tilde{y}_{1,0}+\Tilde{y}_{2,0})+(w_2+2\Tilde{w}_u)\Tilde{y}_{3,0}-w_2\Tilde{y}_{4,0}=0\\
        -\Tilde{w}_u(\Tilde{y}_{1,0}+\Tilde{y}_{2,0})-w_2\Tilde{y}_{3,0}+(w_2+2\Tilde{w}_u)\Tilde{y}_{4,0}=0\\
    \end{align*}
  The second expression is subtracted from the first one resulting in
  \begin{align*}
      2w_1\Tilde{y}_{1,0}+2\Tilde{w}_u(\Tilde{y}_{1,0}-\Tilde{y}_{2,0})-2w_1\Tilde{y}_{2,0}&=0\\
      2w_1(\Tilde{y}_{1,0}-\Tilde{y}_{2,0})+2\Tilde{w}_u(\Tilde{y}_{1,0}-\Tilde{y}_{2,0})&=0\\
      2(w_1+\Tilde{w}_u)(\Tilde{y}_{1,0}-\Tilde{y}_{2,0})&=0.\\
  \end{align*}
  As $w_1>\Tilde{w}_u>0$, $\Tilde{y}_{1,0}=\Tilde{y}_{2,0}$. Similarly, the fourth expression is subtracted from the third one:
   \begin{align*}
      2w_2\Tilde{y}_{3,0}+2\Tilde{w}_u(\Tilde{y}_{3,0}-\Tilde{y}_{4,0})-2w_2\Tilde{y}_{4,0}&=0\\
      2w_2(\Tilde{y}_{3,0}-\Tilde{y}_{4,0})+2\Tilde{w}_u(\Tilde{y}_{3,0}-\Tilde{y}_{4,0})&=0\\
      2(w_2+\Tilde{w}_u)(\Tilde{y}_{3,0}-\Tilde{y}_{4,0})&=0\\
  \end{align*}
  Since $w_2>\Tilde{w}_u>0$, $\Tilde{y}_{3,0}=\Tilde{y}_{4,0}$. Lastly, the first and the second expressions are summed up as follows:
  \begin{align*}
      2\Tilde{w}_u(\Tilde{y}_{1,0}+\Tilde{y}_{2,0})-2\Tilde{w}_u(\Tilde{y}_{3,0}+\Tilde{y}_{4,0})&=0\\
      \Tilde{y}_{1,0}+\Tilde{y}_{2,0}&=\Tilde{y}_{3,0}+\Tilde{y}_{4,0}\\
  \end{align*}
Substituting $\Tilde{y}_{1,0}=\Tilde{y}_{2,0}$ and $\Tilde{y}_{3,0}=\Tilde{y}_{4,0}$ yields $\Tilde{y}_{1,0}=\Tilde{y}_{2,0}=\Tilde{y}_{3,0}=\Tilde{y}_{4,0}$. Thus, the eigenvector associated with the smallest eigenvalue $\Tilde{\lambda}_0$ is
\begin{align*}
\Tilde{\mathbf{y}}_0=
    \begin{bmatrix}
    \vspace{1mm}
\Tilde{y}_{1,0}\\\vspace{1mm}
\Tilde{y}_{1,0}\\\vspace{1mm}
\Tilde{y}_{1,0} \\\vspace{1mm}
\Tilde{y}_{1,0} \\
\end{bmatrix}
\end{align*}
 \item $\Tilde{\mathbf{L}}\Tilde{\mathbf{y}}_1=\Tilde{\lambda}_1\Tilde{\mathbf{y}}_1$ such that $\lambda_{1}=4\Tilde{w}_u$
    \begin{align*}
        (w_1-2\Tilde{w}_u)\Tilde{y}_{1,1}-w_1\Tilde{y}_{2,1}-\Tilde{w}_u(\Tilde{y}_{3,1}+\Tilde{y}_{4,1})&=0\\
        -w_1\Tilde{y}_{1,1}+(w_1-2\Tilde{w}_u)\Tilde{y}_{2,1}-\Tilde{w}_u(\Tilde{y}_{3,1}+\Tilde{y}_{4,1})&=0\\
        -\Tilde{w}_u(\Tilde{y}_{1,1}+\Tilde{y}_{2,1})+(w_2-2\Tilde{w}_u)\Tilde{y}_{3,1}-w_2\Tilde{y}_{4,1}&=0\\
        -\Tilde{w}_u(\Tilde{y}_{1,1}+\Tilde{y}_{2,1})-w_2\Tilde{y}_{3,1}+(w_2-2\Tilde{w}_u)\Tilde{y}_{4,1}&=0\\
    \end{align*}

Again, the second expression is subtracted from the first one
\begin{align*}
    2w_1\Tilde{y}_{1,1}-2\Tilde{w}_u \Tilde{y}_{1,1}+2\Tilde{w}_u \Tilde{y}_{2,1}-2w_1w_2&=0\\
    2w_1(\Tilde{y}_{1,1}-\Tilde{y}_{2,1})-2\Tilde{w}_u(\Tilde{y}_{1,1}-\Tilde{y}_{2,1})&=0\\
    2(w_1-\Tilde{w}_u)(\Tilde{y}_{1,1}-\Tilde{y}_{2,1})=0.\\
\end{align*}
Since $w_1>\Tilde{w}_u$, $\Tilde{y}_{1,1}=\Tilde{y}_{2,1}$. Similarly, the fourth expression is subtracted from the third one as follows
\begin{align*}
    2w_2\Tilde{y}_{3,1}-2\Tilde{w}_u\Tilde{y}_{3,1}+2\Tilde{w}_u\Tilde{y}_{4,1}-2w_2\Tilde{y}_{4,1}&=0\\
    2w_2(\Tilde{y}_{3,1}-\Tilde{y}_{4,1})-2\Tilde{w}_u(\Tilde{y}_{3,1}-\Tilde{y}_{4,1})&=0\\
    2(w_2-\Tilde{w}_u)(\Tilde{y}_{3,1}-\Tilde{y}_{4,1})&=0.\\
\end{align*}
Based on the knowledge that $w_2>\Tilde{w}_u$, $\Tilde{y}_{3,1}=\Tilde{y}_{4,1}$. To continue, the first and the second expressions are summed as:
\begin{align*}
    -2\Tilde{w}_u\Tilde{y}_{1,1}-2\Tilde{w}_u\Tilde{y}_{2,1}-2\Tilde{w}_u(\Tilde{y}_{3,1}+\Tilde{y}_{4,1})&=0\\
    -2\Tilde{w}_u(\Tilde{y}_{1,1}+\Tilde{y}_{2,1}+\Tilde{y}_{3,1}+\Tilde{y}_{4,1})&=0\\
\end{align*}
Now, it straightforwardly follows that $\Tilde{y}_{1,1}=\Tilde{y}_{2,1}=-\Tilde{y}_{3,1}=-\Tilde{y}_{4,1}$. Therefore, the associated eigenvector can be written as:
\begin{align*}
\Tilde{\mathbf{y}}_1=
    \begin{bmatrix}
    \vspace{1mm}
    \quad\Tilde{y}_{1,1}\\\vspace{1mm}
    \quad\Tilde{y}_{1,1}\\\vspace{1mm}
    -\Tilde{y}_{1,1}\\\vspace{1mm}
    -\Tilde{y}_{1,1}\\
    \end{bmatrix}
\end{align*}
 \item $\Tilde{\mathbf{L}}\Tilde{\mathbf{y}}_2=\Tilde{\lambda}_{2}\Tilde{\mathbf{y}}_2$ such that $\lambda_{2}=2w_2+2\Tilde{w}_u$

\begin{align*}
    (w_1-2w_2)\Tilde{y}_{1,2}-w_1\Tilde{y}_{2,2}-\Tilde{w}_u(\Tilde{y}_{3,2}+\Tilde{y}_{4,2})&=0\\
    -w_1\Tilde{y}_{1,2}+(w_1-2w_2)\Tilde{y}_{2,2}-\Tilde{w}_u(\Tilde{y}_{3,2}+\Tilde{y}_{4,2})&=0\\
    -\Tilde{w}_u(\Tilde{y}_{1,2}+\Tilde{y}_{2,2})-w_2\Tilde{y}_{3,2}-w_2\Tilde{y}_{4,2}=0\\
    -\Tilde{w}_u(\Tilde{y}_{1,2}+\Tilde{y}_{2,2})-w_2\Tilde{y}_{3,2}-w_2\Tilde{y}_{4,2}=0\\
\end{align*}
Summing up the last two expressions yields
\begin{align*}
    -2\Tilde{w}_u(\Tilde{y}_{1,2}+\Tilde{y}_{2,2})=2w_2(\Tilde{y}_{3,2}+\Tilde{y}_{4,2})
\end{align*}
 Then, the second expression is subtracted from the first one as:
\begin{align*}
    2w_1\Tilde{y}_{1,2}-2w_2\Tilde{y}_{1,2}+2w_2\Tilde{y}_{2,2}-2w_1\Tilde{y}_{2,2}&=0\\
    2w_1(\Tilde{y}_{1,2}-\Tilde{y}_{2,2})-2w_2(\Tilde{y}_{1,2}-\Tilde{y}_{2,2})&=0\\
    2(w_1-w_2)(\Tilde{y}_{1,2}-\Tilde{y}_{2,2})&=0\\
\end{align*}
Based on the assumption that each block has is concentrated around a different coefficient, where $w_2\neq w_1$, $\Tilde{y}_{1,2}$ and $\Tilde{y}_{2,2}$ must be equal.  Thus, the associated eigenvector can be written as:
\begin{align*}
\Tilde{\mathbf{y}}_2=
    \begin{bmatrix}
    \vspace{1mm}
    \quad\Tilde{y}_{1,2}\\\vspace{1mm}
    \quad\Tilde{y}_{1,2}\\\vspace{1mm}
    \quad\Tilde{y}_{3,2}\\\vspace{1mm}
    \quad\Tilde{y}_{4,2}\\
    \end{bmatrix}
\end{align*}

 \item $\Tilde{\mathbf{L}}\Tilde{\mathbf{y}}_3=\Tilde{\lambda}_3\Tilde{\mathbf{y}}_3$ such that $\lambda_{3}=2w_1+2\Tilde{w}_u$
\begin{align*}
    -w_1\Tilde{y}_{1,3}-w_1\Tilde{y}_{2,3}-\Tilde{w}_u(\Tilde{y}_{3,3}+\Tilde{y}_{4,3})&=0\\
    -w_1\Tilde{y}_{1,3}-w_1\Tilde{y}_{2,3}-\Tilde{w}_u(\Tilde{y}_{3,3}+\Tilde{y}_{4,3})&=0\\
    -\Tilde{w}_u(\Tilde{y}_{1,3}+\Tilde{y}_{2,3})+(w_2-2w_1)\Tilde{y}_{3,3}-w_2\Tilde{y}_{4,3}&=0\\
    -\Tilde{w}_u(\Tilde{y}_{1,3}+\Tilde{y}_{2,3})-w_2\Tilde{y}_{3,3}+(w_2-2w_1)\Tilde{y}_{4,3}&=0\\
\end{align*}
Summing up the first two expressions yields
\begin{align*}
    -2w_1(\Tilde{y}_{1,3}+\Tilde{y}_{2,3})=2\Tilde{w}_u(\Tilde{y}_{3,3}+\Tilde{y}_{4,3})
\end{align*}
 Then, the fourth expression is subtracted from the third one as follows
\begin{align*}
    2w_2\Tilde{y}_{3,3}-2w_1\Tilde{y}_{3,3}+2w_1\Tilde{y}_{4,3}-2w_2\Tilde{y}_{4,3}&=0\\
    2w_2(\Tilde{y}_{3,3}-\Tilde{y}_{4,3})-2w_1(\Tilde{y}_{3,3}-\Tilde{y}_{4,3})&=0\\
    2(w_2-w_1)(\Tilde{y}_{3,3}-\Tilde{y}_{4,3})&=0.\\
\end{align*}
Based on the assumption that each block has concentrated around a different coefficient, where $w_2\neq w_1$, $\Tilde{y}_{3,3}$ and $\Tilde{y}_{4,3}$ must be equal. Therefore, the associated eigenvector can be written as:
\begin{align*}
\Tilde{\mathbf{y}}_3=
    \begin{bmatrix}
    \vspace{1mm}
    \quad\Tilde{y}_{1,3}\\\vspace{1mm}
    \quad\Tilde{y}_{2,3}\\\vspace{1mm}
    \quad\Tilde{y}_{3,3}\\\vspace{1mm}
    \quad\Tilde{y}_{3,3}\\
    \end{bmatrix}
\end{align*}
\end{itemize}

To summarize, the eigenvectors associated with the eigenvalues ${\Tilde{\lambda}_0=0}$, ${\Tilde{\lambda}_1=4\Tilde{w}_u}$, ${\Tilde{\lambda}_2=2w_2+2\Tilde{w}_u}$ and ${\Tilde{\lambda}_3=2w_1+2\Tilde{w}_u}$ are
\begin{align*}
\Tilde{\mathbf{y}}_0=
\begin{bmatrix}
\vspace{1mm}
\Tilde{y}_{1,0}\\\vspace{1mm}
\Tilde{y}_{1,0}\\\vspace{1mm}
\Tilde{y}_{1,0}\\\vspace{1mm}
\Tilde{y}_{1,0} \\
\end{bmatrix}
\Tilde{\mathbf{y}}_1=
\begin{bmatrix}
\vspace{1mm}
\quad\Tilde{y}_{1,1}\\\vspace{1mm}
\quad\Tilde{y}_{1,1}\\\vspace{1mm}
-\Tilde{y}_{1,1}\\\vspace{1mm}
-\Tilde{y}_{1,1} \\
\end{bmatrix}
\Tilde{\mathbf{y}}_2=
\begin{bmatrix}
\vspace{1mm}
\quad\Tilde{y}_{1,2}\\\vspace{1mm}
\quad\Tilde{y}_{1,2}\\\vspace{1mm}
\quad\Tilde{y}_{3,2}\\\vspace{1mm}
\quad\Tilde{y}_{4,2}\\
\end{bmatrix}
\Tilde{\mathbf{y}}_3=
\begin{bmatrix}
\vspace{1mm}
\quad\Tilde{y}_{1,3}\\\vspace{1mm}
\quad\Tilde{y}_{2,3}\\\vspace{1mm}
\quad\Tilde{y}_{3,3}\\\vspace{1mm}
\quad\Tilde{y}_{3,3}\\
\end{bmatrix}
\end{align*}
Again, to compute the eigenvectors, the orthogonality information is analyzed for couples of eigenvectors as
\begin{itemize}
    \item $\Tilde{\mathbf{y}}_{1}^{\top}\Tilde{\mathbf{y}}_{3}=0$
    \begin{align*}
    \Tilde{y}_{1,1}(\Tilde{y}_{1,3}+\Tilde{y}_{2,3})&=2\Tilde{y}_{1,1}\Tilde{y}_{3,3}\\
    \Tilde{y}_{1,3}+\Tilde{y}_{2,3}&=2\Tilde{y}_{3,3}
\end{align*}
Based on the available knowledge that $\Tilde{y}_{1,3}+\Tilde{y}_{2,3}=\frac{-2\Tilde{w}_u}{w_1}\Tilde{y}_{3,3}$ and $w_1>\Tilde{w}_u>0$, the mappings yield $\Tilde{y}_{3,3}=0$ and $\Tilde{y}_{1,3}=-\Tilde{y}_{2,3}$.
 
\item $\Tilde{\mathbf{y}}_{1}^{\top}\Tilde{\mathbf{y}}_{2}=0$
\begin{align*}
    2\Tilde{y}_{1,1}\Tilde{y}_{1,2}&=\Tilde{y}_{1,1}(\Tilde{y}_{3,2}+\Tilde{y}_{4,2})\\
    2\Tilde{y}_{1,2}&=\Tilde{y}_{3,2}+\Tilde{y}_{4,2}
\end{align*}
Also knowing that $\Tilde{y}_{3,2}+\Tilde{y}_{4,2}=\frac{-2\Tilde{w}_u}{w_2}\Tilde{y}_{1,2}$ and $w_2>\Tilde{w}_u>0$, the mappings yield $\Tilde{y}_{3,2}=-\Tilde{y}_{4,2}$ and $\Tilde{y}_{1,2}=0$.
\end{itemize}
Consequently, the eigenvectors associated with the eigenvalues in ascending order ${\Tilde{\lambda}_0=0}$, ${\Tilde{\lambda}_1=4\Tilde{w}_u}$, ${\Tilde{\lambda}_2=2w_2+2\Tilde{w}_u}$ and ${\Tilde{\lambda}_3=2w_1+2\Tilde{w}_u}$ are
\begin{align*}
\Tilde{\mathbf{y}}_0=
\begin{bmatrix}
\vspace{1mm}
\Tilde{y}_{1,0}\\\vspace{1mm}
\Tilde{y}_{1,0}\\\vspace{1mm}
\Tilde{y}_{1,0}\\\vspace{1mm}
\Tilde{y}_{1,0} \\
\end{bmatrix}
\hspace{5mm}\Tilde{\mathbf{y}}_1=
\begin{bmatrix}
\vspace{1mm}
\quad\Tilde{y}_{1,1}\\\vspace{1mm}
\quad\Tilde{y}_{1,1}\\\vspace{1mm}
-\Tilde{y}_{1,1}\\\vspace{1mm}
-\Tilde{y}_{1,1} \\
\end{bmatrix}
\hspace{5mm}\Tilde{\mathbf{y}}_2=
\begin{bmatrix}
\vspace{1mm}
\quad0\\\vspace{1mm}
\quad0\\\vspace{1mm}
\quad\Tilde{y}_{3,2}\\\vspace{1mm}
-\Tilde{y}_{3,2}\\
\end{bmatrix}
\hspace{5mm}\Tilde{\mathbf{y}}_3=
\begin{bmatrix}
\vspace{1mm}
\quad\Tilde{y}_{1,3}\\\vspace{1mm}
-\Tilde{y}_{1,3}\\\vspace{1mm}
\quad0\\\vspace{1mm}
\quad0\\
\end{bmatrix}
\end{align*}
which can straight-forwardly be written for the eigenvalues in ascending order ${\Tilde{\lambda}_0=0}$, ${\Tilde{\lambda}_1=4\Tilde{w}_u}$, ${\Tilde{\lambda}_2=2w_1+2\Tilde{w}_u}$ and ${\Tilde{\lambda}_3=2w_2+2\Tilde{w}_u}$ as:
\begin{align*}
\Tilde{\mathbf{y}}_0=
\begin{bmatrix}
\vspace{1mm}
\Tilde{y}_{1,0}\\\vspace{1mm}
\Tilde{y}_{1,0}\\\vspace{1mm}
\Tilde{y}_{1,0}\\\vspace{1mm}
\Tilde{y}_{1,0} \\
\end{bmatrix}
\hspace{5mm}\Tilde{\mathbf{y}}_1=
\begin{bmatrix}
\vspace{1mm}
\quad\Tilde{y}_{1,1}\\\vspace{1mm}
\quad\Tilde{y}_{1,1}\\\vspace{1mm}
-\Tilde{y}_{1,1}\\\vspace{1mm}
-\Tilde{y}_{1,1} \\
\end{bmatrix}
\hspace{5mm}\Tilde{\mathbf{y}}_2=
\begin{bmatrix}
\vspace{1mm}
\quad\Tilde{y}_{1,2}\\\vspace{1mm}
-\Tilde{y}_{1,2}\\\vspace{1mm}
\quad0\\\vspace{1mm}
\quad0\\
\end{bmatrix}
\hspace{5mm}\Tilde{\mathbf{y}}_3=
\begin{bmatrix}
\vspace{1mm}
\quad0\\\vspace{1mm}
\quad0\\\vspace{1mm}
\quad\Tilde{y}_{3,3}\\\vspace{1mm}
-\Tilde{y}_{3,3}\\
\end{bmatrix}
\end{align*}

\newpage
\setcounter{secnumdepth}{0}
\section{Appendix D~: Experimental Setup and Additional Results}
\subsection{D.1~Outlier Effects and Robustness}
\subsubsection{D.1.1~Experimental Setup}
\setlength{\extrarowheight}{1pt}
\vspace{-2.5mm}
\begin{table}[ht]
\begin{tabularx}{\linewidth}{p{3.5cm}P{6cm}P{2.5cm}P{5cm}}
\hline\hline
    Group Information & \multicolumn{2}{c}{Feature Space Parameters}& Number of samples\\
    \cline{2-3}
        &$\boldsymbol{\mu}$&$\vartheta$ &\\
    \midrule
    Cluster-1 ($c_1$)&$\boldsymbol{\mu}_{c_1}$=[5.50; 4.50; 2.00; 0.75; 2.50; 4.50] & $\vartheta_{c_1}=0.50$ & $N_{c_1}=50$ for $N_{\mathrm{out}}=0$\\
    Cluster-2 ($c_2$)&$\boldsymbol{\mu}_{c_2}$=[7.50; 1.00; 5.50; 2.50; 1.00; 1.50]  &$\vartheta_{c_2}=0.50$ & $N_{c_2}=50$ for $N_{\mathrm{out}}=0$ \\
    Cluster-3 ($c_3$)& $\boldsymbol{\mu}_{c_3}$=[8.50; 0.75; 6.00; 4.50; 1.50; 1.25]&$\vartheta_{c_3}=0.50$  & $N_{c_3}=50$ for $N_{\mathrm{out}}=0$\\
    Type II Outliers ($o_2$)&$\boldsymbol{\mu}_{o_2}$=[7.00; 0.25; 5.00; 2.00; 0.50; 0.75] & $\vartheta_{o_2}=1.50$ & $N_{\mathrm{out}} \in \{ 0,5,10,\dots,20\}$\\
    \hline\hline
\end{tabularx}
\caption{Detailed numerical information for the synthetic data set}
\vspace{-3mm}
\end{table}
\begin{itemize}
    \item The original $i$th feature vector $\mathbf{x}_{i,k}$ associated with the $k$th cluster $c_k$, such that $k=1,\dots,K$ and $i=1,\dots,N_{c_k}$, is computed as\vspace{-3mm}
\begin{align*}
    \mathbf{x}_{i,k}=\boldsymbol{\mu}_{c_k}+\vartheta_{c_k}\boldsymbol{\upsilon},
\vspace{-1mm}
\end{align*}
where $\boldsymbol{\upsilon}$ denotes a vector 
of independently and identically distributed random variables
from a uniform distribution on the interval $[-0.5,0.5]$.

\item If $\mathbf{x}_{i,k}$ is replaced with a Type I outlier, the outlier  $\Tilde{\mathbf{x}}_i^{(1)}$ is computed as
\vspace{-0.5mm}
\begin{align*}
    \Tilde{\mathbf{x}}_i^{(1)}=\mathbf{x}_{i,k}+\vartheta_{o_1}\boldsymbol{\upsilon},
\vspace{-1.5mm}
\end{align*}
where $\vartheta_{o_1}\in\{0,1,2,\dots,10\}$.

\item If $\mathbf{x}_{i,k}$ is replaced with a Type II outlier, the outlier  $\Tilde{\mathbf{x}}_i^{(2)}$ is computed as
\begin{align*}
    \Tilde{\mathbf{x}}_i^{(2)}=\boldsymbol{\mu}_{o_2}+\vartheta_{o_2}\boldsymbol{\upsilon}.
\end{align*}
\end{itemize}
\vspace{-3mm}
\subsubsection{D.1.2~Additional Results}
\vspace{-3.5mm}
\begin{figure*}[!h]
  \centering
 \includegraphics[trim={1.5cm 12.4cm 0cm 1.5cm},clip,width=19.5cm]{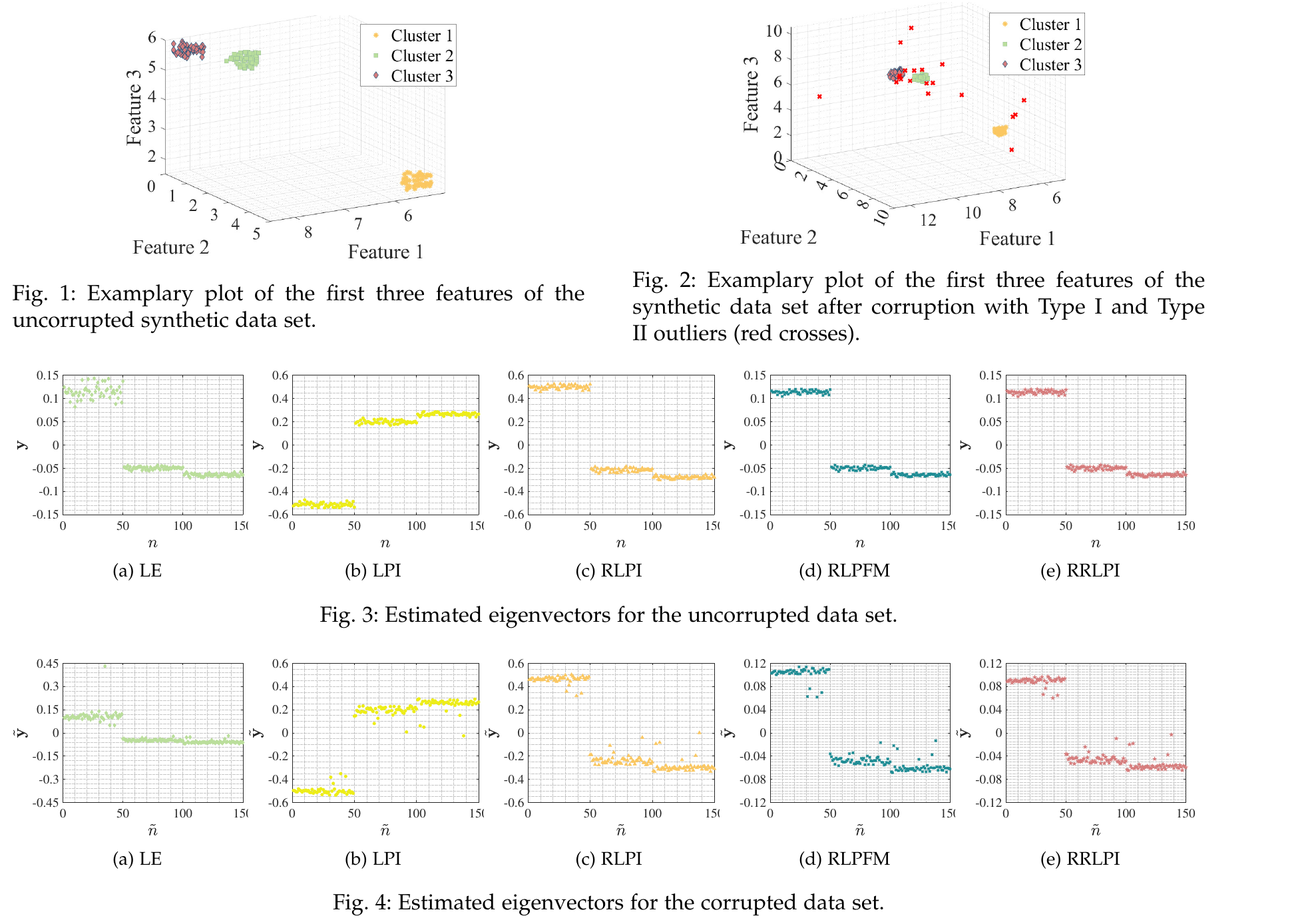}
 \vspace{-9.5mm}
\end{figure*}
\newpage
\renewcommand{\thefigure}{5}
\begin{figure*}[h!]
  \centering
 \captionsetup{justification=centering}
\includegraphics[trim={4.5cm 16.3cm 4.5cm 1.75cm},clip,width=13cm]{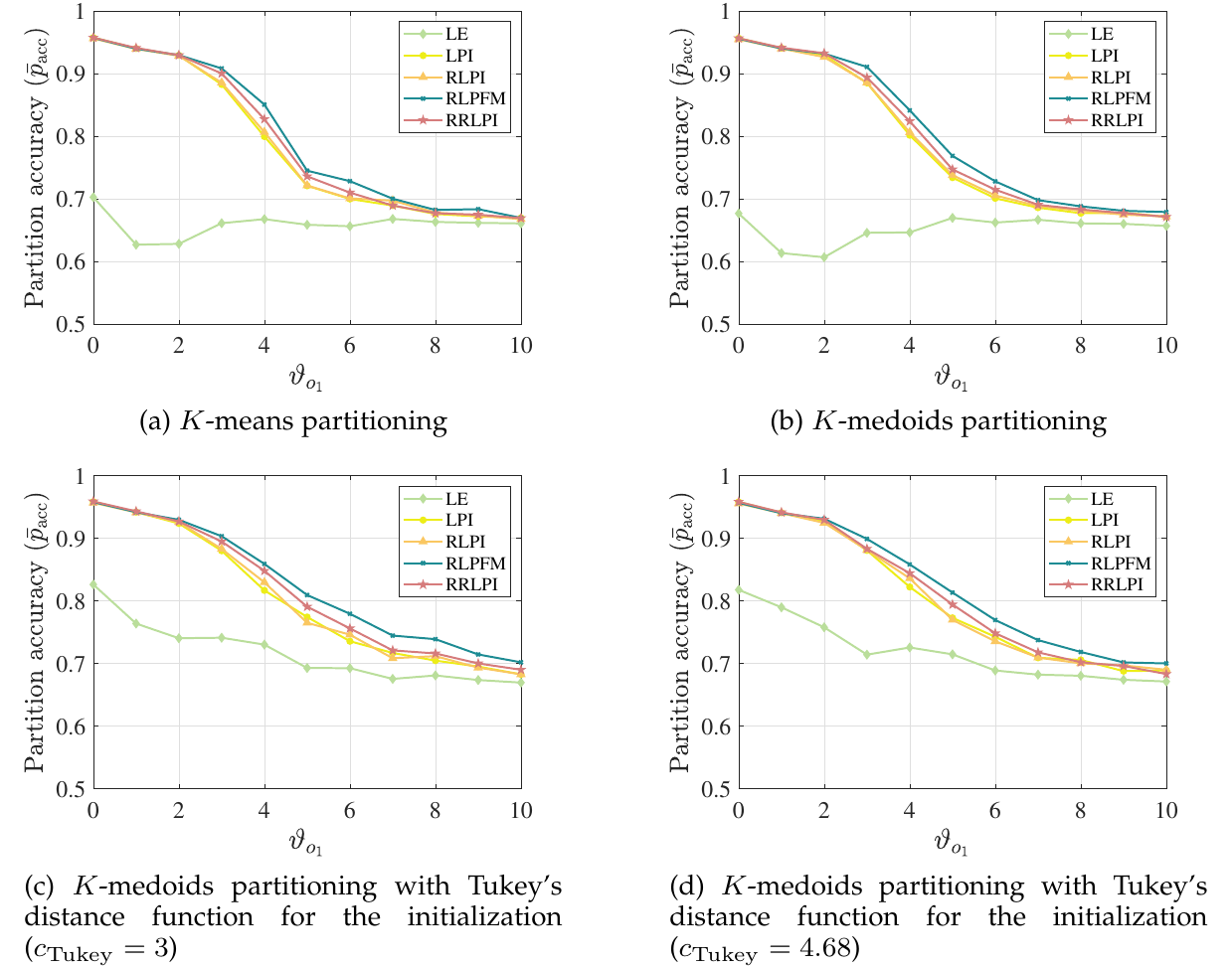}
 \caption{ $\bar{p}_{\mathrm{acc}}$ performance of different partitioning methods for increasing $\vartheta_{o_1}$
associated with Type I outlier\\($n=300$, $N_{\mathrm{out}}=10$, $\vartheta_{o_2}=1.5$, $\vartheta_{c_K}=0.5$ s.t. $K=1,\dots,K$.)}
\end{figure*}
\renewcommand{\thefigure}{6}
\begin{figure*}[h!]
 \centering
 \captionsetup{justification=centering}
 \includegraphics[trim={4.5cm 16.3cm 4.5cm 1.75cm},clip,width=13cm]{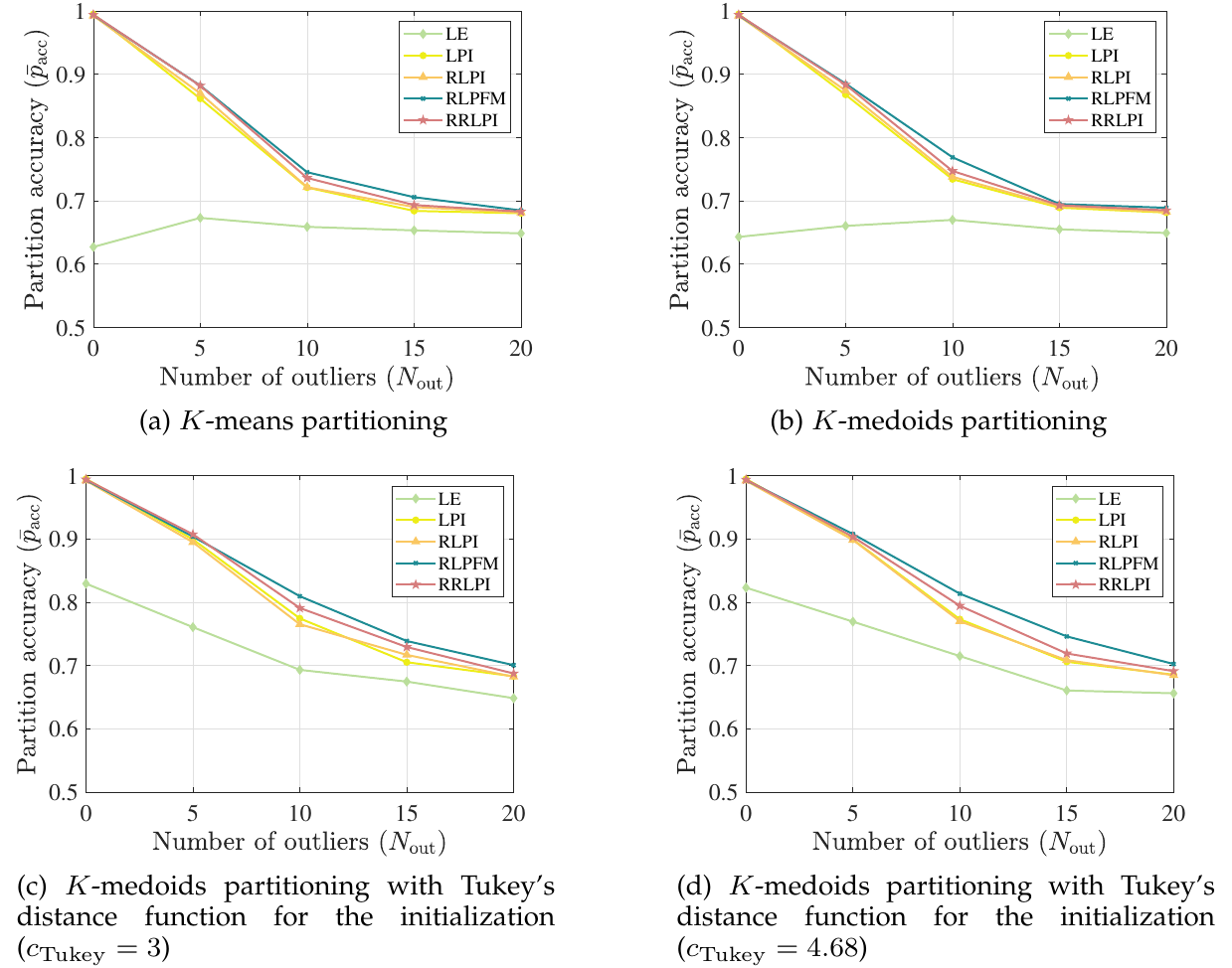}
 \caption{$\bar{p}_{\mathrm{acc}}$ performance of different partitioning methods for each of outlier type with increasing $N_{\mathrm{out}}$\\ ($n=300,\vartheta_{o_1}=5$,$\vartheta_{o_2}=1.5$, $\vartheta_{c_K}=0.5$ s.t. $K=1,\dots,K$.)}
\end{figure*}
\newpage
\subsection{D.2~Cluster Enumeration}
\subsubsection{D.2.1~Additional Results}
\vspace{-5mm}
\renewcommand{\thefigure}{7}
\begin{figure*}[h!]
  \centering
  \captionsetup{justification=centering}
\subfloat[ $\bar{p}_{\mathrm{det}}$ with respect to different penalty parameters for $K$-means partitioning]{\includegraphics[trim={0mm 0mm 0mm 0mm},clip,width=6cm]{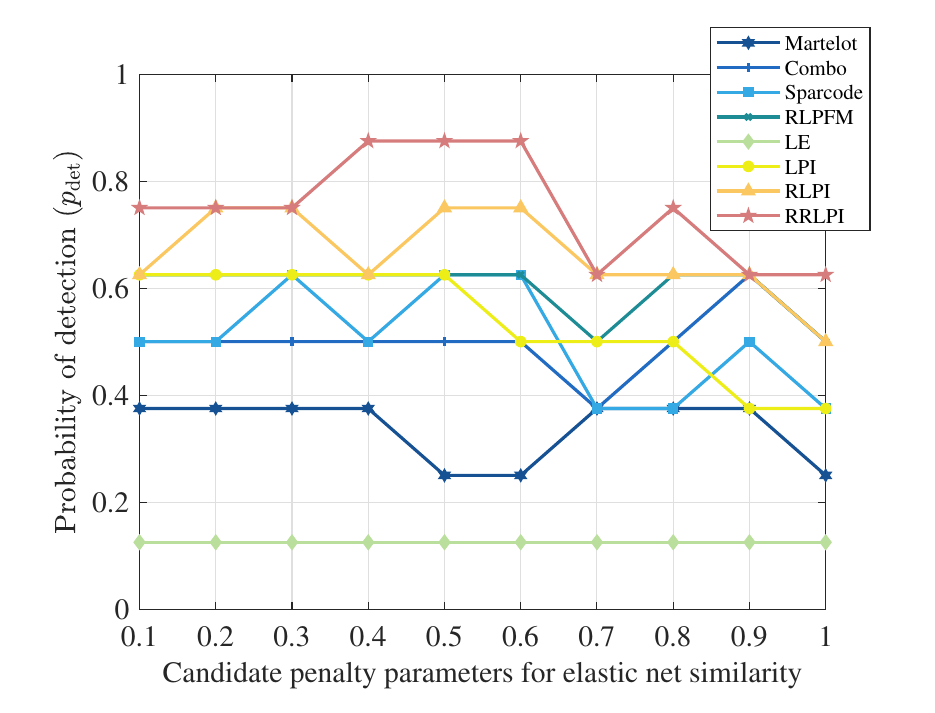}\label{fig:PdetPenaltyKmeans}}
\hspace{1.5cm}
\subfloat[Performance rank for $K$-means partitioning]{\includegraphics[trim={0mm 0mm 0mm 0mm},clip,width=6cm]{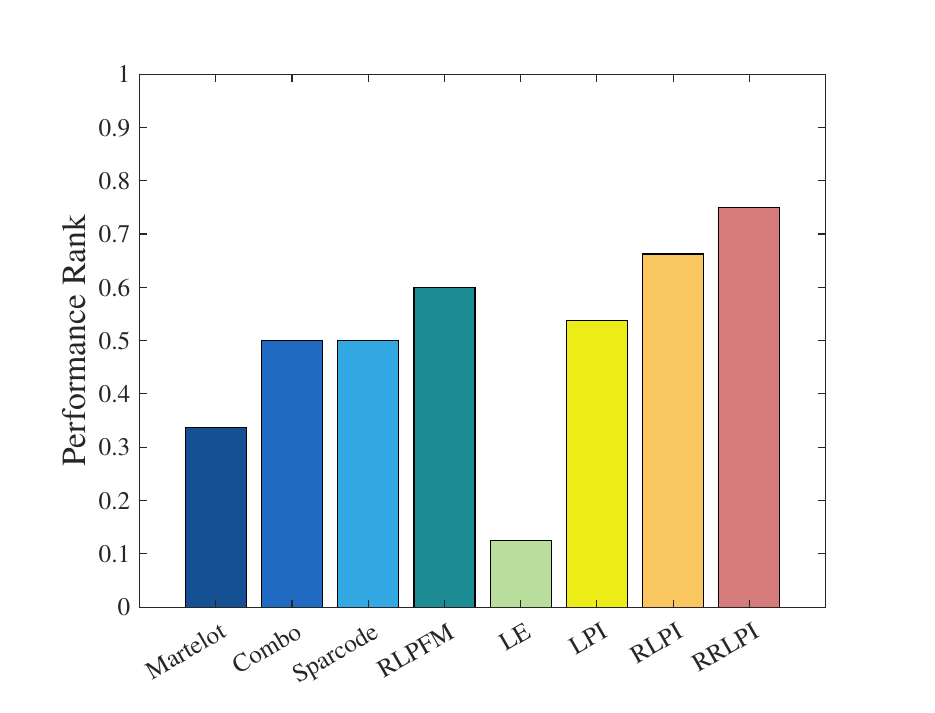}\label{fig:PdetPerRankKmeans}}\\\vspace{-4mm}
\subfloat[ $\bar{p}_{\mathrm{det}}$ with respect to different penalty parameters for $K$-medoids partitioning]{\includegraphics[trim={0mm 0mm 0mm 0mm},clip,width=6cm]{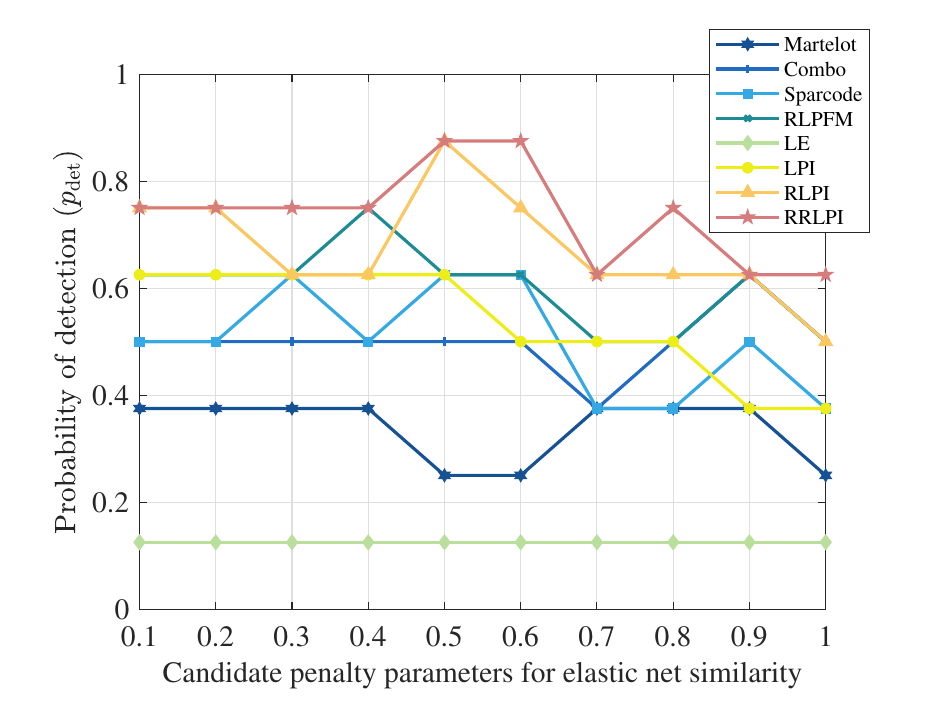}\label{fig:PdetPenaltyKmedoids}}
\hspace{1.5cm}
\subfloat[Performance rank for $K$-medoids partitioning]{\includegraphics[trim={0mm 0mm 0mm 0mm},clip,width=6cm]{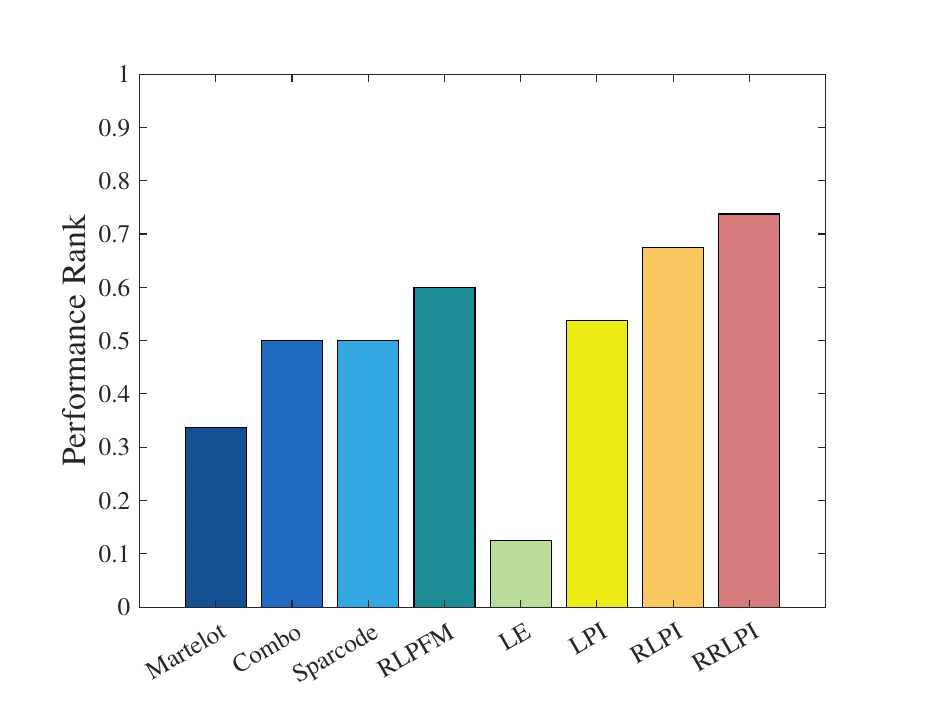}\label{fig:PdetPerRankKmedoids}}\\\vspace{-4mm}
\subfloat[ $\bar{p}_{\mathrm{det}}$ with respect to different penalty parameters for $K$-medoids partitioning with Tukey's distance function for the initialization ($c_{\mathrm{Tukey}}=3$) ]{\includegraphics[trim={0mm 0mm 0mm 0mm},clip,width=6cm]{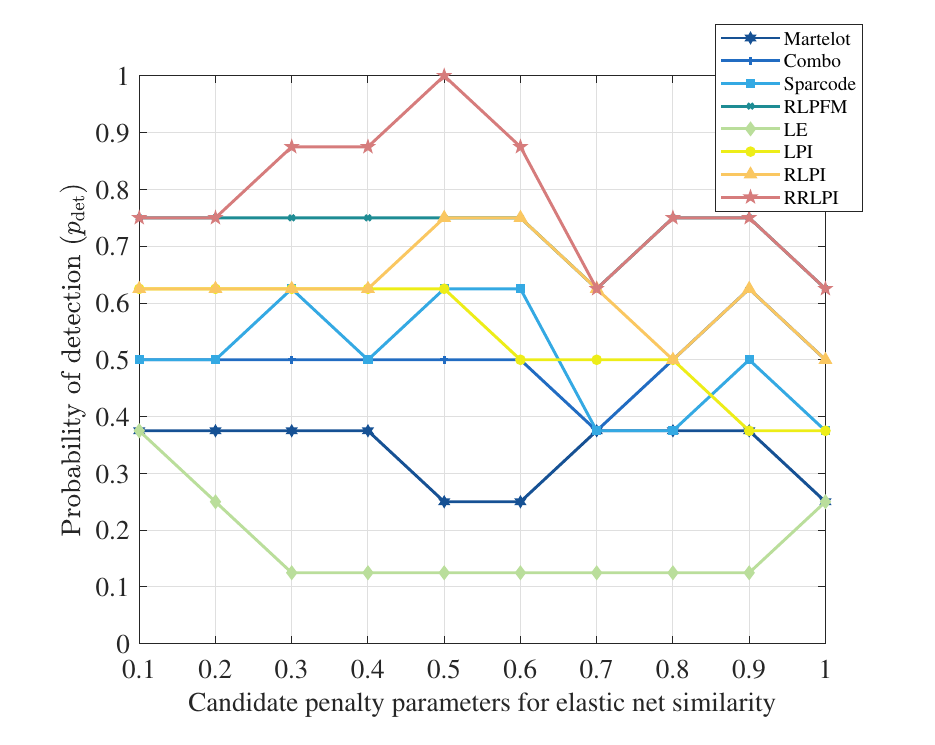}\label{fig:PdetPenaltyKmedoidsTukeyc3}}
\hspace{1.5cm}
\subfloat[Performance rank for $K$-medoids partitioning with Tukey's distance function for the initialization ($c_{\mathrm{Tukey}}=3$)]{\includegraphics[trim={0mm 0mm 0mm 0mm},clip,width=6cm]{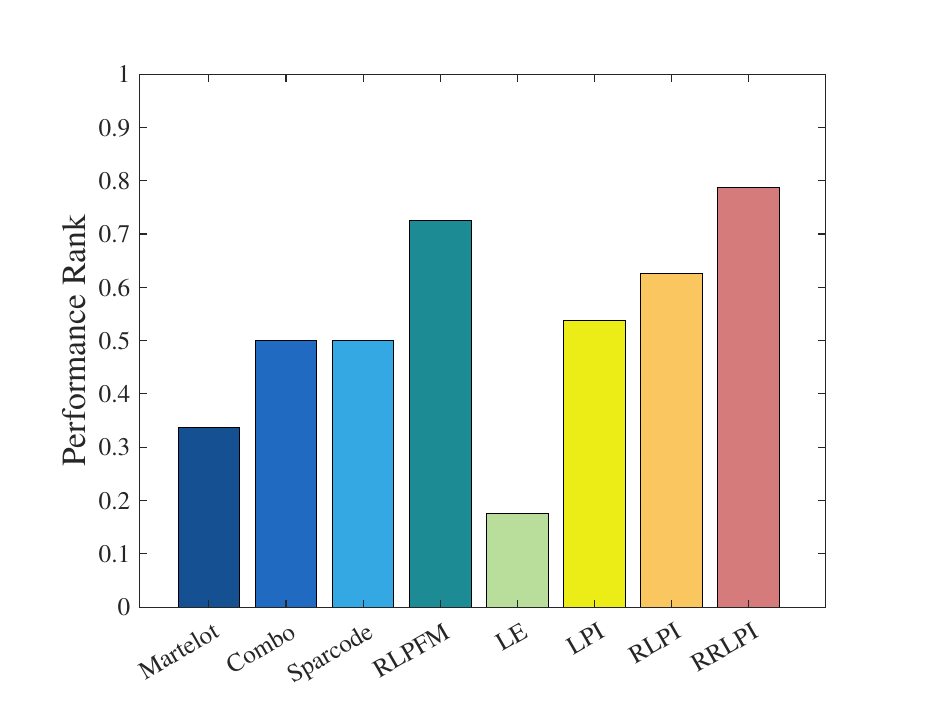}\label{fig:PdetPerRankKmedoidsTukeyc3}}\\\vspace{-4mm}
\subfloat[ $\bar{p}_{\mathrm{det}}$ with respect to different penalty parameters for $K$-medoids partitioning with Tukey's distance function for the initialization ($c_{\mathrm{Tukey}}=4.68$) ]{\includegraphics[trim={0mm 0mm 0mm 0mm},clip,width=6cm]{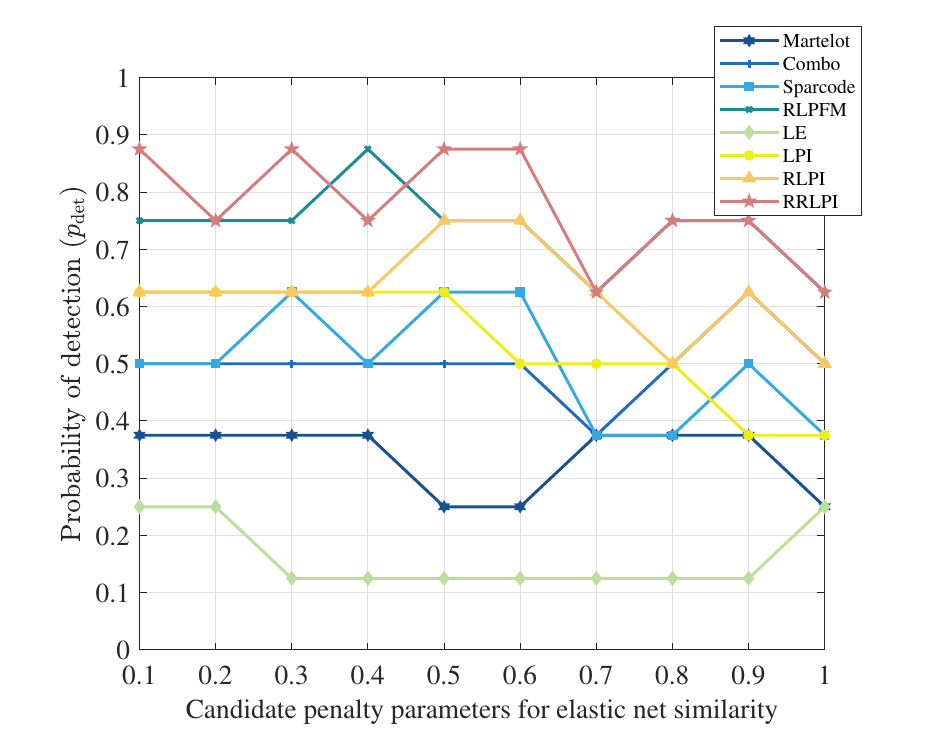}\label{fig:PdetPenaltyKmedoidsTukeyc4dot68}}
\hspace{1.5cm}
\subfloat[Performance rank for $K$-medoids partitioning with Tukey's distance function for the initialization ($c_{\mathrm{Tukey}}=4.68$)]{\includegraphics[trim={0mm 0mm 0mm 0mm},clip,width=6cm]{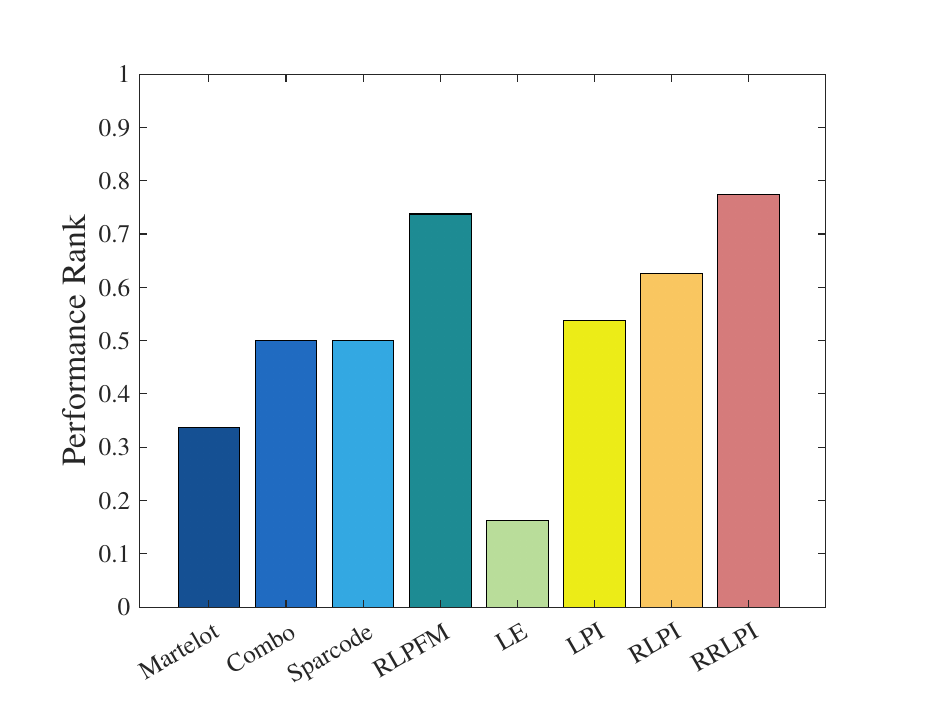}\label{fig:PdetPerRankKmedoidsTukeyc34dot68}}\\\vspace{-2mm}
\caption{Numerical results for cluster enumeration using different partitioning algorithms.}
  \label{fig:imagesegmentationnumericalresults}
\end{figure*}
\newpage
\renewcommand{\thefigure}{8}
\newpage
\begin{figure*}[h!]
  \centering
 \includegraphics[trim={0cm 1cm 0cm 5mm},clip,width=8cm]{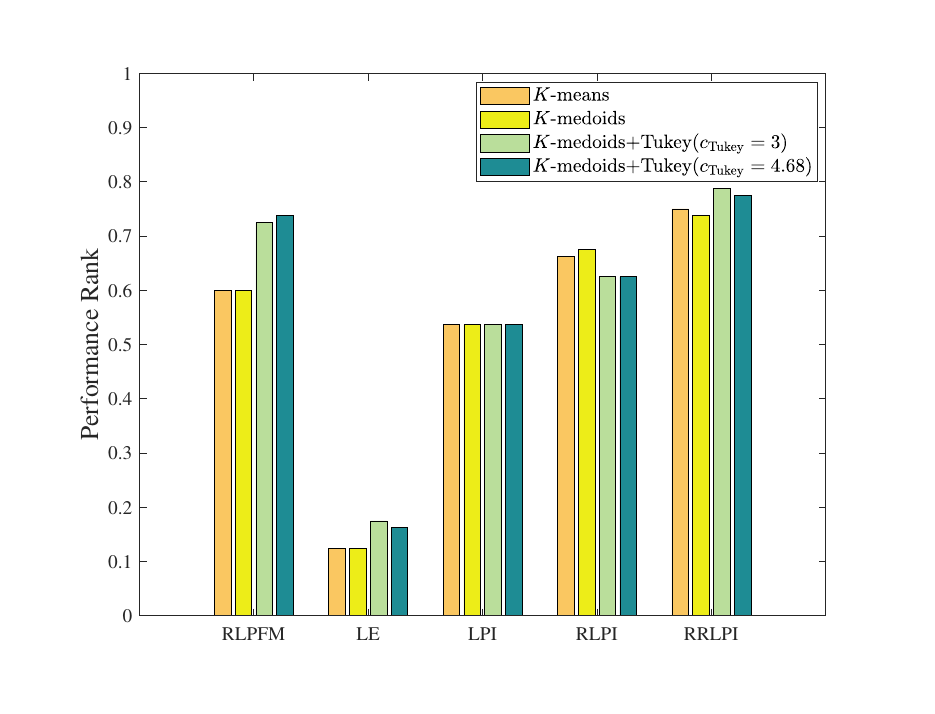}
    \caption{Overall performance rank for different partitioning algorithms.}
\end{figure*}
\vspace{-5mm}
\subsection{D.3~Image Segmentation}
\subsubsection{D.3.1~Experimental Setup}
\vspace{-5mm}
\begin{table}[h!]
 \centering
 \includegraphics[trim={1.5cm 8cm 1cm 1.1cm},clip,width=17cm]{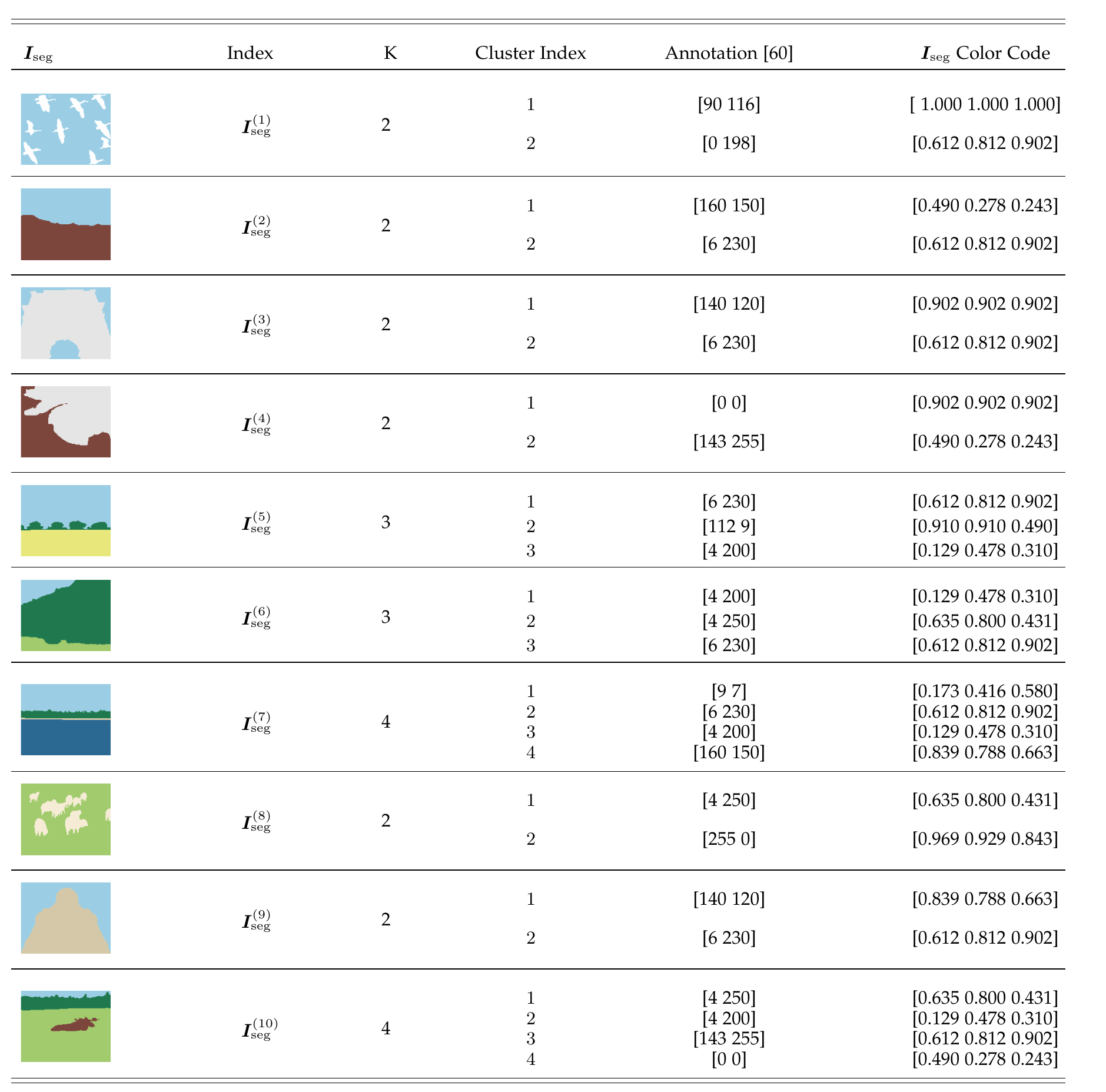}
\caption{Reference annotations for the ground truth images.}
\end{table}
\newpage
\subsubsection{D.3.2~Additional Results}
\begin{table}[h!]
  \centering
  \includegraphics[trim={1.5cm 4.5cm 1cm 1cm},clip,width=19cm]{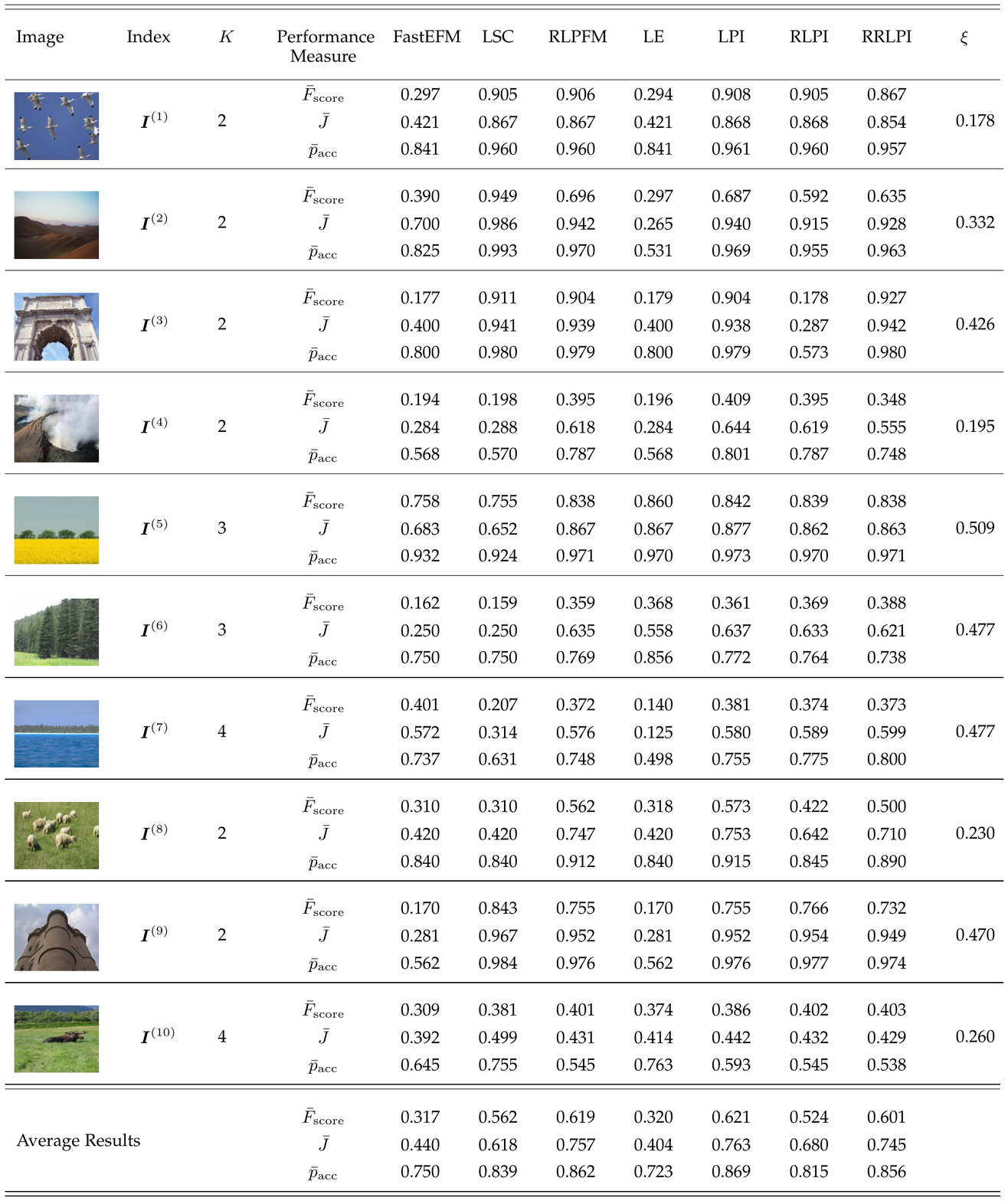}
 \caption{Detailed performance results for the original images}
\end{table}
\begin{table}[!tbp]
  \centering
  \includegraphics[trim={1.5cm 3cm 1cm 2.5cm},clip,width=19cm]{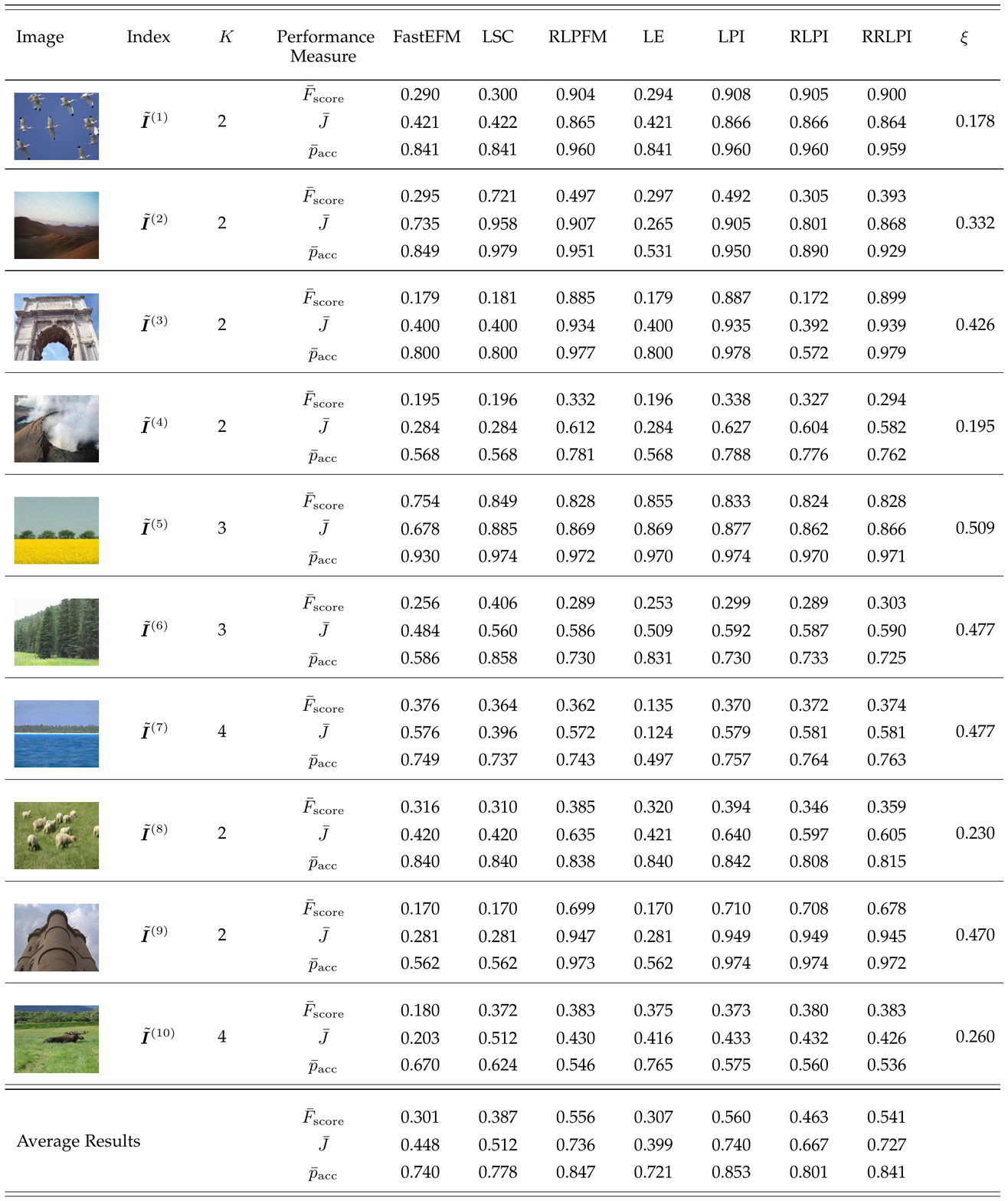}
 \caption{Detailed performance results for the corrupted images. ($\sigma^{(\xi)}=10^{-3}$)}
\end{table}

\renewcommand{\thefigure}{9}
\begin{figure*}[!tbp]
  \centering
  \captionsetup{justification=centering}
  \footnotesize
   \stackon[5pt]{\includegraphics[trim={2cm 10cm 0cm 10cm},clip,width=20cm]{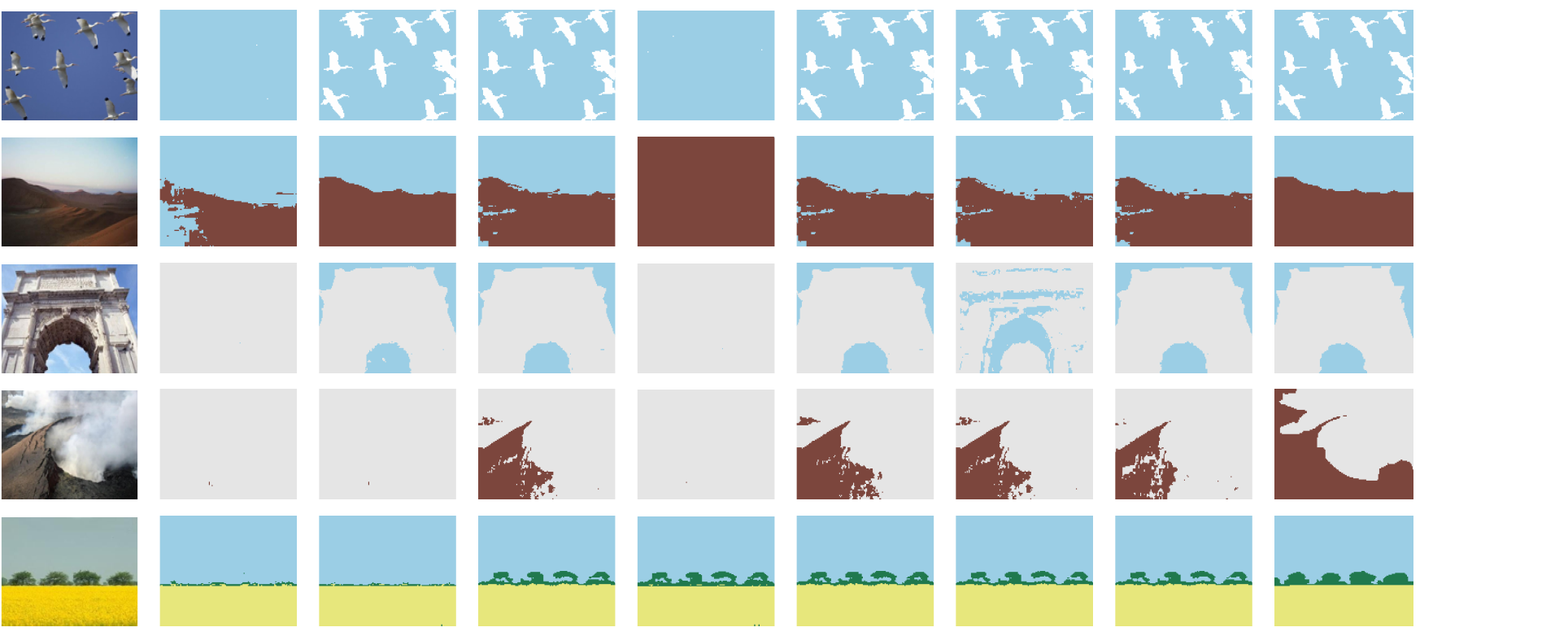}\label{fig:flowers}}{\hspace{-0.8cm}$\bm{I}^{(i)}$\hspace{1.2cm} FastEFM\hspace{1.2cm} LSC\hspace{1.2cm} RLPFM\hspace{1.3cm} LE\hspace{1.6cm} LPI\hspace{1.4cm} RLPI\hspace{1.1cm} RRLPI\hspace{1.3cm}$\bm{I}_{\mathrm{seg}}$\hspace{1.4cm}}\hspace{-3mm}
          \includegraphics[trim={2cm 10cm 0cm 10cm},clip,width=20cm]{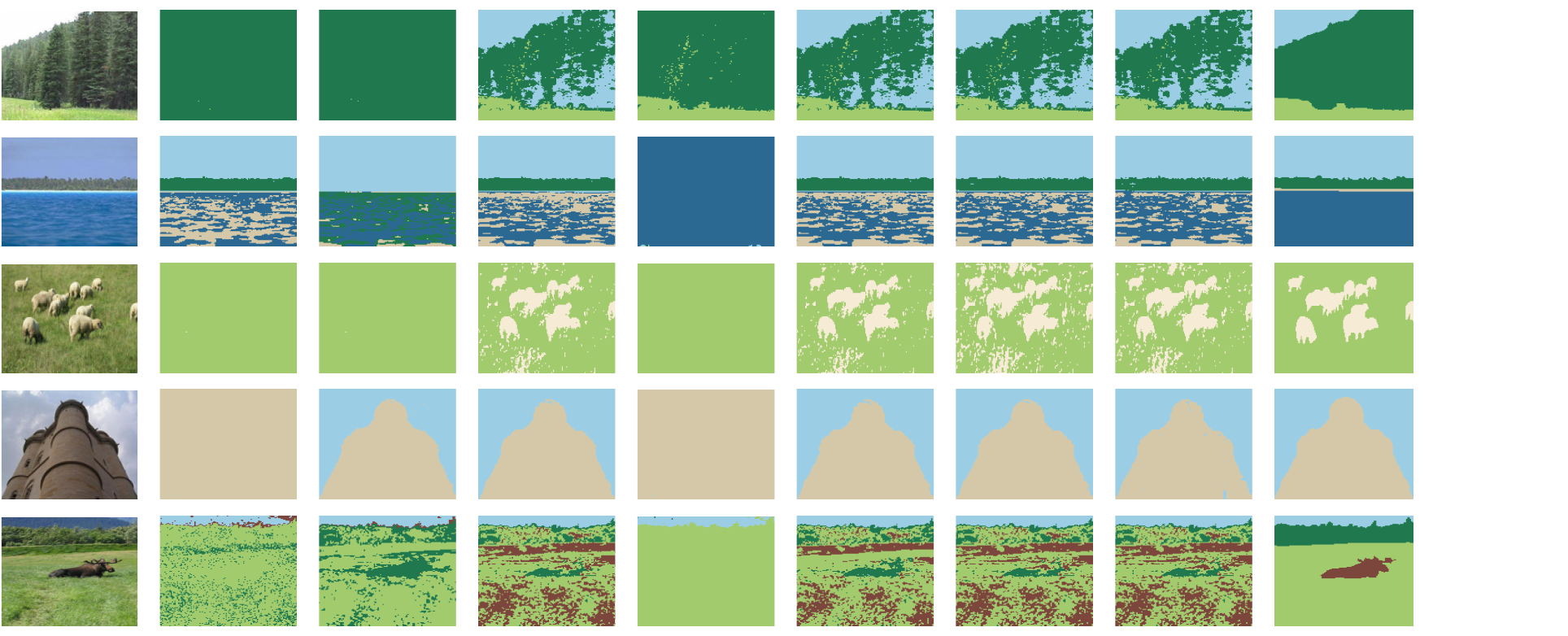}\\
  \caption{Image segmentation results for the original images}
  \label{fig:Images_Supplementaryegmentation}
  \vspace{-2mm}
\end{figure*}
\renewcommand{\thefigure}{10}
\begin{figure*}
    \centering
\hspace{-8mm}
\subfloat[$\bm{I}$]{\includegraphics[trim={1mm 5mm 8mm 5mm},clip,width=4.7cm]{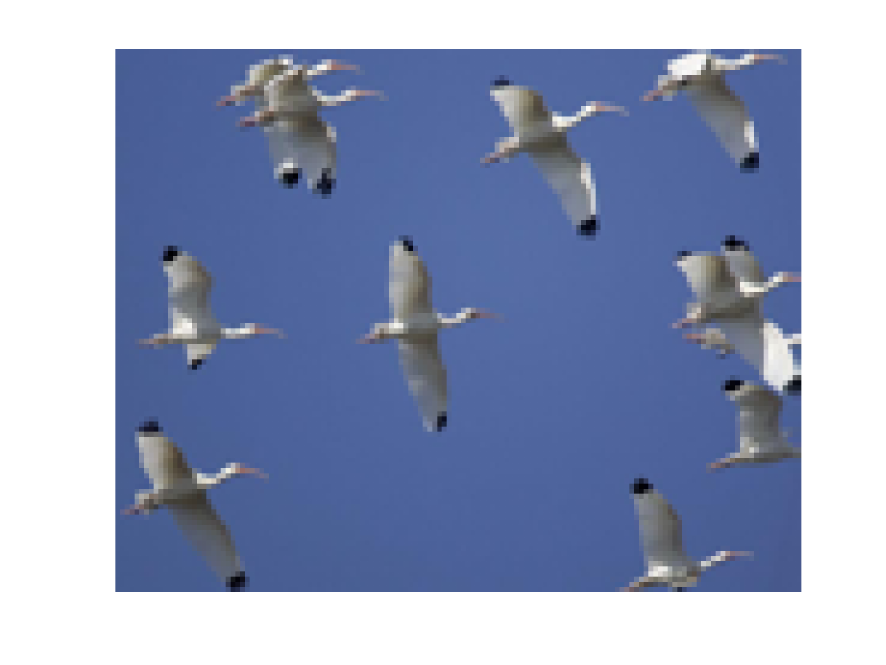}}
\subfloat[$\hat{\bm{I}}_{\mathrm{seg}}$ for LE]{\includegraphics[trim={1mm 5mm 8mm 5mm},clip,width=4.7cm]{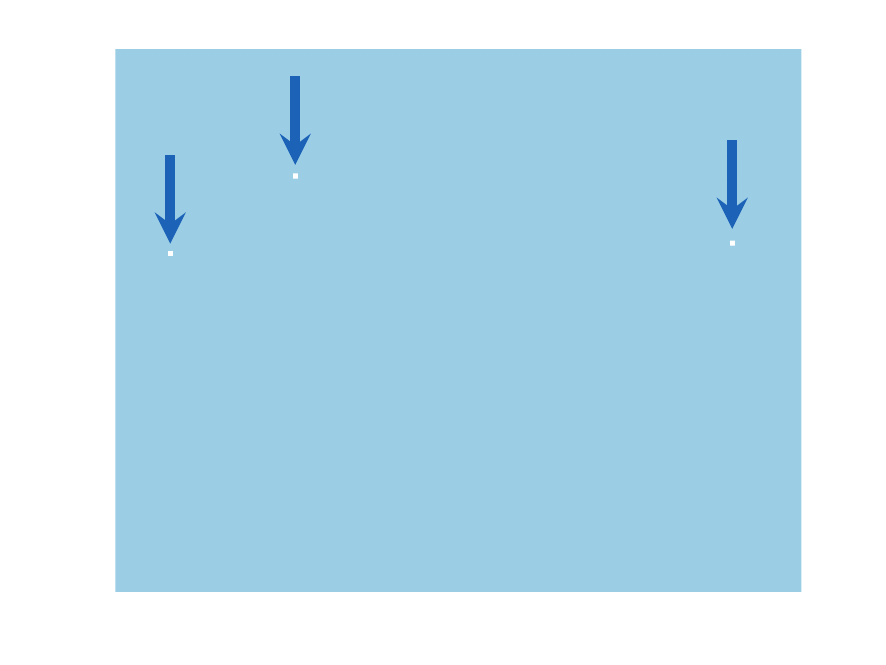}\label{fig:LESegresult}}
\subfloat[$\hat{\bm{I}}_{\mathrm{seg}}$ for RRLPI]{\includegraphics[trim={1mm 5mm 8mm 5mm},clip,width=4.7cm]{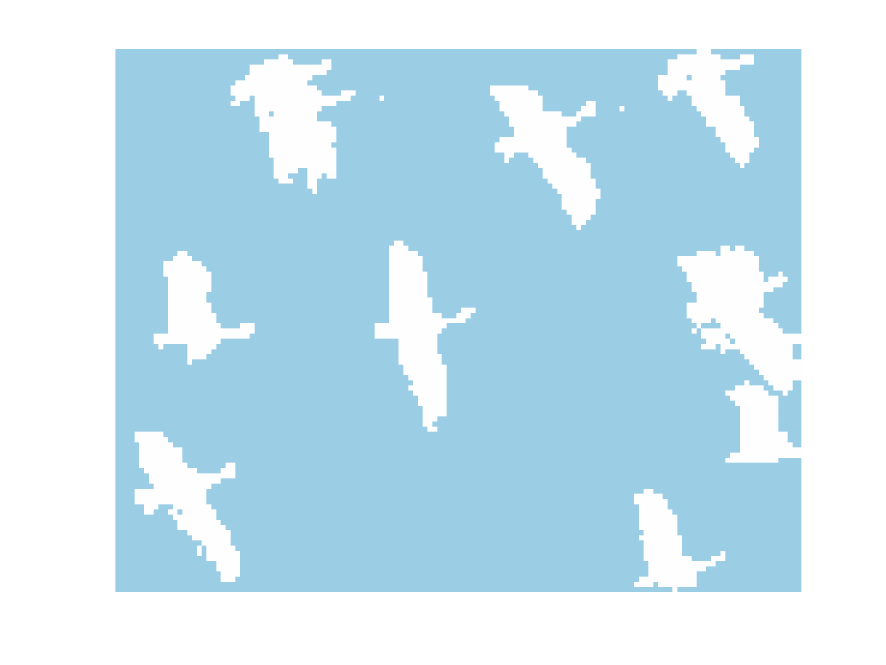}\label{fig:RRLPISegresult}}
\subfloat[Annotated image ${\bm{I}}_{\mathrm{seg}}$]{\includegraphics[trim={1mm 5mm 8mm 5mm},clip,width=4.7cm]{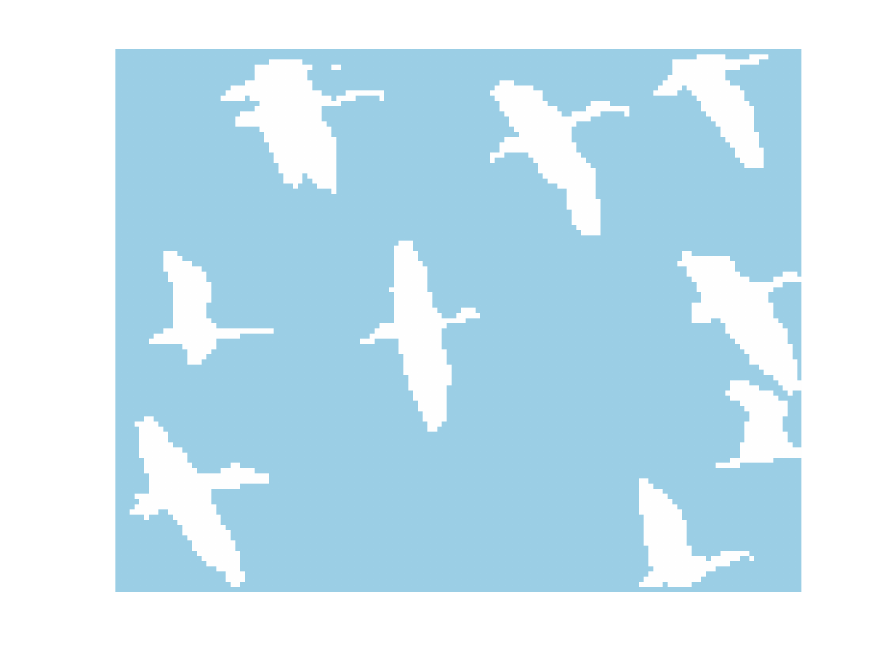}}
\caption{Example segmentations for LE and RRLPI methods. The embeddings that are mapped far away from the group of pixels are pointed out using arrows.}
\label{fig:outliereffectimage}
\end{figure*}
\renewcommand{\thefigure}{11}
\begin{figure*}[!tbp]
  \vspace{1.5cm}
  \centering
  \captionsetup{justification=centering}
  \footnotesize
   \stackon[5pt]{\includegraphics[trim={2cm 10cm 0cm 10cm},clip,width=20cm]{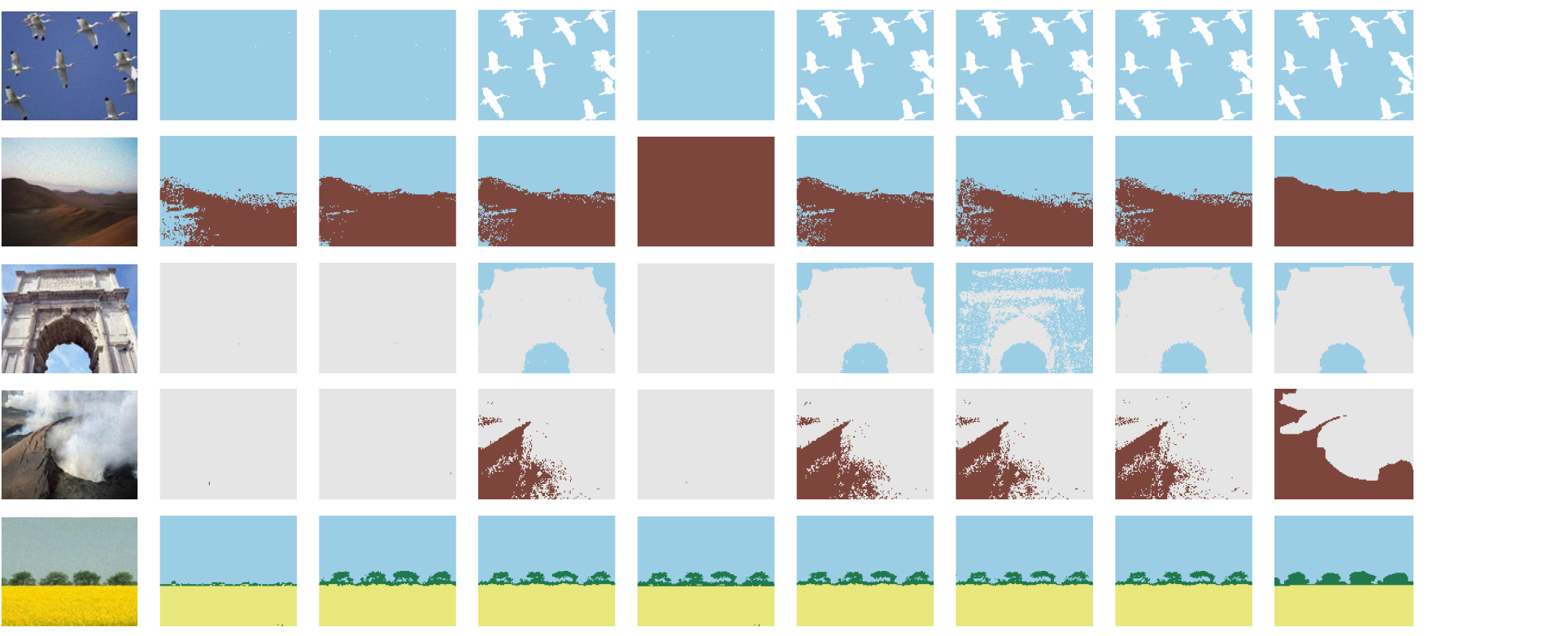}\label{fig:flowers}}{\hspace{-0.8cm}$\Tilde{\bm{I}}^{(i)}$\hspace{1.2cm} FastEFM\hspace{1.2cm} LSC\hspace{1.2cm} RLPFM\hspace{1.3cm} LE\hspace{1.6cm} LPI\hspace{1.4cm} RLPI\hspace{1.1cm} RRLPI\hspace{1.3cm}$\bm{I}_{\mathrm{seg}}$\hspace{1.4cm}}\hspace{-3mm}
          \includegraphics[trim={2cm 10cm 0cm 10cm},clip,width=20cm]{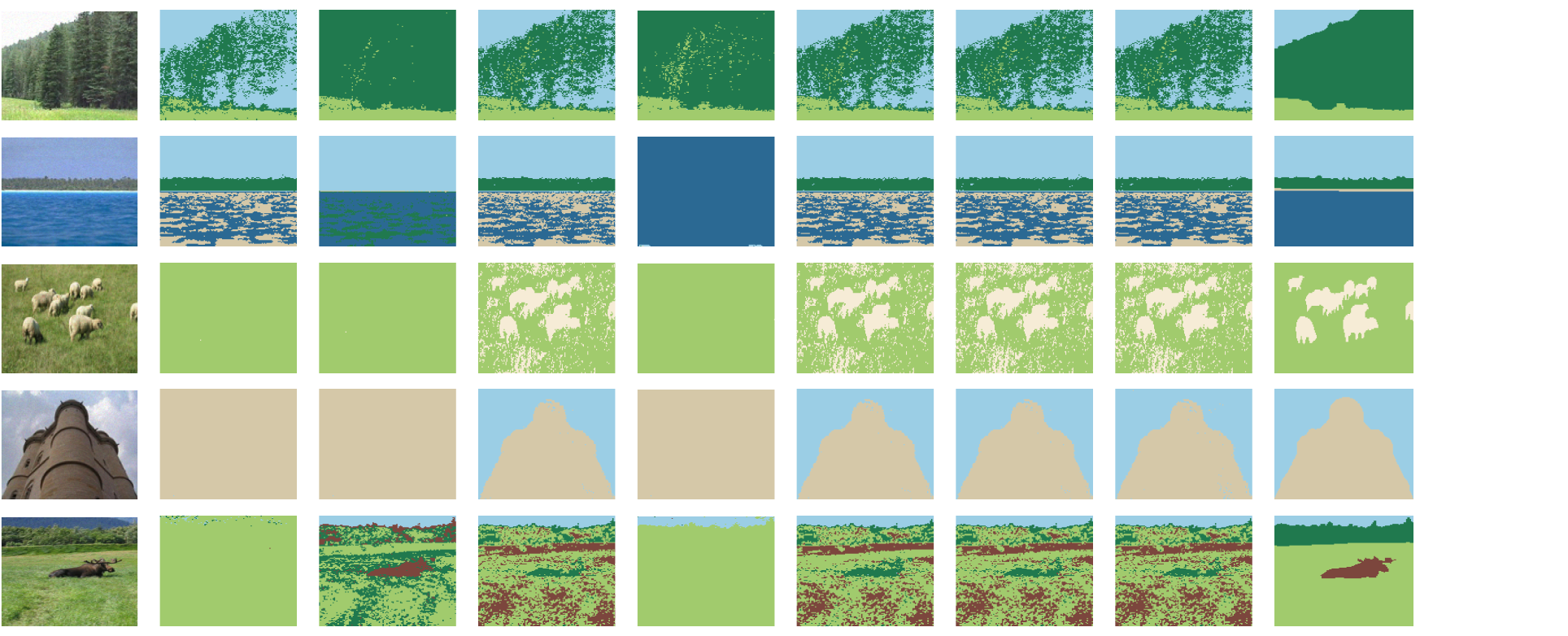}\\
  \caption{Image segmentation results for the corrupted images ($\sigma^{(\xi)}=10^{-3}$)}
  \label{fig:Images_Supplementaryegmentation}
  \vspace{-2mm}
\end{figure*}

%% file: packages.tex
\usepackage{amsmath,amssymb,amsfonts,bm,setspace}
\usepackage[ruled,lined,commentsnumbered]{algorithm2e}
\usepackage{bbm}
\usepackage{tikz}
\usepackage{stackengine}
\usetikzlibrary{backgrounds}
\usetikzlibrary{shapes.geometric, arrows, positioning}
\usepackage{cancel}
\usepackage{relsize}
\usepackage{xcolor}
\usepackage{multirow}
\usepackage{multicol}
\usepackage{comment}
\usepackage{tabularx,ragged2e,booktabs,caption}
\usepackage{lipsum}
\usepackage{subfiles} 
\usepackage{tikz}
\usetikzlibrary{shapes.geometric, arrows}
\usepackage{relsize}
\usepackage{xcolor}
\usepackage{amsthm}
\usepackage{amsmath}
\usepackage{lipsum, babel}
  \usepackage[caption=false,font=footnotesize]{subfig}